\documentclass{article}
\usepackage{ijcai26}

\usepackage{times}
\usepackage{soul}
\usepackage{url}
\usepackage[hidelinks]{hyperref}
\usepackage[utf8]{inputenc}
\usepackage[small]{caption}
\usepackage{graphicx}
\usepackage{amsmath}
\usepackage{amsthm}
\usepackage{booktabs}
\usepackage[switch]{lineno}

\usepackage{thm-restate}
\usepackage{natbib}
\usepackage{acronym}
\usepackage{amssymb}
\usepackage{environ}
\usepackage{tabularx}
\usepackage{todonotes}

\usepackage{tikz}
\usetikzlibrary{calc}
\usepackage{pifont}%
\usepackage[ruled, linesnumbered, noend]{algorithm2e} 
\usepackage{algpseudocode}

\usepackage{cleveref}

\definecolor{CoalitionColor}{HTML}{1D42A6}

\usetikzlibrary{decorations.pathreplacing}
\tikzstyle{agent}=[draw, shape=circle, minimum size=30pt]
\tikzstyle{smallagent}=[draw, shape=circle, minimum size=15pt]
\tikzstyle{aux}=[draw, shape=circle, minimum size=5pt]
\newcommand{\convexpath}[2]{
  [   
  create hullcoords/.code={
    \global\edef\namelist{#1}
    \foreach [count=\counter] \nodename in \namelist {
      \global\edef\numberofnodes{\counter}
      \coordinate (hullcoord\counter) at (\nodename);
    }
    \coordinate (hullcoord0) at (hullcoord\numberofnodes);
    \pgfmathtruncatemacro\lastnumber{\numberofnodes+1}
    \coordinate (hullcoord\lastnumber) at (hullcoord1);
  },
  create hullcoords
  ]
  ($(hullcoord1)!#2!-90:(hullcoord0)$)
  \foreach [
  evaluate=\currentnode as \previousnode using \currentnode-1,
  evaluate=\currentnode as \nextnode using \currentnode+1
  ] \currentnode in {1,...,\numberofnodes} {
    let \p1 = ($(hullcoord\currentnode) - (hullcoord\previousnode)$),
    \n1 = {atan2(\y1,\x1) + 90},
    \p2 = ($(hullcoord\nextnode) - (hullcoord\currentnode)$),
    \n2 = {atan2(\y2,\x2) + 90},
    \n{delta} = {Mod(\n2-\n1,360) - 360}
    in 
    {arc [start angle=\n1, delta angle=\n{delta}, radius=#2]}
    -- ($(hullcoord\nextnode)!#2!-90:(hullcoord\currentnode)$) 
  }
}

\newtheorem{example}{Example}
\newtheorem{theorem}{Theorem}
\newtheorem{lemma}{Lemma} 
\newtheorem{corollary}{Corollary}

\makeatletter
\newcommand{\problemtitle}[1]{\gdef\@problemtitle{#1}}
\newcommand{\probleminput}[1]{\gdef\@probleminput{#1}}
\newcommand{\problemsolution}[1]{\gdef\@problemsolution{#1}}
\NewEnviron{problem}{%
	\problemtitle{}\probleminput{}\problemsolution{}
	\BODY
	\par\addvspace{.5\baselineskip}
	\noindent
	\begin{tabularx}{\linewidth}{@{\hspace{\parindent}} l X c}
		\toprule
		\multicolumn{2}{@{\hspace{\parindent}}l}{\@problemtitle} \\
		\textbf{Input:} & \@probleminput \\
		\textbf{Solution:} & \@problemsolution \\
		\bottomrule
	\end{tabularx}
	\par\addvspace{.5\baselineskip}
}
\makeatother

\acrodef{ASHG}{additively separable hedonic game}
\acrodef{FHG}{fractional hedonic game}
\acrodef{MFHG}{modified fractional hedonic game}
\acrodef{FEG}{friends and enemies game}
\acrodef{NS}{Nash stable}
\acrodef{IS}{individually stable}
\acrodef{CNS}{contractually Nash stable}
\acrodef{CIS}{contractually individually stable}
\acrodefindefinite{ASHG}{an}{an}

\newcommand{\vf}{v}
\newcommand{\uf}{u}

\title{Stability Under Valuation Updates in Coalition Formation}

\author{
Fabian Frank
\and
Matija Novakovi\'c\and
Ren\'e Romen
\affiliations
Technical University of Munich\\
\emails
\{fabian\_w.frank, matija.novakovic, rene.romen\}@tum.de 
}

\begin{document}

\maketitle

\begin{abstract}
Coalition formation studies how to partition a set of agents into disjoint coalitions under consideration of their preferences. We study the classical objective of stability in a variant of additively separable hedonic games where agents can change their valuations. Our objective is to find a stable partition after each change. To minimize the reconfiguration cost, we search for nearby stable coalition structures. Our focus is on stability concepts based on single-agent deviations. We present a detailed picture of the complexity of finding nearby stable coalition structures in additively separable hedonic games, for both symmetric and non-symmetric valuations. Our results show that the problem is NP-complete for Nash stability, individual stability, contractual Nash stability, and contractual individual stability. We complement these results by presenting polynomial-time algorithms for contractual Nash stability and contractual individual stability under restricted symmetric valuations. Finally, we show that these algorithms guarantee a bounded average distance over long sequences of updates.
\end{abstract}
\section{Introduction}
Many real-life applications, such as the formation of study groups, the interaction of robots, or international organizations, are multi-agent systems in which agents have different preferences for the agents they interact with. Coalition formation research is concerned with modeling cooperation in such systems.
The objective is to improve the outcome of the agents through coordinated actions. In this paper, we assume selfish agents that aim to maximize their own outcomes, in contrast to cooperative agents that jointly aim to maximize a global objective. 
Hedonic games \citep{DrGr80a} provide a general framework for modeling this setting.
In these games, agents' preferences are given by complete rankings over their potential coalitions, i.e., all subsets of agents that contain them.
The output is a coalition structure where each agent is assigned to a unique coalition.
Since each agent can be in an exponential number of different coalitions, this requires ranking an exponential number of coalitions for each agent. A well-studied subclass of hedonic games are \acp{ASHG} \citep{BoJa02a}, in which agents' preferences are efficiently encoded in cardinal valuations for other agents and the valuation of a coalition is given by a sum-based aggregation.
They have, for example, been used to model the formation of research teams \citep{AlRe04a}.
Most research in this area focuses on a single game for which a solution is searched.

In many real-life scenarios, however, the game changes in various ways over time, and the solution has to be adapted accordingly.
For example, consider a malfunction of an agent that causes other agents to update their valuations towards her once they notice it.
Or consider an agent that significantly over-performs, for example, because of updated software, which again leads to other agents updating their valuations towards her. 
As a third example, consider a group of agents collaborating on a project. Each interaction with other teammates will affect their valuations for her. 

The starting point of our setting is a stable partition that serves as the status quo in a given game. We then investigate how to respond to valuation updates, with the goal of finding a new stable partition.
One approach is to compute a new solution from scratch each time an update occurs.
This is reasonable if there are no costs or other negative effects arising from potentially completely rearranging the coalition structure.
Nevertheless, in many settings, changing the coalition structure has a negative impact. Consider, for example, a group of agents that works together on a task. Removing too many agents from that project will affect the group's current workflow, and it might take a long time to reach the same level of routine as before. In such settings, finding a close stable partition would be beneficial.
We assume a cost proportional to a natural distance measure between the previous and new solution, such that minimizing the distance is equivalent to minimizing the cost.
Specifically, we study the distance metric defined as the minimum number of agents that need to change their coalition to arrive from one coalition structure to another \citep{day1981complexity}.

We focus on stability, one of the most studied solution concepts for ASHGs \citep{BoJa02a,Olse09a,ABS11c,Woeg13a, GaSa19a,BBW21b,BBT23a}, and explore
four well-known stability notions based on deviations by single agents, namely Nash stability (NS), individual stability (IS), contractual Nash stability (CNS), and contractual individual stability (CIS).
In particular, given an initial instance, a stable coalition structure for that instance, and small valuation changes, we study the complexity of finding a nearby stable coalition structure according to our distance function for all four stability notions. 

\section{Our Contribution}
We introduce a model that allows us to study ASHGs under valuation updates.
Our focus is on the algorithmic question of how to find a nearby stable partition, given an ASHG, a stable starting position, and valuation changes.
For this, we study a general setting in which any agent in group $A$ can update their valuations towards anyone in group $B$. As a second case, we consider the symmetric setting in which the valuations between agents are symmetric, and where any valuation update of an agent $a$ towards an agent $b$ also applies to $b$'s valuation for $a$. As important special cases, we consider valuation updates in which all agents change their valuation towards one agent, updates in which one agent changes her valuations towards many agents, and the case that only one agent updates one valuation, or in the symmetric case that only two agents update their valuations for each other.

All of our hardness results hold both for symmetric and non-symmetric valuations.
Our positive results, on the other hand, only hold for the symmetric setting.

We prove that checking whether a close stable partition exists is NP-complete for NS, IS, CNS, and CIS, even if we heavily restrict the set of allowed valuation values, such as $\{-1, 0, 1\}$.
In particular, in \Cref{symmetric_Single_Hardness} we show that finding a stable partition with distance at most $k$ after a single agent updates her preferences towards one single other agent is NP-hard. The same holds in the symmetric setting when only one pair of agents changes their valuation towards each other.
Our reductions are from variants of the set cover problem and also show that finding an approximation is computationally hard even for the case when only the valuations $\{1, 0, -1\}$ are allowed.

Next, we study the four considered stability notions for strict \acp{ASHG}.
We show that both the non-symmetric and symmetric problems remain NP-complete for both NS and IS, even under severely restricted valuations.
For CIS, we show that there always exists a stable partition with distance at most $3$ to the starting partition. For CNS, we give an algorithm that always computes a partition with distance at most $4$ under the assumption that valuations can take on only one negative value. In both cases, we give a polynomial-time algorithm that computes such a partition.

The valuation restrictions that we study subsume restrictions proposed by \citet{DBHS06a}, which are frequently considered in the literature.
\Cref{tab:overview} shows an overview of our computational complexity results for the symmetric setting.

Finally, we also consider a sequence of valuation updates and show that in symmetric \acp{FEG}, one can compute a sequence of stable partitions such that the average number of changes is bounded by a constant for all stability notions we consider. Moreover, we extend this result to strict symmetric \acp{ASHG} for CNS and CIS. Similar results for NS and IS do not exist, which we prove by providing an instance where the average number of changes is linear in the number of agents. In order to evaluate the different settings, we define the \emph{average distance} $d_X(n)$ as the limit of average changes over the worst possible sequence of updates. Notably, $d_X$ only depends on the considered stability notion and subclass of games.

\begin{table}
\centering
\resizebox{1\columnwidth}{!}{
	\renewcommand{\arraystretch}{1.2}
\begin{tabular}{ m{.55cm} | m{1.25cm} m{1.15cm} m{1.25cm} m{1.25cm} m{1.25cm}}
& strict & FENGs & FEGs & AFGs & AEGs \\
X & $\mathbb{Q} \setminus \{0\}$ & $\{-1,0,1\}$ & $\{-1,1\}$ & $\{-1,n\}$ & $\{-n,1\}$ \\ [0.5ex] 
\hline
NS   &NP-c (Th.~\ref{NS_IS_hardness_AEGs}) & NP-c (Th.~\ref{symmetric_Single_Hardness}) & NP-c (Th.~\ref{thm:fegs_hardness}) & NP-c (Th.~\ref{thm:fegs_hardness}) & NP-c (Th.~\ref{NS_IS_hardness_AEGs}) \\
IS   & NP-c (Th.~\ref{IS_completeness_symmetric_FEGs}) & NP-c (Th.~\ref{symmetric_Single_Hardness}) & NP-c (Th.~\ref{IS_completeness_symmetric_FEGs}) & NP-c (Th.~\ref{IS_completeness_symmetric_FEGs}) & NP-c (Th.~\ref{NS_IS_hardness_AEGs}) \\
CNS  & NP-c (Th.~\ref{CNS_hardness_without_0_edges}) & NP-c (Th.~\ref{symmetric_Single_Hardness}) & P (Co.~\ref{cor:CNS_polytime_friends}) & P (Co.~\ref{cor:CNS_polytime_friends}) & P (Co.~\ref{cor:CNS_polytime_friends})\\
CIS  & P (Co.~\ref{Cor:Cis positive})& NP-c (Th.~\ref{symmetric_Single_Hardness}) & P (Co.~\ref{Cor:Cis positive}) & P (Co.~\ref{Cor:Cis positive}) & P (Co.~\ref{Cor:Cis positive})\\
\end{tabular}
}
	\caption{Computational complexity of \textsc{X-1-1-Sym-Altered}. The rows indicate the stability notions, and the columns correspond to the allowed valuation values. For \textsc{X-1-1-Altered}, we show throughout the paper that every entry is NP-complete.}
	\label{tab:overview}
\end{table}

\section{Related Work}
Hedonic games have been introduced by \citet{DrGr80a} and have since then been extensively studied, especially additively separable hedonic games \citep{BoJa02a}. Often, the valuations are also restricted to a limited amount of values. Common considered restrictions only allow valuations $\{-1,n\}$, $\{-n,1\}$ \citep{DBHS06a}, and $\{-1,1\}$, which distinguish only whether an agent is desired or undesired.

A key challenge in \acp{ASHG} is to find stable partitions, meaning that no agent has an incentive to deviate.
In the literature, many different notions of stability have been studied, most commonly different versions of core stability \citep{Kara11a, Woeg13a, Pete17b, PeEl15a, DBHS06a} and Nash (NS), individual (IS), contractual Nash (CNS), and contractual individual stability (CIS) \citep{GaSa19a, PeEl15a, BBT23a}.
In this paper, we focus on the latter ones. 
For most stability notions, a stable partition may not exist.
The existence decision problem is NP-complete for Nash stability \citep{Olse09a}, individual stability \citep{SuDi10a}, and contractual Nash stability \citep{BBT23a}.
\citet{Ball04a} shows that a CIS partition always exists using a potential function argument.
\citet{ABS11c} provide a polynomial-time algorithm to find CIS partitions.
\citet{PeEl15a} give an overview and provide a framework for hardness of different stability notions and classes of games. For symmetric ASHGs, there always exists a stable partition for all four stability notions \citep{BoJa02a}.

Next, we discuss literature on dynamic coalition formation models.
First, there are models that study the convergence of deviation sequences to stable states \citep{
BMM22a,BBK25a, BBT23a}.
\citet{BBK25a} show convergence under the assumption that agents increase their valuations for agents that join and decrease their valuations for agents that leave their coalitions. They also consider the question whether a stable partition can be reached from a given partition after $k$ deviations. They show NP-hardness for NS, IS, and CNS for non-symmetric non-strict ASHGs. This differs from our setting in the fact that we consider starting partitions that were stable in the original game, and we do not insist on the stable partition to be reachable via deviations.

\citet{BBT23a} prove that the individual stability dynamics converge with at most one non-negative valuation, and that contractual Nash stability dynamics converge with at most one non-positive valuation. The main difference of our model to those is that the underlying ASHG changes.
Second, there are online coalition formation models that allow agents to arrive online and require an algorithm to immediately and irrevocably assign them to a coalition upon arrival \citep{FMM+21a,BuRo23a,BuRo25a}.
\citet{BuRo25a} show that finding a stable partition online is not possible in most cases, even if they are guaranteed to exist.
Finally, \citet{CoAg24a} study a model in which agents have incomplete information, i.e., agents are unsure of their preferences and must learn them.
In comparison to those papers, our model allows for valuation changes across all agents and for any agent to be reassigned to any coalition.

Another branch of research related to this is based on robustness.
Given a stability notion $X$, \citet{DBLP:conf/ijcai/IgarashiOSY19} study the complexity of deciding whether a game admits an $X$-stable coalition such that it remains stable in any subgame obtained by removing $k$ agents. They study this problem in symmetric games with valuations in $\{-1, n\}$, which we also consider. Interestingly, many of their positive results hold for stability notions in which the problems we consider are NP-complete, and their hardness results hold when we can give polynomial-time algorithms for the problems we consider. 
The key difference between their work and ours is that they aim to find stable partitions that remain stable after $k$ agents disappear, whereas we seek to find new stable partitions that are sufficiently close to the original partition after agents change their preferences.

\citet{DBLP:conf/aaai/BarrotOSY19} study how unknown agents might affect different stability concepts, and \citet{DBLP:conf/prima/OkimotoSDIM18} study robustness with respect to the social welfare of each agent.

\section{Model}
\label{sec:Model}

\subsection{Hedonic Games}
Let $N$ be a finite set of \emph{agents}.
A nonempty subset $C\subseteq N$ is called a \emph{coalition}.
The set of coalitions containing agent~$i\in N$ is denoted by
$\mathcal N_i:=\{C\subseteq N\mid i\in C\}$.
A set $\pi$ of disjoint coalitions containing all members of $N$ is a \emph{partition} of $N$.
For agent $i\in N$ and partition $\pi$, let $\pi(i)$ denote the unique coalition in~$\pi$ that~$i$ belongs to.

A (cardinal) \emph{hedonic game} is a pair $G = (N,\uf)$ where $N$ is the set of agents and $\uf = (\uf_i)_{i\in N}$
is a tuple of \emph{utility functions} $u_i\colon \mathcal N_i \to \mathbb Q$.
Agents seek to maximize utility and prefer partitions in which their coalition achieves a higher utility.
Hence, we define the utility of a partition~$\pi$ for agent~$i$ as $\uf_i(\pi) := \uf_i(\pi(i))$. The \emph{social welfare} of $G = (N,v)$ given a partition $\pi$ is defined as $\text{SW}(G, \pi) = \sum_{i \in N} u_i(\pi)$.
We denote by $n(G):= |N|$ the number of agents and write $n$ if $G$ is clear from the context.
Following \citet{BoJa02a}, an \emph{\ac{ASHG}} is a hedonic game $(N,\uf)$, where for each agent $i\in N$ there exists a
\emph{valuation function} $\vf_i\colon N\setminus\{i\}\to \mathbb Q$ such that for all $C\in \mathcal N_i$ it holds that
$\uf_i(C) = \sum_{j\in C\setminus \{i\}}\vf_i(j)$.
Note that this implies that the utility for a singleton coalition is~$0$.
Since the valuation functions contain all information for computing utilities, we can also represent \iac{ASHG} as the pair
$(N,\vf)$, where $\vf = (\vf_i)_{i\in N}$ is the tuple of valuation functions.
Additionally, \iac{ASHG} can be succinctly represented as a complete, directed, weighted graph, where the weights of the directed edges induce the valuation functions.

\Iac{ASHG} $(N,\vf)$ is said to be \emph{symmetric} if for every pair of distinct agents $i,j\in N$, it holds that $\vf_i(j) = \vf_j(i)$.
We write $\vf(i,j)$ for the symmetric valuation between~$i$ and~$j$.
A complete undirected weighted graph can represent a symmetric \ac{ASHG}\@.
For simplicity, we also denote this graph by $(N,\vf)$.
\Iac{ASHG} $(N,\vf)$ is said to be \emph{strict} if for every pair of distinct agents $i,j\in N$, it holds that $\vf_i(j) \neq 0$.
Furthermore, a subclass of \acp{ASHG} is (valuation) \emph{restricted} if it contains only valuations from a set of allowed values.
We consider the following classes of restricted \acp{ASHG}, friend-and-enemy games (FEGs), friend-enemy-and-neutral games (FENGs), appreciation-of-friends games (AFGs), and aversion-to-enemies games (AEGs).
They restrict the allowed valuations to $\{-1,1\}$,$\{-1,0,1\}$,$\{-1,n\}$, and $\{-n,1\}$ respectively.

\subsection{Stability}
Next, we define different types of single-agent deviations and their corresponding stability notions.
Assume, we are given a hedonic game $(N,\uf)$, an agent $i\in N$, and a partition $\pi$ of $N$.
A \emph{deviation} by agent $i$ is a beneficial move of agent $i$ from one coalition to another.
Formally, from partition $\pi$ a different partition $\pi'$ can be reached via a deviation of agent $i$ if for all $j \in N \setminus \{i\}$ it holds that $\pi(j) \setminus \{i\} = \pi'(j) \setminus \{i\}$.
A deviation by $i$ is called
\begin{itemize}
    \item a \emph{Nash deviation} (NS) if $\uf_i(\pi')>\uf_i(\pi)$,
    \item an \emph{individual deviation} (IS) if it is a Nash deviation, and for all $j\in\pi'(i)$ it holds that $\uf_j(\pi')\ge\uf_j(\pi)$,
    \item a \emph{contractual Nash deviation} (CNS) if it is a Nash deviation, and for all $j\in\pi(i)$ it holds that $\uf_j(\pi')\ge\uf_j(\pi)$,
    \item a \emph{contractual individual deviation} (CIS) if it is an individual deviation and a contractual Nash deviation.
\end{itemize}
For each type of deviation, we define the respective stability concept as the absence of corresponding deviations. For example, a partition $\pi$ is said to be Nash stable (NS) if no agent can perform a Nash deviation.

Note that Nash stability implies the other three notions, while both contractual Nash stability and individual stability imply contractual individual stability \citep[see][for a more extensive overview]{AzSa15a}.
For every \ac{ASHG}, there exists a CIS partition \citep{ABS11d}.
This is not the case for the other stability notions, meaning that for each of them, there exists \iac{ASHG} without any such stable partition \citep{BoJa02a,SuDi07b}. However, using a potential function argument, \citet{BoJa02a} show that every symmetric \ac{ASHG} admits a Nash stable partition.

\subsection{Problem Statement} 
We say that two partitions $\pi, \pi'$ have distance $1$ if we can reach $\pi'$ from $\pi$ via a change of coalition for some agent $i$.
Further, we define $d(\pi,\pi')$ as the length of a minimal sequence $\pi_0, \pi_1, \dots, \pi_s$, such that for all $t \in [s]$ it holds that $d(\pi_{t-1},\pi_t) = 1$ and $\pi=\pi_0$ and $\pi'=\pi_s$.
This distance can be computed in polynomial time \citep[Theorem 9 in][]{day1981complexity}.

In this paper, we are interested in the following problem, where $X \in \{\text{NS, IS, CNS, CIS}\}$. 

\begin{problem}
	\problemtitle{\textsc{X-A-B-Altered}}
	\probleminput{\Iac{ASHG} $G=(N,\vf)$, an X partition $\pi$ of $G$, two sets of agents $D,E\subseteq N$ with $\lvert D\rvert = A$ and $\lvert E \rvert = B$,
	an updated valuation function $\vf'$ such that $\vf'_i(j)\neq \vf_i(j)$ implies that $i\in D$ and $j \in E$, and a distance parameter $k\in\mathbb{N}$.}
	\problemsolution{An X partition $\pi'$ of $G'=(N,\vf')$,
	where $d(\pi, \pi') \leq k$ if it exists.}
\end{problem}

In this paper, we often consider the restricted setting in which $A$ or $B$ equals $1$, meaning that either one agent changes her valuation for all other agents or all agents change their valuations for one agent. Additionally, we consider the case of symmetric preferences and define the problem as follows:

\begin{problem}
	\problemtitle{\textsc{X-A-B-Sym-Altered}}
	\probleminput{A symmetric ASHG $G=(N,\vf)$, an X partition $\pi$ of $G$, two sets of agents $D,E\subseteq N$ with $\lvert D\rvert = A$ and $\lvert E \rvert = B$,
	an updated symmetric valuation function $\vf'$ such that 
    $\vf'(i,j) \neq \vf(i,j)$ implies that  $i\in D$, $j \in E$ or $i\in E$, $j \in D$, and a distance parameter $k\in\mathbb{N}$.
    }
	\problemsolution{An X partition $\pi'$ of $G'=(N,\vf')$,
	where $d(\pi, \pi') \leq k$ if it exists.}
\end{problem}

As a consequence of the definition, \textsc{X-A-B-Sym-Altered} and \textsc{X-B-A-Sym-Altered} are equivalent.
Observe that \textsc{X-A-B-Sym-Altered} is not a special case of \textsc{X-A-B-Altered}. For the simplest case, note that in \textsc{X-1-1-Altered} exactly one agent changes her valuation towards one other agent, whereas in \textsc{X-1-1-Sym-Altered} two agents change their valuations towards each other.
Also, note that if $A \leq A'$ and $B \leq B'$, an instance of \textsc{X-A-B-(Sym)-Altered} is also a valid instance of \textsc{X-A'-B'-(Sym)-Altered}. Thus, showing NP-hardness for the former implies NP-hardness for the latter. 
Finally, we always assume that updates do not affect the subclass of \acp{ASHG}. If, for example, we consider FEGs, the updated valuations must also be in $\{-1, 1\}$.

\section{Results}
\label{sec:Results}
In this section, we present our results.
They are divided into three parts. In the first part, we show hardness results for all four considered stability notions for general ASHGs. In the second part, we restrict our attention to strict ASHGs.
Finally, in the last part, we conduct an average-case analysis on the distance for a sequence of valuation updates.

To show NP-completeness, we give reductions from the set cover problem and restricted variants of it, which are known to be NP-complete \citep{Karp72a, GaJo79a}.
Before we give the reductions, we state two lemmas that guarantee that all problems we consider are in NP and that, if $k$ is fixed, there are only a polynomial number of partitions with distance at most $k$ from the original partition. The proof of \Cref{lemmaEverythingINNP} is based on a result of \citet{day1981complexity}. Both proofs can be found in the appendix, along with all other omitted proofs.

\begin{restatable}{lemma}{everythingInNOP}
\label{lemmaEverythingINNP}
    The distance $d(\pi, \pi')$ between two partitions $\pi$ and $\pi'$ can be computed in polynomial time. Moreover, checking whether $\pi$ is $X$ for $X \in \{\text{NS, IS, CNS, CIS}\}$ can be done in polynomial time.
\end{restatable}

\begin{restatable}   
{lemma}{lemmaDistanceK}
    For any coalition structure and any partition $\pi$, there exist at most $n^{2k}$  partitions with distance $k$ from $\pi$. 
    \label{Lemma:Distancek}
\end{restatable}

The latter lemma directly implies that if $k$ is a constant, we can enumerate all possible coalition structures with distance at most $k$ in polynomial time and can therefore efficiently find a closest stable partition if it exists. 

Next, we show that for all stability notions we consider, finding a stable partition after updating only one valuation is already NP-hard if we allow 0-valuations and at least one positive and one negative valuation value.

\begin{restatable}
    {theorem}{thmeverythingISNPhard}\label{symmetric_Single_Hardness}
    Let $\alpha, \beta: \mathbb{N} \rightarrow \mathbb{Q}_{>0}$ be two polynomial-time computable functions. Then,
    \begin{itemize}
        \item \emph{X-1-1-Altered} and \emph{X-1-1-Sym-Altered} are NP-complete for $X \in \{\text{IS}, \text{CIS}\}$ in \acp{ASHG} with valuations in $\{-\beta(n), 0, \alpha(n)\}$.
        \item \emph{X-1-1-Altered} and \emph{X-1-1-Sym-Altered} are NP-complete for $X \in \{\text{NS}, \text{CNS}\}$ in \acp{ASHG} with valuations in $\{-\beta(n), 0, \alpha(n)\}$ if it holds that $\alpha(n) \leq \beta(n)$ for all $n \in \mathbb{N}$.
    \end{itemize}
\end{restatable}

\begin{proof}[Proof sketch]
    We give a proof sketch for the case $\alpha(n) \leq \beta(n)$ and $X \in \{\text{NS, IS, CNS, CIS}\}$ for \textsc{X-1-1-Sym-Altered}. The other settings are covered by slight modifications of this argument, which are detailed in the appendix. The proof is based on a reduction from the set cover problem.
    
    Let $(E, \mathcal{S}, k)$ be a set cover instance. Without loss of generality, we may assume that $k < \lvert E \rvert, \lvert \mathcal{S} \rvert$ and that $\mathcal{S}$ is a set cover of $E$, as otherwise the problem is trivial.
    Let $G$ be a symmetric \ac{ASHG} with agent set $N = E \cup \mathcal{S} \cup Y\cup \{z_1, z_2\}$, where $\lvert Y \rvert = \lvert E \rvert + \lvert \mathcal{S} \rvert + 3$. The valuations of $G$ are given by 
    \begin{itemize}
        \item $v(e, S) = -\beta$ for all $S \in \mathcal{S}$ and $e \in S$,
        \item $v(z_1, e) = \alpha$ for all $e \in E$,
        \item $v(z_2, e) = v(z_2, S) = -\beta$ for all $S \in \mathcal{S}$ and $e \in E$,
        \item $v(z_1, y) = v(z_1, z_2) = \alpha$ for all $y \in Y$,
        \item $0$ for all other valuations,
    \end{itemize}
    where we set $\alpha := \alpha(\lvert N \rvert)$ and $\beta := \beta(\lvert N \rvert)$.

    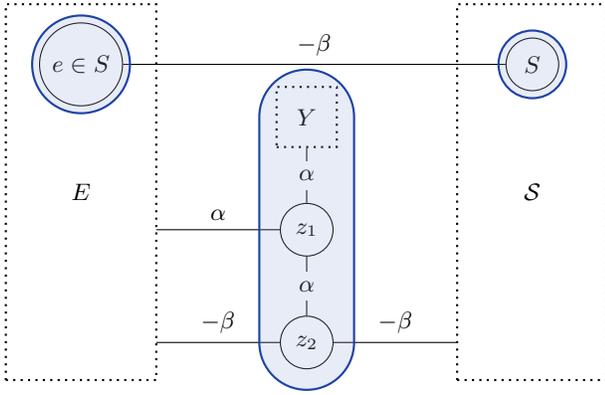
\begin{figure}[ht]
    \centering
\begin{tikzpicture}[>=stealth,
        every node/.style={font=\small},
        main/.style={circle,draw,minimum size=7mm,font=\small},
        rect/.style={draw, thick, dotted, minimum width=2cm, minimum height=5cm, align=center}]
        
    \node[rect] (U)  at (-3.75,0) {$E$};
    \node[rect] (S*)  at ( 2.25,0) {$\mathcal{S}$};
    
    \node[rect, minimum width = 0.8cm, minimum height = 0.8cm, ] (y)  at (-0.75, 1)   {$Y$};
    \node[main] (z)  at (-0.75, -0.5)   {$z_1$};
    \node[main] (zp) at (-0.75,-2)   {$z_2$};

    \node[main] (u) at (-3.75, 1.7) {$e \in S$};

    \node[main] (S) at (2.25, 1.7) {$S$};

    \begin{scope}[every node/.style={font=\small}]
                \path
                    (y) edge[-] node[pos = 0.5, fill = white] {$\alpha$} (z)
                    (z)  edge[-] node[pos = 0.5, fill = white] {$\alpha$} (zp)
                    (u) edge[-] node[above] {$-\beta$} (S);
            \end{scope}     
            
    \draw[-] (z) -- node[above] {$\alpha$} (-2.75, -0.5);
    \draw[-] (zp) -- node[above] {$-\beta$} (1.25, -2);
    \draw[-] (zp) -- node[above]  {$-\beta$} (-2.75, -2);

    \draw[thick,CoalitionColor, fill=CoalitionColor!50, fill opacity=0.2]  \convexpath{y,zp}{.63cm};
    \node (u2) at (-3.75, 1.70001) {};
    \node (S2) at (2.25, 1.70001) {};
    \draw[thick,CoalitionColor, fill=CoalitionColor!50, fill opacity=0.2]  \convexpath{u,u2}{.65cm};
    \draw[thick,CoalitionColor, fill=CoalitionColor!50, fill opacity=0.2]  \convexpath{S,S2}{.45cm};
\end{tikzpicture}
\caption{Illustration of the symmetric ASHG $G$ from the reduction in \Cref{symmetric_Single_Hardness}. Dotted rectangles are independent sets. All edges that are not shown correspond to a valuation of $0$. Coalitions in $\pi$ are highlighted in blue.}
\end{figure}
    It holds that $\pi = \left\{ \{e\}_{e \in E}, \{S\}_{S \in \mathcal{S}}, Y \cup \{z_1, z_2\} \right\}$ is an $X$ partition in $G$ since no agent can make a Nash deviation.

    Alter the valuation between agents $z_1$ and $z_2$ such that $v'(z_1, z_2) = -\beta$ and denote the altered game by $G'$.
    We claim that a set cover $\mathcal{C} \subseteq \mathcal{S}$ of $E$ with $\lvert \mathcal{C} \rvert \leq k$ exists if and only if $G'$ has an $X$ partition $\pi'$ with $d(\pi, \pi') \leq k + 1$.
    
    To this end, suppose that $\mathcal{C} \subseteq \mathcal{S}$ is a set cover of $E$ of size at most $k$. We claim that $\pi' = \left\{ \{e\}_{e \in E}, \{S\}_{S \in \mathcal{S \setminus C}}, \mathcal{C} \cup Y \cup \{z_1\}, \{z_2\} \right\}$ is an $X$ partition in $G'$. To see this, note that $u_a(\pi'(A)) \geq 0$ for all $a \in N$ and the only deviations by an agent to another coalition in which she has a positive valuation for at least one other agent are deviations by $e \in E$ to $\pi'(z_1)$.
    However, no $e \in E$ can Nash deviate to $\pi'(z_1)$ since there is some $S \in \mathcal{C} \subseteq \pi'(z_1)$ with $e \in S$ and $v(e, S) = -\beta$, so $u_e(\pi'(z_1)) \leq \alpha - \beta \leq 0$. Clearly, it holds that $d(\pi, \pi') \leq k+1$.

    To prove the reverse direction, suppose that $\pi'$ is an $X$ partition in $G'$ with $d(\pi, \pi') \leq k + 1$. Note that since $v(y, z_1) = \alpha$ and $v(y, a) = 0$ for all $a \in N \setminus \{z_1\}$, we have that $y \in \pi'(z_1)$ for all $y \in Y$. Moreover, $z_2$ is in a singleton coalition in $\pi'$ since $v'(z_2, a) < 0$ for all $a \in N \setminus (Y \cup \{z_2\})$.
    Let $R_E$ and $R_{\mathcal{S}}$ denote the agents in $E$ and $\mathcal{S}$ that were moved to another coalition from $\pi$. Let $f: E \rightarrow \mathcal{S}$ map each $e \in E$ to some $S \in \mathcal{S}$ with $e \in S$. We claim that $\mathcal{C} := f(R_E) \cup R_{\mathcal{S}}$ is a set cover of $E$ with $\lvert \mathcal{C} \rvert \leq k$. By assumption, it holds that $\lvert R_E \rvert + \lvert R_{\mathcal{S}} \rvert \leq k$ since $z_2$ has moved to another coalition as well. To see that $\mathcal{C}$ is a set cover of $E$, assume for contradiction that there is some $e \in E$ that is not covered by $\mathcal{S}$. Hence, by construction of $\mathcal{C}$, there is no $S \in \pi'(z_1)$ with $v(e, S) < 0$. Therefore, $e$ can $X$ deviate to $\pi'(z_1)$ since there is no $a \in \pi'(e)$ with $v(e, a) > 0$ and no $b \in \pi'(z_1)$ with $v(e, b) < 0$. This is a contradiction to $\pi'$ being $X$ in $G'$, so $\mathcal{C}$ is in fact a set cover of $E$. 
\end{proof}

Due to the one-to-one correspondence of the problem with the set cover problem, we can even show the following statement \citep{hardnessSetCover}.

\begin{restatable}{corollary}{approxIsHard}
Approximating \textsc{X-A-B-Sym-Altered} better than by a logarithmic factor is NP-hard. 
\end{restatable}

These results paint a rather negative picture. Whenever agents can be indifferent to other agents, allowing one positive and one, in absolute size, larger negative value is sufficient for hardness. In the reduction, the 0-valuations are crucial.
Allowing agents to be completely indifferent to whether an agent is in her coalition or not circumvents the restrictions that different stability notions impose. Therefore, in the following, we consider strict \acp{ASHG}.

\section{Strict ASHGs} 
\label{sec:StrictASHGs}
The results for IS and NS differ from the results for CNS and CIS. This is because contractual deviations align better with stability and allow for efficient algorithms in some settings. The intuition is that for the contractual notion, for an agent not to have any valid deviation, it is sufficient to have one agent in the same coalition that approves her, whereas for individual stability, we might require an agent in every other coalition that blocks her from deviating.

\subsection{CNS and CIS}
In this section, we consider CNS and CIS for strict \acp{ASHG}. Our negative results hold in both the symmetric and non-symmetric settings, whereas our positive results require symmetry.

We first prove that even in these restricted settings, \textsc{X-1-N-Sym-Altered} and \textsc{X-1-N-Altered} remain NP-hard for both CNS and CIS.

\begin{restatable}
    {theorem}{theoremOnlyTwoWeights}\label{theorem:altered_complete}
	\textsc{X-1-N-Altered} and \textsc{X-1-N-Sym-Altered} are NP-complete for $X \in \{\text{CNS, CIS}\}$ in strict \acp{ASHG} with valuations in $\{-\beta(n), \alpha(n)\}$, where $\alpha, \beta : \mathbb{N} \rightarrow \mathbb{Q}_{> 0}$ are polynomial-time computable functions.
    \label{CISHardness}
    \end{restatable}

The theorem shows that finding close stable partitions in the symmetric setting is computationally hard when altering all preferences of an agent, even when only a single positive and negative valuation value exists. In particular, this implies hardness for FEGs, AFGs, and AEGs.

In the following, we again consider the setting in which an agent can only change her valuations for one other agent (or, in the symmetric setting, only one pair of valuations changes).
Even in this setting, the problem remains hard for CNS. The reduction is again from set cover.
\begin{restatable}{theorem}{CNSSym}
\label{CNS_hardness_without_0_edges}
    \textsc{CNS-1-1-Altered} and \textsc{CNS-1-1-Sym-Altered} are NP-complete for strict \acp{ASHG}.
\end{restatable}

The hardness of \textsc{CNS-1-1-Sym-Altered} is closely related to the number of singleton coalitions in $\pi$ since we can compute a CNS partition $\pi'$ whose distance to $\pi$ only depends on the number of singleton coalitions in $\pi$.

\begin{restatable}{theorem}{CNSFourTheorem}
\label{thm:general_CNS_distance_bound}
For the \textsc{CNS-1-1-Sym-Altered} problem, Algorithm \ref{alg_closeCNS} computes a CNS partition $\pi'$ with $d(\pi, \pi') \leq 4 + \phi(\pi) - \phi(\pi')$, where $\phi(\pi)$ denotes the number of singleton coalitions in $\pi$. 
\end{restatable}

\begin{proof}[Proof sketch]
    The proof uses the following two observations about CNS partitions in strict symmetric \acp{ASHG}.
    \begin{itemize}
        \item For any two agents $s$ and $s'$ in singleton coalitions, it holds that $v(s, s') < 0$ and $u_s(C) \leq 0$ for any $C \in \pi$ with $\lvert C \rvert > 1$.
        \item For any agent $x$ with $\lvert \pi(x) \rvert > 1$, there must be some $y \in \pi(x)$ with $v(x, y) > 0$.
    \end{itemize}
    Now let the altered valuation be between two agents $a$ and $b$.
    We perform a case analysis depending on whether $a$ and $b$ are in the same coalition or not. 
    
    \textbf{Case 1:  $\pi(a) = \pi(b).$} 
    Without loss of generality, we can assume that $v'(a, b) < 0$ as otherwise $\pi$ is again CNS in $G'$. If there is an agent $s_a$ (or $s_b$) in a singleton coalition in $\pi$  with $v(a, s_a) > 0$ (or $v(b, s_b) > 0$), then $a$ (or $b$) cannot make a CNS deviation after merging their respective coalitions.
    If either $a$ or $b$ leaves their coalition, then only agents in singleton coalitions can join this coalition by CNS deviations. Due to the first observation, no agent in a singleton coalition can make a CNS deviation to a coalition other than $\pi(a)$ since $\pi$ is CNS in $G$.
    Since there can be at most $2$ singleton coalitions in $\pi'$ that are not in $\pi$ and at most two agents in non singleton coalitions can make a CNS deviation in the algorithm, it holds that $d(\pi, \pi') \leq 4 + \phi(\pi) - \phi(\pi').$

    \textbf{Case 2: $\pi(a) \neq \pi(b).$} 
    If $\pi$ is not CNS in $G'$, then either $a$ or $b$ is in a singleton coalition. In this case, merging $\pi(a)$ and $\pi(b)$ yields a CNS partition $\pi'$ with $d(\pi, \pi') = 1$. 
\end{proof}
\begin{algorithm}
\caption{CloseCNS}
\label{alg_closeCNS}
\SetKwInOut{Input}{input}\SetKwInOut{Output}{output}
\Input{A symmetric strict \ac{ASHG} $(N, v)$, a CNS partition $\pi$, and a pair of agents $(a, b)$ together with an altered valuation $v'(a, b)$.}
\Output{A CNS stable partition $\pi'$ for the altered \ac{ASHG}.}
\If{$\pi(a) = \pi(b)$}{
\For{$x \in \{a, b\}$}{
\If{$\exists$ $s_x$ with $\pi(s_x) = \{s_x\}$ and $v(x, s_x) > 0$}{
$C \gets \pi(x) \cup \pi(s_x)$\\
$\pi \gets \left( \pi \setminus \{\pi(x), \pi(s_x)\} \right) \cup \{C\}$  
}\ElseIf{$x$ can make a CNS deviation}{Let $x$ make a CNS deviation to the largest possible coalition}
}
}
Let agents make CNS deviations until no more are possible and denote the resulting partition by $\pi'$\\
\Return $\pi'$
\end{algorithm}

Algorithm \ref{alg_closeCNS} also works for CIS if, instead of CNS deviations, we consider CIS deviations. However, we can prove a stronger bound on $d(\pi, \pi')$ that is independent of $\phi(\pi)$ and is based on any sequence of CIS deviations.

\begin{restatable}{theorem}{CISdistance}\label{thm:CIS_FEGS_polynomial_time}
    Let $G$ and $G'$ be strict symmetric \acp{ASHG} such that $G'$ arises from $G$ by altering the valuation between two agents, and let $\pi$ be CIS in $G$.
    Then, any sequence of CIS deviations in $G'$ starting from $\pi$ converges to a CIS partition $\pi'$ with $d(\pi, \pi') \leq 3$.
\end{restatable}

Due to \Cref{thm:CIS_FEGS_polynomial_time} we can guarantee the existence of a
CIS partition with distance at most $3$. If we are interested, if there exists an even closer CIS partition, we can compute all possible partitions with distance $1$ and $2$ in polynomial time due to \Cref{Lemma:Distancek} and check whether any is CIS.

\begin{corollary}\label{Cor:Cis positive}
        \textsc{CIS-1-1-Sym-Altered} can be decided in polynomial time in strict \acp{ASHG}.
\end{corollary}

To prove a constant distance bound for CNS, we need to further restrict the valuations, such that only a constant number of agents in singleton coalitions in $\pi$ can deviate.

\begin{restatable}{theorem}{CNSDistanceFourTheoremLemma}
    \label{thm:symmetric_fegs_distance_CNS}

    For \textsc{CNS-1-1-Sym-Altered} in strict \acp{ASHG} with at most one negative valuation value, \Cref{alg_closeCNS} computes a CNS partition $\pi'$ with $d(\pi, \pi')\leq 4.$
\end{restatable} 

\begin{proof}[Proof sketch]
    Let $-\beta$ denote the negative valuation value.
    By the proof of \Cref{thm:general_CNS_distance_bound}, we only need to bound the number of agents in singleton coalitions that make a CNS deviation to $\pi(a)$ in the case $\pi(a) = \pi(b)$.
    Again, we may assume that $v'(a, b) < 0$. At most two agents, namely $a$ and $b$, can leave $\pi(a) = \pi(b)$ by a CNS deviation.
    Moreover, for any two agents $s, s'$ in singleton coalitions in $\pi$, it holds that $v(s, s') = -\beta < 0$ and $u_s(\pi(a) \setminus \{a, b\}) \leq u_s(\pi(a)) + 2 \beta \leq 2 \beta$, where $u_s(\pi(a)) \leq 0$ since $\pi$ is CNS in $G$. Hence, at most two agents in singleton coalitions can join this coalition by a CNS deviation.
\end{proof}
As for CIS, we can infer from \Cref{thm:symmetric_fegs_distance_CNS} that we can find a closest CNS partition efficiently under such valuation restrictions.

\begin{corollary}\label{cor:CNS_polytime_friends}
        \textsc{CNS-1-1-Sym-Altered} can be decided in polynomial time in strict \acp{ASHG} with at most one negative valuation value (which includes \acp{FEG}, AEGs, and AFGs).
\end{corollary}

While one can show a similar bound as in \Cref{thm:symmetric_fegs_distance_CNS} for multiple negative valuation values if they are sufficiently close to one another, this is not possible in general, as \Cref{CNS_hardness_without_0_edges} shows. 
In the following, we give an example, which shows that the bounds from \Cref{thm:general_CNS_distance_bound} and 
\Cref{thm:symmetric_fegs_distance_CNS} are tight.

\begin{example}
    Consider a symmetric \ac{FEG} $G$ with $8$ agents and positive valuations defined by $v(6, 7) = 1$ and $v(i, j) = 1$ for all $i \in \{0, 4, 5\}, j \in \{1,2, 3\}$.
    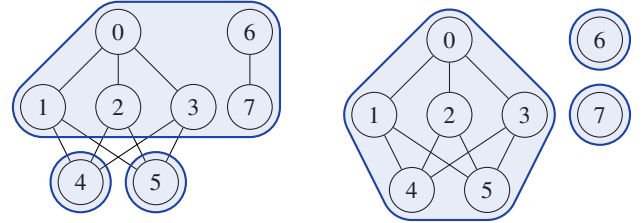
\begin{figure}[ht]
\begin{tikzpicture}[
    main/.style={circle,draw,minimum size=4mm,font=\footnotesize},
    dottedbox/.style={draw, blue, dotted,thick, rounded corners=10pt, inner sep=6pt}
]

  \node[main] (n0) at ( 0, 0) {0};

  \node[main] (n1) at (-1,-1) {1};
  \node[main] (n2) at ( 0,-1) {2};
  \node[main] (n3) at ( 1,-1) {3};

  \node[main] (n4) at (-0.5,-2) {4};
  \node[main] (n5) at ( 0.5,-2) {5};

  \node[main] (n6) at ( 1.75, 0) {6};
  \node[main] (n7) at ( 1.75,-1) {7};

  \begin{scope}[every node/.style={font=\small}]
                \path
                    (n0) edge[-] node[] {} (n1)
                    (n0) edge[-] node[] {} (n2)
                    (n0) edge[-] node[] {} (n3)
                    (n1) edge[-] node[] {} (n4)
                    (n1) edge[-] node[] {} (n5)
                    (n2) edge[-] node[] {} (n4)
                    (n2) edge[-] node[] {} (n5)
                    (n3) edge[-] node[] {} (n4)
                    (n3) edge[-] node[] {} (n5)
                    (n6) edge[-] node[] {} (n7);
    \end{scope}     
    \node (4) at (-0.9, -2.4) {};
    \node (5) at (0.1, -2.4) {};
    \draw[thick,CoalitionColor, fill=CoalitionColor!50, fill opacity=0.2]  \convexpath{n6, n7, n1, n0}{.4cm};
    \draw[thick,CoalitionColor, fill=CoalitionColor!50, fill opacity=0.2] 
    \convexpath{4}{.4cm};
    \draw[thick,CoalitionColor, fill=CoalitionColor!50, fill opacity=0.2]
    \convexpath{5}{.4cm};

\end{tikzpicture}
\qquad
\begin{tikzpicture}[
    main/.style={circle,draw,minimum size=4mm,font=\footnotesize},
    dottedbox/.style={draw, blue, dotted,thick, rounded corners=5pt, inner sep=6pt}
]

  \node[main] (n0) at ( 0, 0) {0};

  \node[main] (n1) at (-1,-1) {1};
  \node[main] (n2) at ( 0,-1) {2};
  \node[main] (n3) at ( 1,-1) {3};

  \node[main] (n4) at (-0.5,-2) {4};
  \node[main] (n5) at ( 0.5,-2) {5};

  \node[main] (n6) at ( 2, 0) {6};
  \node[main] (n7) at ( 2,-1) {7};

\begin{scope}[every node/.style={font=\small}]
                \path
                    (n0) edge[-] node[] {} (n1)
                    (n0) edge[-] node[] {} (n2)
                    (n0) edge[-] node[] {} (n3)
                    (n1) edge[-] node[] {} (n4)
                    (n1) edge[-] node[] {} (n5)
                    (n2) edge[-] node[] {} (n4)
                    (n2) edge[-] node[] {} (n5)
                    (n3) edge[-] node[] {} (n4)
                    (n3) edge[-] node[] {} (n5);
    \end{scope}  
    \node (6) at (1.6, -0.4) {};
    \node (7) at (1.6, -1.4) {};
    \draw[thick,CoalitionColor, fill=CoalitionColor!50, fill opacity=0.2]  \convexpath{n3, n5, n4, n1, n0}{.4cm};
    \draw[thick,CoalitionColor, fill=CoalitionColor!50, fill opacity=0.2] 
    \convexpath{6}{.4cm};
    \draw[thick,CoalitionColor, fill=CoalitionColor!50, fill opacity=0.2] 
    \convexpath{7}{.4cm};
\end{tikzpicture}

\caption{\label{fig:tight}A symmetric \ac{FEG} where the closest CNS coalition has distance 4. There is an edge between agents $i$ and $j$ if and only if $v(i, j) = 1$. Otherwise, it holds that $v(i,j) = -1$.} The original game is on the left and the altered game on the right. Coalitions in $\pi$ and $\pi'$ are highlighted in blue.
\end{figure}
    
    Then, $\pi = \{\{0, 1, 2, 3,6, 7\}, \{4\}, \{5\} \}$ is CNS in $G$ since every agent either has a (symmetric) positive valuation for another agent in her coalition or cannot increase her utility by deviating. Alter the valuation between agents $6$ and $7$, such that $w'(6, 7) = -1$ and denote the altered game by $G'$. Clearly, $\pi' = \{\{0, 1, 2, 3, 4, 5\}, \{6\}, \{7\}\}$ is CNS in $G'$ and it holds that $d(\pi, \pi') = 4$. 
    We claim that for any CNS partition $\tilde{\pi}$ in $G'$, it holds that $d(\pi, \tilde{\pi})\geq 4$. To this end, note that 
    agents $6$ and $7$ must be in a singleton coalition in $\pi$, respectively, since they have negative valuations for all agents in $G'$. Moreover, it holds that $\lvert \pi'(4) \rvert, \lvert\pi'(5)\rvert > 1$ and $\pi'(4)\neq \{4, 5\}$. Hence, we have that $d(\pi, \tilde{\pi}) \geq 4$.
\end{example}

On the other hand, if we do not assume that the game is symmetric finding a close partition is not possible in general.
Below, we give an example of an (non-symmetric) FEG in which the closest CIS partition has distance $\Theta(n)$ after altering one valuation of an agent.

\begin{example}
Define an FEG $G = (N, v)$ with $N = [n]$ and $v_i(j) = 1$ if and only if $j = (i \mod n) + 1$ for all $i, j \in [n]$  (see \Cref{figure:CISNonSym}).
Clearly, the grand coalition $\pi = \{N\}$ is CNS/CIS in $G$ since every agent is a friend of at least one other agent, so no agent can make a CNS/CIS deviation. Alter the valuation of the last agent such that $v'_n(1) = -1$ and denote the altered game by $G'$. We claim that for any CNS/CIS partition $\pi'$ in $G'$, it holds that $\lvert C \rvert \leq 3$ for all coalitions $C \in \pi'$.
To see this, note that for any $C \in \pi'$, there exists some $i \in C$ with $v_j(i) = -1$ for all $j \in C \setminus \{i\}$. Since $i$ can have at most one friend in $C$, she can make a CNS/CIS deviation to the empty coalition if $\lvert C \rvert \geq 4$. Therefore, it holds that $d(\pi, \pi') \geq n - 3$.

    \begin{figure}[ht]
\begin{tikzpicture}[
    main/.style={circle,draw,minimum size=4mm,font=\footnotesize},
    dottedbox/.style={draw, blue, dotted,thick, rounded corners=10pt, inner sep=6pt}
]

  \node[main] (n0) at ( -1, 0) {1};

  \node[main] (n1) at (0.25, 0) {2};
  \node[main] (n2) at ( 1,-1) {3};
  \node[main] (n3) at ( 0.25,-2) {4};

  \node[main] (n4) at (-1,-2) {$5$};

  \node[main] (nn) at ( -1.75,-1) {$6$};

  \begin{scope}[every node/.style={font=\small}]
                \path
                    (n0) edge[->] node[] {} (n1)
                    (n1) edge[->] node[] {} (n2)
                    (n2) edge[->] node[] {} (n3)
                    (n3) edge[->] node[] {} (n4)
                    (n4) edge[->] node[] {} (nn)
                    (nn) edge[->] node[] {} (n0);
    \end{scope}
    \draw[thick,CoalitionColor, fill=CoalitionColor!50, fill opacity=0.2]
    \convexpath{nn, n0, n1, n2, n3,  n4}{.4cm};
\end{tikzpicture}
\qquad
\begin{tikzpicture}[
    main/.style={circle,draw,minimum size=4mm,font=\footnotesize},
    dottedbox/.style={draw, blue, dotted,thick, rounded corners=10pt, inner sep=6pt}
]

  \node[main] (n0) at ( -1, 0) {1};

  \node[main] (n1) at (0.25, 0) {2};
  \node[main] (n2) at ( 1,-1) {3};
  \node[main] (n3) at ( 0.25,-2) {4};

  \node[main] (n4) at (-1,-2) {$5$};

  \node[main] (nn) at ( -1.75,-1) {$6$};

  \begin{scope}[every node/.style={font=\small}]
                \path
                    (n0) edge[->] node[] {} (n1)
                    (n1) edge[->] node[] {} (n2)
                    (n2) edge[->] node[] {} (n3)
                    (n3) edge[->] node[] {} (n4)
                    (n4) edge[->] node[] {} (nn);
    \end{scope}     
    \draw[thick,CoalitionColor, fill=CoalitionColor!50, fill opacity=0.2]
    \convexpath{n0, n1, n2}{.4cm};
    \draw[thick,CoalitionColor, fill=CoalitionColor!50, fill opacity=0.2]
    \convexpath{nn, n3, n4}{.4cm};
\end{tikzpicture}
\caption{An FEG where the closest CNS/CIS partition has distance $n-3$ after altering one valuation of agent $n$ depicted for $n = 6$. There is an arc $(i, j)$ between agents $i$ and $j$ if and only if $v_i(j) = 1$. Otherwise, it holds that $v_i(j) = -1$. The original game is on the left and the altered game on the right. Coalitions in $\pi$ and $\pi'$ are highlighted in blue.}
\label{figure:CISNonSym}
\end{figure}
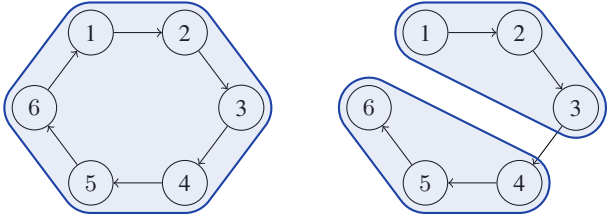
\end{example}

Since such examples exist, it is not surprising that in the non-symmetric setting even for FEGs, AFGs, and AEGs determining the existence of a partition with distance at most $k$ is NP-hard.

\begin{restatable}{theorem}{CINSasym}\label{thm:CISandCNS_single_asymmetric_NP_complete}
    
    \textsc{X-1-1-Altered} is NP-complete for $X \in \{\text{CNS, CIS}\}$ in strict \acp{ASHG} with valuations in $\{-\beta(n), \alpha(n)\}$ for any polynomial-time computable functions $\alpha, \beta : \mathbb{N} \rightarrow \mathbb{Q}_{> 0}$.
\end{restatable}

\subsection{NS and IS}
As mentioned, the picture is quite different for NS and IS.
We first show NP-hardness for IS if there exists at least one positive and one negative valuation value such that the positive valuation value is at least as large as the negative one.

\begin{restatable}{theorem}{ISsymmetric}\label{IS_completeness_symmetric_FEGs}
    \textsc{IS-1-1-Altered} and \textsc{IS-1-1-Sym-Altered} are NP-complete in strict \acp{ASHG} with valuations in $\{-\beta(n), \alpha(n)\}$, where $\alpha, \beta: \mathbb{N} \rightarrow \mathbb{Q}_{> 0}$ are polynomial-time computable functions with $\alpha(n) \geq \beta(n)$ for all $n \in \mathbb{N}$.
\end{restatable}

Note that \Cref{IS_completeness_symmetric_FEGs} covers the case of FEGs and AFGs.
Using a reduction from Exact Cover by 3-Sets (X3C), we can prove NP-completeness of \textsc{IS-1-1-(Sym)-Altered} for the remaining case of AEGs. 
Since NS and IS are equivalent in symmetric AEGs, this immediately implies the hardness of \textsc{NS-1-1-Sym-Altered}. However, we can also utilize this fact to prove hardness for the non-symmetric case.

\begin{restatable}{theorem}{AEGsNP}\label{NS_IS_hardness_AEGs}
    \textsc{X-1-1-Altered} and \textsc{X-1-1-Sym-Altered} are NP-complete for $X \in \{\text{NS, IS}\}$ in AEGs.
\end{restatable}

Those results show that \textsc{X-1-1-(Sym)-Altered} remains hard even for very restricted settings. However, for Nash stability, we relied on the fact that IS and NS are equivalent in symmetric AEGs to prove hardness. For FEGs and AFGs, this is not possible.
Here, one needs to provide a reduction by arguing only about the valuations of each agent, without relying on any single agent blocking others from making a deviation.
To this end, we give a reduction from a restricted version of X3C, which allows us to better control how the utilities of agents change after altering the valuation between two of them.

\begin{restatable}{theorem}{NSsymFEGHard}\label{thm:fegs_hardness}
    \textsc{NS-1-1-Altered} and \textsc{NS-1-1-Sym-Altered} are NP-complete in FEGs and AFGs. 
\end{restatable}

\section{Average Distance of Multiple Updates}
\label{sec:AverageAnalysis}

In the previous sections, we showed that, in most of the considered settings, deciding whether a close partition exists is NP-hard. 
In the following, we move away from considering only single updates (in the symmetric setting, this means two agents update their valuation towards each other) and study the case in which a sequence of updates occurs. In this setting, we can prove more positive results.
Let $\mathcal{G}$ be a subclass of \acp{ASHG} in which an $X$ partition always exists, and define $\mathcal{G}(n) \subseteq \mathcal{G}$ to be all games in the subclass with $n$ agents.
Denote by $G^m = (G_0,\dots, G_m)$ a sequence of $m+1$ ASHGs such that $G_{i+1}$ arises from $G_i$ by altering at most one valuation between two agents. Further, let $\pi_0$ be an $X$ partition in $G_0$.
We define the \emph{average distance} of $X$ as \[d_{X}^{\mathcal{G}}(n) = \lim_{m\to \infty} \max_{\substack{G^m \in \mathcal{G}^{m+1}(n) \\ \pi_0}}\min_{\pi^m}\frac{\sum_{i=1}^m d(\pi_{i-1}, \pi_{i})}{m},\] where
$\pi^m = (\pi_1, \dots, \pi_m)$ is a vector of $m$ partitions such that $\pi_i$ is $X$ 
in $G_i$. If $\mathcal{G}$ is clear from the context, we omit the superscript. 
First, we note that  $d_{X}^{\mathcal{G}}(n)$ is in fact well-defined.

\begin{restatable}{lemma}{existencelimes}
    Let $X$ be some stability notion, $n \in \mathbb{N}$, and let $\mathcal{G}$ be a subclass of ASHGs for which an $X$ partition always exists. Then, $d_{X}^{\mathcal{G}}(n)$ exists and is well-defined.
\end{restatable}

In the previous section, we saw changing a single valuation in a strict symmetric game might require the moving of many agents to uphold CNS stability if there are at least two negative valuation values. However, when considering a sequence of updates, \Cref{alg_closeCNS} allows us to find CNS-stable partitions for which we can bound the average distance.

 \begin{restatable}{theorem}{sequenceBoundCNS}
     
 \label{thm:CNS_average_distance}
    Let $G_0, \dots, G_m$ be a sequence of strict symmetric \acp{ASHG}, where $G_{i+1}$ arises from $G_i$ by altering a valuation between two agents, and let $\pi_0$ be a CNS partition in $G_0$. Then, \Cref{alg_closeCNS} computes partitions $\pi_1, \dots, \pi_m$, where $\pi_i$ is CNS in $G_i$, and
    \begin{align*}
        \sum_{i = 1}^{m} d(\pi_{i - 1}, \pi_{i}) \leq 4m + \phi(\pi_0) - \phi(\pi_m).
    \end{align*}
\end{restatable}

\begin{proof}
    By \Cref{thm:general_CNS_distance_bound}, each $\pi_i$ is CNS in $G_i$ and we have that 
    \begin{align*}
        \sum_{i = 1}^{m } d(\pi_{i-1}, \pi_{i}) &\leq 4m + \sum_{i = 1}^{m } \phi(\pi_{i-1}) - \phi(\pi_{i}) \\
        &= 4m + \phi(\pi_0) - \phi(\pi_m).
    \end{align*}
\end{proof}

Together with \Cref{thm:CIS_FEGS_polynomial_time}, we can show a constant bound on $d_X(n)$ for $X \in \{\text{CNS, CIS}\}$ in strict symmetric games.

\begin{theorem}
Let $\mathcal{G}$ be the family of strict symmetric \acp{ASHG}. Then, it holds that $d_{\text{CNS}}^{\mathcal{G}}(n) \leq 4$ and 
$d_{\text{CIS}}^{\mathcal{G}}(n) \leq 3$.
\end{theorem}

Since any CIS deviation is a Pareto improvement, we can prove that $d_{\text{CIS}}(n)$ is constant in FENGs in the non-symmetric case.

\begin{restatable}{theorem}{CISasymmetricPotential}
    
\label{theorem:CIS_FEG_amortized}
    In non-symmetric FENGs, it holds that $d_{\text{CIS}}(n) \leq 2$. 
\end{restatable}
\begin{proof}[Proof sketch]
A deviating agent $a$ increases her utility at least by $1$. Furthermore, since it is a CIS deviation, it does not decrease the utility of any other agent. On the other hand, a valuation update can decrease the social welfare by at most $2$. Since the social welfare of any partition always lies in $\{-n(n-1), \dots, n(n-1)\}$, the claim follows.
\end{proof}

In the symmetric case, this holds for all stability notions we consider since any Nash deviation in a symmetric FENG increases the social welfare \citep{BoJa02a}.

\begin{restatable}
{theorem}{thmsymfegs} \label{thm:potential}
    In symmetric FENGs, it holds that $d_X(n) \leq 2$ for $X \in \{\text{NS, IS, CNS, CIS}\}$.
\end{restatable}

Finally, we show that for unbounded valuations, the statement no longer holds.

\begin{restatable}{theorem}{ISUpDown}
    \label{IS-UpAndDown}
    In strict symmetric \acp{ASHG}, it holds that $d_X(n) \geq \lfloor \frac{n}{2}\rfloor$ for $X \in \{\text{NS, IS}\}$. 
\end{restatable}

\begin{proof}
First, suppose that $n$ is even, i.e., $n = 2\ell + 2$ for some $\ell \in \mathbb{N}$.
Define a strict symmetric \ac{ASHG} $G = (N, v)$ (see \Cref{fig:IS_updown_app}) with $N = X \cup Y \cup \{z_1, z_2\}$, where $\lvert X \rvert = \lvert Y \rvert  = \ell$ and the valuations
\begin{itemize}
    \item $v(x, x') = v(x, y) = v(y, y') = 1$ for all $x, x' \in X$ and $y, y' \in Y$
    \item $v(z_1, x) = v(z_2, y) = n$ for all $x \in X, y \in Y$,
    \item $v(z_1, y) = v(z_2, x) = 1$ for all $x \in X, y \in Y$,
    \item $v(z_1, z_2) = -n^2$.
\end{itemize}

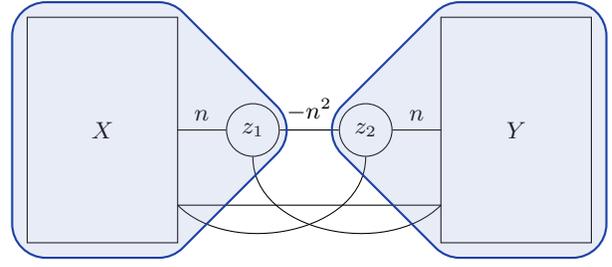
\begin{figure}[ht]
    \centering
\begin{tikzpicture}[>=stealth,
        every node/.style={font=\small},
        main/.style={circle,draw,minimum size=7mm,font=\small},
        rect/.style={draw, minimum width=2cm, minimum height=3cm, align=center}]
        
    \node[rect] (U)  at (-2.75,0) {$X$};
    \node[rect] (S*)  at (2.75, 0) {$Y$};
    
    \node[main] (z)  at (-0.75, 0) {$z_1$};
    \node[main] (z2) at (0.75, 0) {$z_2$};

    \begin{scope}[every node/.style={font=\small}]
                \path
                    (z) edge[-] node[above] {$-n^2$} (z2)
                    (z) edge[-, out = 270, in = 225] node[above] {} (1.75, -1)
                    (z2) edge[-, out = 270, in = -45] node[above] {} (-1.75, -1);
            \end{scope}     
            
    \draw[-] (z) -- node[above] {$n$} (-1.75, 0);
    \draw[-] (z2) -- node[above] {$n$} (1.75, 0);
    \draw[-] (-1.75, -1) -- node[above] {} (1.75, -1);

    \node (Xtl) at (-3.5, 1.25) {};
    \node (Xtr) at (-2., 1.25) {};
    \node (Xbl) at (-3.5, -1.25) {};
    \node (Xbr) at (-2., -1.25) {};
    \draw[thick,CoalitionColor, fill=CoalitionColor!50, fill opacity=0.2] 
    \convexpath{Xbr, Xbl, Xtl, Xtr, z}{.45cm};

    \node (Ytl) at (3.5, 1.25) {};
    \node (Ytr) at (2., 1.25) {};
    \node (Ybl) at (3.5, -1.25) {};
    \node (Ybr) at (2., -1.25) {};
    \draw[thick,CoalitionColor, fill=CoalitionColor!50, fill opacity=0.2] 
    \convexpath{Ybl, Ybr,z2, Ytr, Ytl}{.45cm};

\end{tikzpicture}
\caption{\label{fig:IS_updown_app}Illustration of $G$ in \Cref{IS-UpAndDown}. Rectangles are cliques with valuations of $1$, and edges without an edge weight correspond to valuations of $1$. Coalitions in $\pi$ are highlighted in blue.}
\end{figure}

We claim that $\pi = \{\{z_1\} \cup X, \{z_2\} \cup Y\}$ is the only IS partition in $G$. To see this, note that in any IS partition $\tilde{\pi}$, we have that $\tilde{\pi}(z_1) \neq \tilde{\pi}(z_2)$ since otherwise $u_{z_i}(\tilde{\pi}(z_i)) \leq -n^2 + \frac{n^2 - n - 2}{2} < 0$. Moreover, for any $x \in X$ and $y \in Y$, we have that $u_x(N \setminus \{z_1\}) = n - 2 < n = v(x, z_1)$ and $u_y(N \setminus \{z_2\}) = n - 2 < n = v(y, z_2)$. Since $v(z_1, z_2)$ is the only negative valuation in $G$, it follows that $\tilde{\pi} = \pi$.

Alter the valuation between agents $z_1$ and $z_2$, such that $v(z_1, z_2) = n^2$. Then, by an analogous argument, it follows that the grand coalition $\pi' = \{N\}$ is the only IS partition in $G'$. 

Since any NS partition is IS and every symmetric \ac{ASHG} admits at least one NS partition \citep{BoJa02a}, it follows that $\pi$ and $\pi'$ are the only NS partitions in $G$ and $G'$, respectively. 

By alternating between $G$ and $G'$, it follows that $n \geq d_{X}(n) = \frac{n}{2} \in \Theta(n)$ since $d(\pi, \pi') = \frac{n}{2}$.

If $n$ is odd, consider the same construction as before, but with an additional agent $\tilde{a}$ that has a negative valuation for every other agent. In any NS or IS partition this agent will be a singleton coalition, so the claim follows completely analogously to the case where $n$ is even.
\end{proof}

\section{Conclusion}
In this paper, we proved the NP-completeness of deciding whether a close stable partition exists in a symmetric \acp{ASHG} after an agent has changed her preferences for another agent for the stability notions NS, IS, CNS, CIS. 
This holds even if the valuations are restricted to $\{-1, 0, 1\}$. 
However, for such restricted valuations and sufficiently long sequences of valuation updates, one can compute a sequence of stable partitions such that the average distance is constant, which shows a clear contrast to the previous hardness result.

In strict symmetric \acp{ASHG}, finding a closest stable partition can be computed in polynomial time for CIS, and for CNS under the additional restriction that there is only one negative valuation value. Without this restriction, the problem is NP-complete, but for sequences of valuation updates of length $\Omega(n)$, our algorithm computes a sequence of CNS partitions for which the average distance is constant. 
For NS and IS, the problem remains NP-complete in symmetric \acp{FEG}, AFGs, and AEGs.

We analyzed the average distance $d_X(n)$ over the worst possible sequence of valuation updates, where $d_X(n)$ only depends on the considered stability notion and subclass of games. We show that $d_X(n) \in \Theta(n)$ for $X \in \{\text{NS, IS}\}$ in strict symmetric \acp{ASHG}. It is an interesting direction for future research to study how $d_X(n)$ behaves for other subclasses of \acp{ASHG} and stability notions.


\appendix

\section{Missing Proofs in Section \ref{sec:Results}}

In this appendix, we present the proofs of \Cref{sec:Results} missing from the main body of our paper.
\everythingInNOP*
\begin{proof}
    Computing the distance between two partitions can be done in polynomial time \citep{day1981complexity}. Moreover, one can check in polynomial time whether an agent exists that can make an $X$ deviation by iterating over all other coalitions in the partition.
\end{proof}

\lemmaDistanceK*
\begin{proof}
    We do an induction over $k$.
    Let $k = 1$: Every agent can join another coalition or form a new one. Since there exist at most $n$ coalitions and at most $n-1$ coalitions different from the one the agent is in, there exist at most $n$ deviations for every agent. This gives an upper bound of $n^2$ for the total number of partitions with distance $1$, concluding the induction base.
    Let the statement hold for some $k$. Now consider all structures with distance $k+1$.
    By definition, there exists a path of length $k+1$.
    By induction hypothesis, there are at most $n^{2k}$ partitions with distance $k$, and from each of them at most $n^{2}$ partitions can be reached. 
    Therefore, in total, we can upper bound the number of partitions with distance $k+1$ by $n^{2k+2}$.
\end{proof}

For the NP reductions, it is convenient to argue which agents have moved between two partitions $\pi, \pi'$. Nevertheless, from a given starting partition $\pi$, a partition $\pi'$ can be reached through different paths. Consider the following example with $3$ agents and partitions $\pi = \{\{1,2\},\{3\}\}$, as well as $\pi' = \{\{1,2,3\}\}$. Then, both $1$ and $2$ can move to the coalition containing $3$, or agent $3$ can move to the coalition containing $1$ and $2$. Observe that $d(\pi,\pi') = 1$ and therefore, moving $1$ and $2$ would not be on a shortest path. For our reductions, we are interested in deviations on shortest paths since we are interested in the distance between $\pi$ and $\pi'$.

\thmeverythingISNPhard*

    We first prove the result for \textsc{X-1-1-Sym-Altered}. Afterwards, we explain how to adapt the proof for \textsc{X-1-1-Altered}.
    
    \begin{proof}
    We first prove the statement for \textsc{X-1-1-Sym-Altered} and the case $\beta(n) \geq \alpha(n) \geq 0$ by a reduction from the set cover problem for all $X \in \{\text{IS},\text{CIS},\text{NS}, \text{CNS}\}$. Then, we argue that the same reduction works for general $\alpha(n)$ and $\beta(n)$ for IS and CIS.
    
    Let $(E, \mathcal{S}, k)$ be a set cover instance. Without loss of generality, we can assume that $k < \lvert E \rvert, \lvert \mathcal{S} \rvert$ and that $\mathcal{S}$ is a set cover of $E$, as otherwise the problem is trivial.
    Let $G$ be a symmetric \ac{ASHG} (see \Cref{fig:thmOneGeneral}) with agent set $N = E \cup \mathcal{S} \cup Y \cup \{z_1, z_2\}$, where $\lvert Y \rvert = \lvert E \rvert + \lvert \mathcal{S} \rvert + 3$, and valuations in $\{-\beta, 0, \alpha\}$ with $\beta = \beta(\lvert N \rvert)$ and $\alpha = \alpha(\lvert N \rvert)$ defined as follows:
    \begin{itemize}
        \item $v(e, S) = -\beta$ for all $S \in \mathcal{S}$ and $e \in S$,
        \item $v(z_1, e) = \alpha$ for all $e \in E$,
        \item $v(z_2, e) = v(z_2, S) = -\beta$ for all $S \in \mathcal{S}$ and $e \in E$,
        \item $v(z_1, y) = v(z_1, z_2) = \alpha$ for all $y \in Y$,
        \item $0$ for all other valuations.
    \end{itemize}

    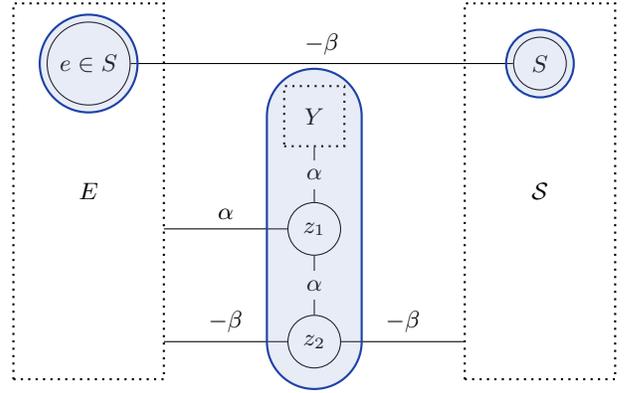
\begin{figure}[ht]
    \centering
\begin{tikzpicture}[>=stealth,
        every node/.style={font=\small},
        main/.style={circle,draw,minimum size=7mm,font=\small},
        rect/.style={draw, thick, dotted, minimum width=2cm, minimum height=5cm, align=center}]
        
    \node[rect] (U)  at (-3.75,0) {$E$};
    \node[rect] (S*)  at ( 2.25,0) {$\mathcal{S}$};
    
    \node[rect, minimum width = 0.8cm, minimum height = 0.8cm, ] (y)  at (-0.75, 1)   {$Y$};
    \node[main] (z)  at (-0.75, -0.5)   {$z_1$};
    \node[main] (zp) at (-0.75,-2)   {$z_2$};

    \node[main] (u) at (-3.75, 1.7) {$e \in S$};

    \node[main] (S) at (2.25, 1.7) {$S$};

    \begin{scope}[every node/.style={font=\small}]
                \path
                    (y) edge[-] node[pos = 0.5, fill = white] {$\alpha$} (z)
                    (z)  edge[-] node[pos = 0.5, fill = white] {$\alpha$} (zp)
                    (u) edge[-] node[above] {$-\beta$} (S);
            \end{scope}     
            
    \draw[-] (z) -- node[above] {$\alpha$} (-2.75, -0.5);
    \draw[-] (zp) -- node[above] {$-\beta$} (1.25, -2);
    \draw[-] (zp) -- node[above]  {$-\beta$} (-2.75, -2);

    \draw[thick,CoalitionColor, fill=CoalitionColor!50, fill opacity=0.2]  \convexpath{y,zp}{.63cm};
    \node (u2) at (-3.75, 1.70001) {};
    \node (S2) at (2.25, 1.70001) {};
    \draw[thick,CoalitionColor, fill=CoalitionColor!50, fill opacity=0.2]  \convexpath{u,u2}{.65cm};
    \draw[thick,CoalitionColor, fill=CoalitionColor!50, fill opacity=0.2]  \convexpath{S,S2}{.45cm};
\end{tikzpicture}
\caption{Illustration of $G$ from the reduction in \Cref{symmetric_Single_Hardness}. Dotted rectangles are independent sets. All edges that are not shown correspond to a valuation of $0$. Coalitions in $\pi$ are highlighted in blue.}
\label{fig:thmOneGeneral}
\end{figure}
    We claim that $\pi = \left\{ \{e\}_{e \in E}, \{S\}_{S \in \mathcal{S}}, Y \cup \{z_1, z_2\} \right\}$ is an $X$ partition in $G$.
    To see that this is the case, we perform a case analysis to show that no type of agent can make an $X$ deviation:
    \begin{itemize}
        \item The only agent for which any agent $e \in E$ has a positive valuation is $z_1$. However, it holds that $u_e(\pi(z_1)) = \alpha - \beta \leq 0$, so $e$ does not want to deviate to another coalition.
        \item Any agent $S \in \mathcal{S}$ does not have any positive valuations. Therefore, she cannot increase her utility by deviating since $u_S(\pi(S)) = 0$.
        \item Agent $z_1$ has a utility of $\alpha (\lvert Y \rvert+ 1) > \alpha \lvert N\setminus Y\rvert$ for $\pi(z_1)$, so she cannot deviate.
        \item No agent in $Y \cup \{z_2\}$ has an incentive to deviate since each has a utility of at least $\alpha$ in their current coalition and for all other coalitions they have a utility of at most $0$.
    \end{itemize}
    Now, alter the valuation between agent $z_1$ and $z_2$ such that $v'(z_1, z_2) = -\beta$ and denote the altered game by $G'$. We claim that a set cover $\mathcal{C} \subseteq \mathcal{S}$ of $E$ with $\lvert \mathcal{C} \rvert \leq k$ exists if and only if $G'$ has an $X$ partition $\pi'$ with $d(\pi, \pi') \leq k + 1$.
    To this end, suppose that $\mathcal{C} \subseteq \mathcal{S}$ is a set cover of $E$ with $\lvert C \rvert \leq k$. We claim that $\pi' = \left\{ \{e\}_{e \in E}, \{S\}_{S \in \mathcal{S \setminus C}}, \mathcal{C} \cup Y \cup \{z_1\}, \{z_2\} \right\}$ is $X$ in $G'$. Clearly, it holds that $d(\pi, \pi') \leq k+1$. To see that $\pi'$ is an $X$ partition in $G'$, we again perform a case analysis to show that no agent can make an $X$ deviation.
    \begin{itemize}
        \item The only agent for which $e \in E$ has a positive valuation is $z_1$. However, since there is some $S \in \pi(z_1)$ with $v(e, S) = -\beta$, it holds that $u_e(\pi'(z_1)) \leq \alpha - \beta \leq 0$, so $e$ cannot deviate.
        \item An agent $S \in \mathcal{S}$ only has non-positive valuations. Therefore, she does not want to deviate since $u_S(\pi'(S)) = 0$.
        \item An agent $y \in Y$ has utility $\alpha$ for $\pi'(y)$, so she cannot deviate since she has a nonpositive valuation for any agent outside her coalition.
        \item Agent $z_1$ has a utility of $\alpha\lvert Y \rvert > \alpha \lvert N \setminus Y \rvert$ for $\pi'(z_1)$, so she cannot deviate.
        \item Agent $z_2$ does not want to deviate because $u_{z_2}(\pi'(z_2)) = 0$ and $v'(z_2, a) \leq 0$ for all $a \in N$.
    \end{itemize}
    To prove the reverse direction, suppose that $\pi'$ is an $X$ partition in $G'$ with $d(\pi, \pi') \leq k + 1$. Note that since $v(y, z_1) = \alpha$ and $v(y, a) = 0$ for all $y \in Y$ and $a \in N \setminus \{z_1\}$, we have that $Y \subseteq \pi'(z_1)$. Moreover, $z_2$ is in a singleton coalition in $\pi'$ since $v'(z_2, a) < 0$ for all $a \in N \setminus Y$.
    Let $R_E$ and $R_{\mathcal{S}}$ denote the agents in $E$ and $\mathcal{S}$ that were moved to another coalition from $\pi$. Note that it suffices to only consider these agents since $\lvert Y \rvert > \lvert N \setminus Y \rvert, Y \subseteq \pi(z_1), \pi'(z_1)$, and no agent in $Y$ has been moved in a shortest sequence of changes. 
    Let $f: E \rightarrow \mathcal{S}$ map each $e \in E$ to some $S \in \mathcal{S}$ with $e \in S$. We claim that $\mathcal{C} := f(R_E) \cup R_{\mathcal{S}}$ is a set cover of $E$ with $\lvert \mathcal{C}\rvert \leq k$. By assumption, it holds that $\lvert R_E \rvert + \lvert R_{\mathcal{S}} \rvert \leq k$ since $z_2$ also moved to another coalition. To see that $\mathcal{C}$ is a set cover of $E$, assume for contradiction that there is some $e \in E$ that is not covered by $\mathcal{C}$. Hence, by construction of $\mathcal{C}$, there is no $S \in \pi'(z_1)$ with $v(e, S) < 0$ and $\pi'(z_1) \neq \pi'(e)$. Therefore, $e$ can $X$ deviate to $\pi'(z_1)$ since there is no $a \in \pi'(e)$ with $v(e, a) > 0$ and no $b \in \pi'(z_1)$ with $v(e, b) < 0$ (remember that $z_2$ is the only agent $e$ has a negative utility for but as argued above $z_2$ has to be in a singleton coalition). This is a contradiction to $\pi'$ being $X$ in $G'$, so $\mathcal{C}$ is in fact a set cover of $E$.

\textbf{General $\alpha(n)$ and $\beta(n)$ for IS and CIS:}
It remains to prove the case that $\alpha (n) > \beta (n)$ for $X \in \{\text{IS},\text{CIS}\}$. We can use the same set cover reduction as above. 

First, we show that $\pi$ is an IS (or CIS) partition in $G$. The arguments for any agent $S \in \mathcal{S}$ and for all agents in $\{y, z_1, z_2\}$ are analogous as above. For $e \in E$,
note that any coalition other than $\pi(e)$ contains at least one agent with a negative valuation for $e$. Hence, $e$ cannot make an IS (or CIS) deviation.

Next, we show that the partition \\$\pi' = \left\{ \{e\}_{e \in E}, \{S\}_{S \in \mathcal{S \setminus C}}, \{y, z_1\} \cup \mathcal{C}, \{z_2\} \right\}$ is IS (or CIS) in $G'$.
Again, the only case we need to argue differently from above is for the agents $e \in E$. The only coalition $e$ can deviate to increase her utility is the coalition $\pi(z_1)$. But, by assumption, $\mathcal{C}$ is a set cover, so there is some $S \in \mathcal{C} \cap \pi(z_1)$ such that $e \in S$. By construction, $v(e, S) = - \beta < 0$, so $e$ cannot make an IS (or CIS) deviation.
The proof of the reverse direction is exactly the same as above.

For the non-symmetric case, the proof works analogously. We use exactly the same construction for $G$ with the only exception that the valuations between $z_1$ and $z_2$ are not symmetric anymore and set to
$v_{z_1}(z_2) = 0$ and $v_{z_2}(z_1) = \alpha$. After the valuation of $z_2$ for $z_1$ is altered such that $v_{z_2}(z_1) = - \beta$, an analogous argument as in the symmetric case applies.
\end{proof}

\approxIsHard*
\begin{proof}
    In the reduction for Theorem \ref{symmetric_Single_Hardness}, we showed that every set cover of size $\ell$ corresponds to a solution of \textsc{X-1-1-Altered} with distance at most $\ell +1$ and vice versa.
    Without loss of generality, we can assume $\ell >1$ since otherwise both problems are trivially in $\text{P}$.
    So if we could approximate \textsc{X-1-1-Altered} by any function $f(\ell)$, we could also approximate set cover with a function $2f(\ell)$. However, since there is no efficient algorithm that approximates set cover better than by a logarithmic factor unless $\text{P} = \text{NP}$
    \citep{hardnessSetCover}, we conclude that the same holds for \textsc{X-1-1-Altered}.

\end{proof}

\section{Missing Proofs in Section \ref{sec:StrictASHGs}}

\subsection{CNS and CIS}

In this appendix, we present the proofs of \Cref{sec:StrictASHGs} missing from the main body of our paper.

\theoremOnlyTwoWeights*

\begin{proof}
    Here, we prove the claim for \textsc{X-1-N-Sym-Altered}. The claim for \textsc{X-1-N-Altered} (and \textsc{X-N-1-Altered}) follows from \Cref{thm:CISandCNS_single_asymmetric_NP_complete}.
    
     Let $(E, \mathcal{S}, k)$ be a set cover instance. Without loss of generality, we can assume that $1 < k < \lvert E \rvert - 1, \lvert \mathcal{S} \rvert - 1$ and that $\mathcal{S}$ is a set cover of $E$, as otherwise the problem is trivial.
    Let $G$ be a symmetric \ac{ASHG} with agent set $N = E \cup \mathcal{S} \cup \{z\}$ and valuations defined by
    \begin{itemize}
        \item $v(e, S) = \alpha$ for all $S \in \mathcal{S}$ and $e \in S$,
        \item $v(e, z) = \alpha$ for all $e \in E$,
        \item $v(S, S') = \alpha$ for all $S, S' \in \mathcal{S}$,
        \item $-\beta$ for all other valuations,
    \end{itemize}
    where we set $\alpha := \alpha(\lvert N \rvert)$ and $\beta := \beta(\lvert N \rvert)$. Since $\alpha$ and $\beta$ are polynomial-time computable functions, $G$ is polynomial in the encoding size of $(E, \mathcal{S})$.
    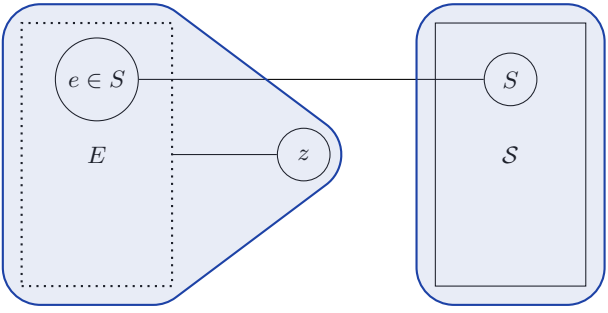
\begin{figure}[ht]
    \centering
\begin{tikzpicture}[>=stealth,
        every node/.style={font=\small},
        main/.style={circle,draw,minimum size=7mm,font=\small},
        rect/.style={draw, minimum width=2cm, minimum height=3.5cm, align=center}]
        
    \node[rect, thick, dotted] (U)  at (-3.5,1) {$E$};
    \node[rect] (S*)  at ( 2, 1) {$\mathcal{S}$};
    
    \node[main] (z)  at (-0.75, 1)   {$z$};

    \node[main] (u) at (-3.5, 2) {$e \in S$};

    \node[main] (S) at (2, 2) {$S$};

    \begin{scope}[every node/.style={font=\small}]
                \path
                    (u) edge[-] node[above] {} (S);
            \end{scope}     
            
    \draw[-] (z) -- node[above] {} (-2.5, 1);

    \node (Utl) at (-4.25, 2.5) {};
    \node (Utr) at (-2.75, 2.5) {};
    \node (Ubl) at (-4.25, -0.5) {};
    \node (Ubr) at (-2.75, -0.5) {};
    \draw[thick,CoalitionColor, fill=CoalitionColor!50, fill opacity=0.2] 
    \convexpath{Ubr, Ubl, Utl, Utr, z}{.5cm};

    \node (Stl) at (1.25, 2.5) {};
    \node (Str) at (2.75, 2.5) {};
    \node (Sbl) at (1.25, -0.5) {};
    \node (Sbr) at (2.75, -0.5) {};
    \draw[thick,CoalitionColor, fill=CoalitionColor!50, fill opacity=0.2] 
    \convexpath{Sbr, Sbl, Stl, Str}{.5cm};

\end{tikzpicture}
\caption{Illustration of $G$ from the reduction in \Cref{theorem:altered_complete}. Undotted rectangles are cliques, and dotted rectangles are independent sets. There is an edge $\{i, j\}$ if and only if $v(i, j) = \alpha$. Coalitions in $\pi$ are highlighted in blue.}
\end{figure}
    
    Clearly, $\pi = \{E \cup\{z\}, \mathcal{S}\}$ is CNS (CIS) since every agent has a (symmetric) positive valuation for another agent in her coalition. 
    Alter the valuations of agent $z$, such that $v(z, e) = -\beta$ for all $e \in E$ and denote the altered game by $G'$. We claim that a set cover $\mathcal{C \subseteq S}$ of $E$ with $\lvert \mathcal{C} \rvert \leq k$ exists if and only if $G'$ has a CNS (CIS) partition $\pi'$ with $d(\pi, \pi') \leq k +1$.
    To this end, suppose that $\mathcal{C \subseteq S}$ is a set cover of $E$ with $\lvert \mathcal{C} \rvert \leq k$. Let $\pi' = \{E \cup \mathcal{C}, \mathcal{S \setminus C}, \{z\}\}$. Clearly, it holds that $d(\pi, \pi') \leq k + 1$. To see that $\pi'$ is CNS (CIS), note that for any agent in $a \in N\setminus \{z\}$, there is some $a' \in \pi(a)$ with $v(a, a') > 0$ since $1 < k < \lvert E \rvert - 1, \lvert \mathcal{S} \rvert - 1$. Hence, $a$ cannot make a CNS (CIS) deviation. Finally, $z$ cannot deviate since it has a negative valuation for any other agent. 

    To prove the reverse direction, suppose that $\pi'$ is CNS (CIS) in $G'$ with $d(\pi, \pi') \leq k + 1$. Let $R_E$ and $R_{\mathcal{S}}$ denote the agents in $E$ and $\mathcal{S}$ that were moved from their coalition in $\pi$ to another. Let $f: E \rightarrow \mathcal{S}$ map each $e \in E$ to some $S \in \mathcal{S}$ with $e \in S$. We claim that $\mathcal{C} := f(R_E) \cup R_{\mathcal{S}}$ is a set cover of $E$ with $\lvert \mathcal{C} \rvert \leq k$.
    It holds that $\pi'(z) = \{z\}$ since $z$ has a negative valuation for any other agent in $G'$. Thus, $\lvert R_E \rvert + \lvert R_{\mathcal{S}}\rvert \leq k$. To see that $\mathcal{C}$ is a set cover of $E$, suppose that there is some $e \in E$ that is not covered by $\mathcal{C}$. By construction of $\mathcal{C}$, there is no $S \in \pi'(e)$ with $v(e, S) > 0$. Hence, $e$ does not have a positive valuation for any $a \in \pi'(e)$, so $e$ can deviate to the empty coalition since $u_e(\pi'(e)) \leq -\beta < 0$. However, this is a contradiction to $\pi'$ being CNS (CIS). Therefore, $\mathcal{C}$ is in fact a set cover of $E$.
\end{proof}

\CNSSym*
\begin{proof} 
     We first prove hardness for \textsc{CNS-1-1-Sym-Altered}. At the end of the proof, we explain how to adapt the reduction for the non-symmetric case.
     
     Let $(E, \mathcal{S}, k)$ be a set cover instance. Without loss of generality, we can assume that $2 < k < \lvert E \rvert -1, \lvert \mathcal{S} \rvert - 1$ and that $\mathcal{S}$ is a set cover of $E$, as otherwise the problem is trivial.
     We create a symmetric \ac{ASHG} $G$ with valuations restricted to $\{-2M, -1, 1\}$, where $M  = \lvert E \rvert + \lvert \mathcal{S} \rvert$ (see \Cref{fig:CNSSYMHARD}). To this end, let $N = E \cup \mathcal{S} \cup \{z_1, z_2\} \cup Y$ with $\lvert Y \rvert = M$ and define the valuations by
    \begin{itemize}
        \item $v(e, S) = -2M$ for all $S \in \mathcal{S}$ and $e \in S$,
        \item $v(S, S') = 1$ for all $S, S' \in \mathcal{S}$,
        \item $v(y, e) = 1$ for all $e \in E$ and $y \in Y$,
        \item $v(y, y') = 1$ for all $y, y' \in Y$,
        \item $v(z_1, e) = -2M$ for all $e \in E$,
        \item $v(z_1, z_2) = 1$,
        \item $-1$ for all other valuations.
    \end{itemize}

    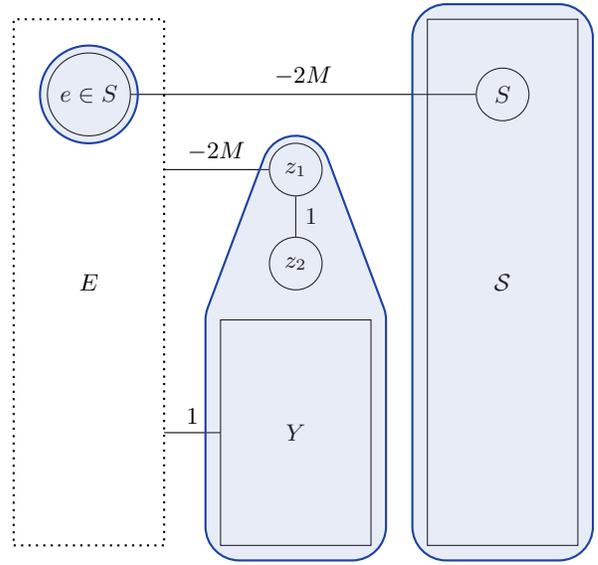
\begin{figure}[ht]
    \centering
\begin{tikzpicture}[>=stealth,
        every node/.style={font=\small},
        main/.style={circle,draw,minimum size=7mm,font=\small},
        rect/.style={draw, minimum width=2cm, minimum height=7cm, align=center}]
        
    \node[rect, thick, dotted] (U)  at (-3.5,0) {$E$};
    \node[rect] (S*)  at ( 2,0) {$\mathcal{S}$};
    
    \node[main] (z)  at (-0.75, 1.5)   {$z_1$};
    \node[main] (zp)  at (-0.75, 0.25)   {$z_2$};

    \node[rect, minimum width = 2cm, minimum height = 3cm] (X) at (-0.75, -2) {$Y$};

    \node[main] (u) at (-3.5, 2.5) {$e \in S$};

    \node[main] (S) at (2, 2.5) {$S$};

    \begin{scope}[every node/.style={font=\small}]
                \path
                    (z)  edge[-] node[right] {$1$} (zp)
                    (u) edge[-] node[above] {$-2M$} (S);
            \end{scope}     
            
    \draw[-] (z) -- node[above] {$-2M$} (-2.5, 1.5);
    \draw[-] (-1.75, -2) -- node[above] {$1$} (-2.5, -2);

    \node (Xtr) at (0, -0.5) {};
    \node (Xtl) at (-1.5, -0.5) {};
    \node (Xbr) at (0, -3.25) {};
    \node (Xbl) at (-1.5, -3.25) {};
    \draw[thick,CoalitionColor, fill=CoalitionColor!50, fill opacity=0.2] 
    \convexpath{z, Xtr, Xbr, Xbl, Xtl}{.45cm};

    \node (u2) at (-3.5, 2.5001) {};
    \draw[thick,CoalitionColor, fill=CoalitionColor!50, fill opacity=0.2] 
    \convexpath{u2, u}{.65cm};

    \node (Str) at (2.75, 3.25) {};
    \node (Stl) at (1.25, 3.25) {};
    \node (Sbr) at (2.75, -3.25) {};
    \node (Sbl) at (1.25, -3.25) {};
    \draw[thick,CoalitionColor, fill=CoalitionColor!50, fill opacity=0.2] 
    \convexpath{Str, Sbr, Sbl, Stl}{.45cm};
\end{tikzpicture}
\caption{Illustration of $G$ from the reduction in \Cref{CNS_hardness_without_0_edges}. Undotted rectangles are cliques with valuations of 1. Dotted rectangles are independent sets. All edges that are not shown correspond to a valuation of $-1$. Coalitions in $\pi$ are highlighted in blue.}
\label{fig:CNSSYMHARD}
\end{figure}
    
    Note that $\pi = \{\{e\}_{e \in E}, \mathcal{S}, \{z_1, z_2\} \cup Y\}$ is CNS since every agent in a non-singleton coalition has at least one (symmetric) positive valuation for some other agent in her coalition. Moreover, any agent in a singleton coalition, namely any $e \in E$, has a valuation of $-2M$ to at least one agent in any nonsingleton coalition, and a valuation of $-1$ to any other agent in a singleton coalition. Therefore, she cannot deviate to increase her utility.
    
    Alter the valuation between agents $z_1$ and $z_2$ such that $v'(z_1, z_2) = -1$ and denote the altered game by $G'$. We claim that a set cover $\mathcal{C \subseteq S}$ of $E$ with $ \lvert \mathcal{C} \rvert \leq k$ exists if and only if $G'$ has a CNS partition $\pi'$ with $d(\pi, \pi') \leq k + 2$. To this end, let $\mathcal{C \subseteq S}$ denote such a set cover. Let $\pi' = \{ \{e\}_{e \in E}, \mathcal{S \setminus C}, \mathcal{C} \cup Y, \{z_1\}, \{z_2\} \}$. Clearly, it holds that $d(\pi, \pi') \leq k + 2$. To see that $\pi'$ is CNS in $G'$, we perform a case analysis to show that no agent can make a CNS deviation. 

    \begin{itemize}
        \item Let $e$ be any agent in $E$. Since $\mathcal{C}$ is a set cover of $E$, there is at least one agent $S\in \mathcal{C} \cup Y$ with $e \in S$ and $v(e, S) = -2M$. Moreover, any deviation can only decrease her utility since $e$ has a negative valuation for any agent in a coalition other than $\mathcal{C} \cup Y$. Thus, she cannot make a CNS deviation. 
        \item For any $S \in \mathcal{S}$, there is some $S' \in  \pi'(S) \cap \mathcal{S}$ with $v(S, S') = 1$ since by assumption $2 < k < \lvert \mathcal{S} \rvert - 1$. Hence, $S$ cannot make a CNS deviation.
        \item An agent $y \in Y$ cannot make a CNS deviation because there is some $y' \in \pi'(y) \cap Y$ with $v(y, y') = 1$. 
        \item Agents $z_1$ and $z_2$ have negative valuations for all other agents in $G'$, so they cannot make a CNS deviation.
    \end{itemize}

    To prove the reverse direction, suppose that $\pi'$ is a CNS partition in $G'$ with $d(\pi, \pi') \leq k + 2$. Note that both $z_1$ and $z_2$ must be in singleton coalitions in $\pi'$ since they only have negative valuations. Let $R_E, R_{\mathcal{S}}$, and $R_Y$ denote the agents in $E, \mathcal{S}$, and $Y$ that have been moved to another coalition from $\pi$.
    
    Moreover, let $f: E \rightarrow \mathcal{S}$ map each $e \in E$ to some $S \in \mathcal{S}$ with $e \in S$, and let $g: R_Y \rightarrow \mathcal{S}$ be a mapping defined by 
    \begin{align*}
        y \rightarrow \begin{cases}
            f(e) & \text{if } \exists e \in \pi'(y) \cap (E \setminus R_E),\\
            \text{some } S \in \mathcal{S} & \text{otherwise.}
        \end{cases}
    \end{align*}
     Since each $e \in E$ is in a singleton coalition in $\pi$, $g$ is well-defined. We claim that $\mathcal{C} := f(R_E) \cup R_{\mathcal{S}} \cup g(R_Y)$ 
     is a set cover of $E$ with $\lvert \mathcal{C} \rvert \leq k$. By assumption, it holds that $\lvert R_E \rvert + \lvert R_{\mathcal{S}} \rvert + \lvert R_Y \rvert \leq k$ since $z_1$ and $z_2$ have moved to singleton coalitions. To see that $\mathcal{C}$ is a set cover of $E$, suppose that there is some $e \in E$ that is not covered by $\mathcal{C}$. To this end, denote the coalition in $\pi'$ that contains the most agents in $Y$ by $\pi'(y^*)$. 
     By construction of $\mathcal{C}$, there is no $S \in \pi'(y^*)$ with $e \in S$, and there is no $y \in \pi'(e)$. Hence, there is no agent $ a \in \pi'(e)$ with $v(a, e) > 0$. Moreover, $e$ has no $-2M$ valuation to any agent in $\pi'(y^*)$, a valuation of $+1$ to at least $\lvert Y \rvert - k = M - k$ many agents in $\pi'(y^*)$, and a valuation of $-1$ to at most $k$ many agents. Therefore, $e$ can make a CNS deviation since $u_e(\pi'(y^*)) > 0$ in $G'$. However, this is a contradiction to $\pi'$ being CNS in $G'$, so $\mathcal{C}$ is in fact a set cover of $E$.

     For the non-symmetric case, instead of $z_1$ and $z_2$, consider the three auxiliary agents $z_1, z_2,$ and $z_3$ with valuations $v_{z_i}(z_{i+1\mod 3}) = 1$ and otherwise $-1$. All other valuations are the same as in the symmetric case. After altering the valuation of agent $z_3$ for $z_1$ such that $v'_{z_3}(z_1) = -1$, the same argument as in the symmetric case applies.
    We refer to the reduction in the proof of \Cref{thm:CISandCNS_single_asymmetric_NP_complete} for the case $X = \text{CIS}$, where a similar construction is used.
\end{proof}

\CNSFourTheorem*

We divide the proof of Theorem \ref{thm:general_CNS_distance_bound} into multiple lemmas. In them, we consider CIS partitions. Since every CNS partition is also a CIS partition, the statements hold for CNS partitions as well.

\begin{lemma}\label{lemma:CNS_blocking_agents_exist}

    Let $G$ be a strict symmetric \ac{ASHG} and let $\pi$ be a CIS partition in $G$. Then, for any $x \in N$ with $\lvert \pi(x) \rvert \geq 2$, there is some $y \in \pi(x)$ with $v(x, y) > 0$. 
\end{lemma}
\begin{proof}

    If there does not exist any agent $y \in \pi(x)$ such that $v(x, y) > 0$, then it holds that $u_x(\pi(x)) < 0$ since the game is strict and symmetric. Hence, agent $x$ could CNS/CIS deviate to the empty coalition, contradicting that $\pi$ is CIS.
\end{proof}

\begin{lemma}\label{lemma:singletons_negative_valuations}
    Let $G$ be a strict symmetric \ac{ASHG} and let $\pi$ be a CIS partition in $G$. If $x$ is in a singleton coalition, i.e., $\pi(x) = \{x\}$, then for any $C \in \pi\setminus \{ \pi(x) \}$, there is some $y \in C$ with $v(x, y) < 0$. Moreover, if $\pi$ is CNS, then it also holds that $u_x(C) \leq 0$.
\end{lemma}

\begin{proof}
    Because $\pi$ is CIS, there must be some $y \in C$ with $v(x, y) < 0$ for all $C \in \pi \setminus \{\pi(x)\}$ as otherwise $x$ could make a CIS deviation to $C$ since the game is strict.
    If $\pi$ is additionally CNS, it must also hold that $u_x(C) \leq 0$ since otherwise $x$ could make a CNS deviation to $C$.
\end{proof}

\begin{lemma}\label{lemma:CNS_different_coalitions}
    Let $G$ be a strict symmetric \ac{ASHG} and let $\pi$ be a CNS (or CIS) partition in $G$. Suppose that the valuation between agents $a$ and $b$ was altered and that $\pi(a) \neq \pi(b)$. Denote the altered game by $G'$.
    If $\pi$ is not CNS (or CIS) in $G'$, then either $\pi(a) = \{a\}$ or $\pi(b) = \{b\}$. Moreover, in this case, it holds that $\pi' = \left( \pi \setminus \{\pi(a), \pi(b)\} \right) \cup \{\pi(a) \cup \pi(b)\}$ is CNS (or CIS) in $G'$.
\end{lemma}

\begin{proof}
    Suppose that $\pi$ is not CNS (or CIS) in $G'$. Then, either $a$ can CNS (or CIS) deviate to $\pi(b)$ or $b$ can CNS (or CIS) deviate to $\pi(a)$ since only the valuation between $a$ and $b$ was altered. By \Cref{lemma:CNS_blocking_agents_exist}, it must hold that either $\pi(a) = \{a\}$ or $\pi(b) = \{b\}$. Without loss of generality, we may assume that $\pi(a) = \{a\}$.
    We claim that $\pi' = \left( \pi \setminus \{\pi(a), \pi(b)\} \right) \cup \{\pi(a) \cup \pi(b)\}$ is CNS (or CIS) in $G'$. 
    By \Cref{lemma:CNS_blocking_agents_exist}, an agent $y \in N\setminus \{a, b\}$ with $\lvert \pi'(y) \rvert \geq 2$ cannot make a CNS (or CIS) deviation. If $\pi(b) = \{b\}$, it holds that $v'(a, b) > 0$. If $ \lvert \pi(b) \rvert \geq 2$, there is some $z_b \in \pi(b)$ with $v(b, z_b) > 0$. In both cases, agent $b$ cannot make a CNS (or CIS) deviation. Agent $a$ cannot deviate because $u_a(\pi'(a)) \geq 0$ in $G'$, so there must be some $z_a$ with $v(a, z_a) > 0$.
    Finally, to show that no agent in a singleton coalition can make a CNS (or CIS) deviation, we make a case distinction. To this end, let $x \in N$ be an agent with $\pi'(x) = \{x\}$.
    
    \textbf{Case 1:} $\pi$ is CNS in $G$.
    By \Cref{lemma:singletons_negative_valuations}, it holds that $u_x(\pi(a)) \leq 0$ and $u_x(\pi(b)) \leq 0$, so $u_x(\pi'(a)) \leq 0$. Again, by \Cref{lemma:singletons_negative_valuations}, it holds that $u_x(C) \leq 0$ for all $C \in \pi' \setminus \{ \pi'(x)\}$. Hence, $x$ cannot make a CNS deviation.

    \textbf{Case 2:} $\pi$ is CIS in $G$. By \Cref{lemma:singletons_negative_valuations}, there it holds that $v(a, x) < 0$, so for any $C \in \pi' \setminus \{\pi'(x)\}$, there is some $y \in C$ with $v(x, y) < 0.$ Hence, $x$ cannot make a CIS deviation.
\end{proof}

\begin{lemma}\label{lemma:only_singletons_deviate}
     In strict symmetric \acp{ASHG}, the only agents that can make a CNS deviation in \Cref{alg_closeCNS} are $a, b$, and agents that are in singleton coalitions in $\pi$.
     Moreover, any such agent can make at most one CNS deviation. Agents $a$ and $b$ can only make a CNS deviation if their coalition has not been merged with a singleton coalition. 
\end{lemma}

\begin{proof}
    By \Cref{lemma:CNS_different_coalitions}, we only need to consider the case where $\pi(a) = \pi(b)$.
    Let $x \in N$ be an agent with $\lvert \pi(x)\rvert \geq 2$ and $x \neq a, b$. Then, by \Cref{lemma:CNS_blocking_agents_exist}, this agent cannot make a CNS deviation at any point in \Cref{alg_closeCNS}. Therefore, only $a, b$ or agents in singleton coalitions can make a CNS deviation in \Cref{alg_closeCNS}.

    Note that, by \Cref{lemma:singletons_negative_valuations}, an agent $x$ that deviates from a singleton coalition, deviates to a non-singleton coalition. 
    Moreover, there is at least one agent $y$ in her new coalition with $v(x, y) > 0$ since otherwise $x$ would not CNS deviate to this coalition. Therefore, neither $x$ nor $y$ can make another CNS deviation.

    Now, consider agents $a$ and $b$.
    If $\pi(a)$ is merged with $\pi(s_a)$, then neither $a$ nor $s_a$ can make a CNS deviation since $v(a, s_a) > 0$.
    If $\pi(a)$ is not merged and $a$ can make a CNS deviation, it deviates to the largest coalition $C_a$ for which it has a positive utility.
    If $C_a \neq \emptyset$, there is some $c_a \in C_a$ with $v(a, c_a) > 0$, so $a$ cannot deviate again.
    Otherwise, $a$ has a negative utility for all other coalitions in $\pi$ and it holds that $v(a, b) < 0$. Since the only other agents that can make CNS deviations are $b$ and agents in singleton coalitions in $\pi$, $a$ cannot make another CNS deviation.
    The same analysis holds for agent $b$ as well.
 
    Finally, note that by \Cref{lemma:singletons_negative_valuations}, any agent $x$ in a singleton coalition can only deviate to a coalition $C$ with $\lvert C \rvert \geq 2$. There is some $y \in C$ with $v(x, y) > 0$, so $x$ cannot deviate again. 
\end{proof}

\begin{proof}[Proof of \Cref{thm:general_CNS_distance_bound}]
By \Cref{lemma:only_singletons_deviate}, there are at most $\phi(\pi) + 2$ many CNS deviations in \Cref{alg_closeCNS}. Hence, the algorithm terminates after $\mathcal{O}(n)$ deviations. Moreover, since the algorithm terminates when no more CNS deviations are possible, the resulting partition $\pi'$ is CNS. Since each deviation corresponds to moving exactly one agent from one coalition to another and $a$ (or $b$) can only make a CNS deviation if her coalition has not been merged, it holds that $d(\pi, \pi') \leq 4 + \phi(\pi) - \phi(\pi')$ because there are at most two singletons coalitions in $\pi'$ that are not in $\pi$.
\end{proof}

\CISdistance*

\begin{proof}
    Let $a$ and $b$ denote the agents whose valuation was altered, and let $G'$ denote the altered game. If $\pi(a) \neq \pi(b)$, any sequence of CIS deviations from $\pi$ in $G'$ starts with a deviation by either $a$ or $b$ to $\pi(b)$ or $\pi(a)$, respectively. By \Cref{lemma:CNS_different_coalitions}, it holds that $d(\pi, \pi') \leq 1$.

    Now suppose that $\pi(a) = \pi(b)$. By \Cref{lemma:CNS_blocking_agents_exist}, no agent $x \neq a, b$ with $\lvert \pi(x)\rvert \geq 2$ can make a deviation in any sequence of CIS deviation in $G'$ starting from $\pi$. Moreover, by \Cref{lemma:singletons_negative_valuations}, no agent $s$ in a singleton coalition in $\pi$ can make a CIS deviation to any $C \in \pi \setminus \{ \pi(a)\}$ with $\lvert C \rvert \geq 2$ in any such sequence. 
    
    If $a$ (or $b$) deviates to the empty coalition, then either no agent in a singleton coalition in $\pi$ has a positive valuation for her, or she can make another CIS deviation afterwards.
    Hence, we may assume that $a$ (or $b$) only makes a CIS deviation to the empty coalition if she cannot deviate to a different coalition. By \Cref{lemma:singletons_negative_valuations}, no agent $s$ in a singleton coalition in $\pi$ can make a CIS deviation to the coalition that $a$ (or $b$) deviates to, since either $a$ (or $b$) or another agent in this coalition has a negative valuation for $s$. Moreover, at most one such $s$ can make a CIS deviation to the coalition that $a$ (or $b$) left. Hence, it holds that $d(\pi, \pi') \leq 3$ since $\pi(a) = \pi(b)$.
\end{proof}

\CNSDistanceFourTheoremLemma*

\begin{proof}
    As in the proof of \Cref{thm:CIS_FEGS_polynomial_time}, it is sufficient to consider the case $\pi(a) = \pi(b)$ by \Cref{lemma:CNS_different_coalitions}.
    Again, note that by \Cref{lemma:CNS_blocking_agents_exist} no agent $x \neq a,b$ with $\lvert \pi(x) \rvert \geq 2$ can make a CNS deviation in \Cref{alg_closeCNS}.
    If $a$ (or $b$) makes a CNS deviation in \Cref{alg_closeCNS}, it holds that $v(a, s) = -\beta$ (or $v(b, s) = -\beta$) for all $s$ in singleton coalitions in $\pi$. By \Cref{lemma:singletons_negative_valuations}, no such agent $s$ can make a CNS deviation to the coalition to which $a$ or $b$ deviated. Again, by \Cref{lemma:singletons_negative_valuations}, such an agent $s$ can only make a CNS deviation to $\pi(a)$ since $a$ and $b$ make a CNS deviation to the largest possible coalition.
    At most two agents, namely $a$ and $b$, can leave $\pi(a)$, so the utility for this coalition can increase by at most $2 \beta$ for any $s$ in a singleton coalition in $\pi$. Since the only negative utility value is $- \beta$, at most two agents in singleton coalitions can join this coalition by \Cref{lemma:singletons_negative_valuations}. Therefore, $d(\pi, \pi') \leq 4$.
    
\end{proof}

\CINSasym*
We divide the proof into two parts. First, we show the statement for CNS and then for CIS.

\begin{proof}[Proof of \Cref{thm:CISandCNS_single_asymmetric_NP_complete} for $X = \text{CNS}$]
    Let $(E, \mathcal{S}, k)$ be a set cover instance. Without loss of generality, we can assume that $2 < k < \lvert E \rvert -1, \lvert \mathcal{S}\rvert - 1$ and that $\mathcal{S}$ is a set cover of $E$, as otherwise the problem is trivial. We create a (non-symmetric) \ac{ASHG} $G$ with $N = S \cup E \cup Y \cup \{z_1, z_2, z_3\}$ with $\lvert Y \rvert = 4 \left( \lvert E\rvert +\lvert S \rvert \right) + 7$ and its valuations defined by
    \begin{itemize}
        \item $v_S(S') = \alpha$ for all $S, S' \in \mathcal{S}$,
        \item $v_e(S) = v_S(e) = \alpha$ for all $S \in \mathcal{S}$ with $e \in S$,
        \item $v_{z_1}(e) = \alpha$ for all $e \in E$,
        \item $v_{z_i}(y)  = \alpha$ for all $y \in Y$ and $i \in [3]$,
        \item $v_{z_1}(z_2) = v_{z_2}(z_3) = v_{z_3}(z_2) =  v_{z_3}(z_1) = \alpha$
        \item $v_y(y') = \alpha$ for all $y, y' \in Y$,
        \item $-\beta$ for all other valuations,
    \end{itemize}
     where we set $\alpha := \alpha (\lvert N \rvert)$ and $\beta := \beta(\lvert N \rvert)$. Since $\alpha$ and $\beta$ are polynomial-time computable functions, $G$ is polynomial in the encoding size of $(E, \mathcal{S})$.
    \begin{figure}[ht]
    \centering
\begin{tikzpicture}[>=stealth,
        every node/.style={font=\small},
        main/.style={circle,draw,minimum size=7mm,font=\small},
        rect/.style={draw, minimum width=2cm, minimum height=7cm, align=center}]
        
    \node[rect] (U)  at (-3.5,0) {$Y$};
    \node[rect] (S*)  at ( 2,0) {$\mathcal{S}$};
    
    \node[main] (z3)  at (0, 1.75)   {$z_3$};
    \node[main] (z2)  at (-0.75, 3)   {$z_2$};
    \node[main] (z1) at (-1.5, 1.75) {$z_1$};

    \node[rect, minimum width = 2cm, minimum height = 4cm, thick, dotted] (X) at (-0.75, -1.5) {$E$};

    \node[main] (u) at (-0.5, -2.5) {$e \in S$};

    \node[main] (S) at (2.25, -2.5) {$S$};

    \node (z3x) at (-2.63, 1) {};

    \begin{scope}[every node/.style={font=\small}]
                \path
                    (z1)  edge[->] node[] {} (z2)
                    (z2)  edge[<->] node[] {} (z3)
                    (z3)  edge[->] node[] {} (z1)
                    (z3) edge[->, out = 225, in = 0] node[] {} (z3x)
                    (u) edge[<->] node[] {} (S);
            \end{scope}     
            
    \draw[->] (z3) -- node[] {} (0, 0.5);
    \draw[->] (z1) -- node[] {} (-2.5, 1.75);
    \draw[->] (z2) -- node[above]  {} (-2.5, 3);

    \node (Ebr) at (0, -3.25) {};
    \node (Ebl) at (-1.5, -3.25) {};
    \draw[thick,CoalitionColor, fill=CoalitionColor!50, fill opacity=0.2] 
    \convexpath{z1, z2, z3, Ebr, Ebl}{.45cm};

    \node (Str) at (2.75, 3.25) {};
    \node (Stl) at (1.25, 3.25) {};
    \node (Sbr) at (2.75, -3.25) {};
    \node (Sbl) at (1.25, -3.25) {};
    \draw[thick,CoalitionColor, fill=CoalitionColor!50, fill opacity=0.2] 
    \convexpath{Str, Sbr, Sbl, Stl}{.45cm};

    \node (Xtr) at (-2.75, 3.25) {};
    \node (Xtl) at (-4.25, 3.25) {};
    \node (Xbr) at (-2.75, -3.25) {};
    \node (Xbl) at (-4.25, -3.25) {};
    \draw[thick,CoalitionColor, fill=CoalitionColor!50, fill opacity=0.2] 
    \convexpath{Xtr, Xbr, Xbl, Xtl}{.45cm};
    
\end{tikzpicture}
\caption{Illustration of $G$ from the reduction in \Cref{thm:CISandCNS_single_asymmetric_NP_complete} for $X = \text{CNS}$. Undotted rectangles are cliques, and dotted rectangles are independent sets. There is an arc $(i, j)$ if and only if $v_i(j) = \alpha$. The coalitions in $\pi$ are highlighted in blue.}
\end{figure}
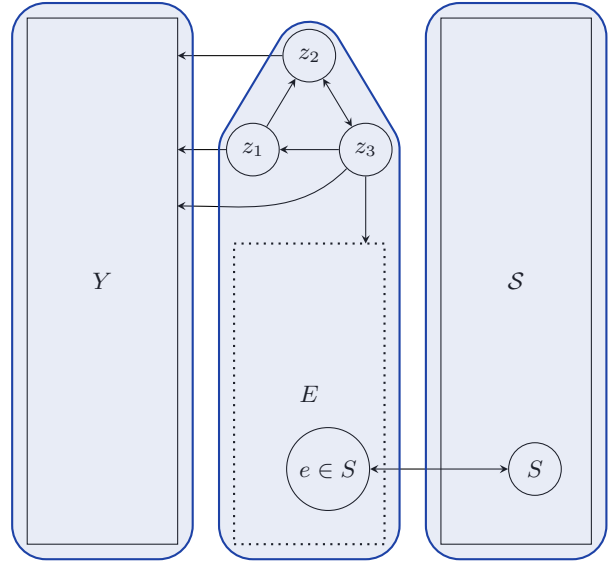

    Then, $\pi = \{E \cup \{z_1, z_2, z_3\}, \mathcal{S}, Y\}$ is a CNS partition in $G$ since every agent is blocked from leaving its coalition by at least one other agent. 
    Hence, no agent can make a CNS deviation.

    Alter the valuation of agent $z_2$ for $z_3$ such that $v_{z_2}(z_3) = -\beta$ and denote the altered game by $G'$. 
     We claim that a set cover $\mathcal{C \subseteq S}$ of $E$ with $ \lvert \mathcal{C} \rvert \leq k$ exists if and only if $G'$ has a CNS partition $\pi'$ with $d(\pi, \pi') \leq k + 3$.
    To this end, let $\mathcal{C \subseteq S}$ denote such a set cover and let $\pi' = \{E \cup \mathcal{C}, \mathcal{S \setminus C}, Y \cup \{z_1, z_2, z_3\}\}$.
    Clearly, it holds that $d(\pi, \pi') \leq k + 3$. To see that this is a CNS partition in $G'$, we perform a case analysis on the types of agents to show that no agent can make a CNS deviation.

    \begin{itemize}
        \item Let $e$ be an agent in $E$. Since $\mathcal{C}$ is a set cover, there is some $S \in \pi'(e)$ with $v_S(e) > 0$. Hence, $e$ cannot make a CNS deviation.
        \item For any agent $S \in \mathcal{S} \setminus \mathcal{C}$, there is another $S \neq S' \in \mathcal{S} \setminus \mathcal{C} $ with $v_{S'}(S) > 0$ since we assumed that $1 < k < \lvert S \rvert- 1$. 
        Hence, $S$ cannot make a CNS deviation.
        For any $S \in \mathcal{S}$ it holds for any $e \in S$ that $v_{e}(S) > 0$. Since $e \in \pi'(S)$, $S$ cannot CNS deviate.
        \item Agent $z_3$ cannot make a CNS deviation since it holds that $u_{z_3}(\pi'(z_3)) = \alpha (\lvert Y \rvert + 2) \geq u_{z_3}(Z)$ for all $Z \in \pi'$.
        \item Any agent $a \in Y \cup \{z_1, z_2\}$ cannot make a CNS deviation since $v_{z_3}(a) > 0$.
    \end{itemize}

    To prove the reverse direction, suppose that $\pi'$ is an CNS partition in $G'$ with $d(\pi, \pi') \leq k + 3$.
    Since $k + 3 < \lvert E \rvert < \frac{1}{2}\lvert Y \rvert$, there exists a unique coalition that contains the most agents of $Y$. We denote this coalition by $\pi'(y^*)$.
    Observe that it must hold that $z_3 \in \pi'(y^*)$ since 
    the utility of $z_3$ is maximal in this coalition and no agent has a positive valuation for $z_3$ in $G'$, so any CNS deviation of $z_3$ is an NS deviation. 
    Therefore, it must also hold that $z_1, z_2 \in \pi'(y^*)$.
    
    Let $R_E$ and $R_{\mathcal{S}}$ denote the set of agents in $E$ and $\mathcal{S}$ that were moved from their coalition in $\pi$ to another. Moreover, let $f: E \rightarrow \mathcal{S}$ be a function mapping every $e \in E$ to some $S \in \mathcal{S}$ with $e \in S$. We claim that $\mathcal{C} := f(R_E) \cup R_{\mathcal{S}}$ is a set cover of $E$ with $\lvert \mathcal{C} \rvert \leq k$. Since we have that $z_i \in \pi'(y^*)$ for all $i \in [3]$, it holds that $\lvert R_E \rvert + \lvert R_{\mathcal{S}} \rvert \leq k$.
    To see that $\mathcal{C}$ is a set cover of $E$, assume for contradiction that there is some $e \in E$ that is not covered by $\mathcal{C}$. Then, by construction of $\mathcal{C}$, there is no $S \in \pi'(e)$ with $e \in S$, and it holds that $e \not \in \pi'(z_3) = \pi'(y^*)$. Hence, it holds that $v_a(e) = v_e(a) < 0$ for all $e \neq a \in \pi'(e)$. Since $k+3 < \lvert E  \rvert - 1 +3 < \lvert \pi(e) \rvert$, $e$ is not in a singleton coalition. Thus, $e$ can make a CNS deviation to the empty coalition. However, this is a contradiction to $\pi'$ being CNS in $G'$, so $\mathcal{C}$ is in fact a set cover of $E$.
\end{proof}

\begin{proof}[Proof of \Cref{thm:CISandCNS_single_asymmetric_NP_complete} for $X = \text{CIS}$]
    The reduction is from exact-cover by 3-sets (X3C), where we are given a tuple $(E, \mathcal{F})$, such that each set $F \in \mathcal{F}$ is a subset of $E$ and contains exactly three elements. The question is whether a subset $\mathcal{C \subseteq F}$ exists, such that each element of $E$ is contained in exactly one of the sets in $\mathcal{C}$. Such a set $\mathcal{C}$ is called an exact cover.
    
    Let $(E, \mathcal{F})$ be an X3C instance. Without loss of generality, we can assume that $\mathcal{F}$ is a set cover of $E$ and that $\lvert E \rvert > 6$ is divisible by 3, as otherwise the problem is trivial. Enumerate the ground set $E = \{e_1, \dots, e_m\}$ and let $M := \lvert E \rvert + \lvert \mathcal{F} \rvert$.
    We create a non-symmetric \ac{ASHG} $G$ (see \Cref{Figure:CISReductionXThreeC}) with $N = E \cup \mathcal{F} \cup Y'_1 \cup \dots \cup Y'_{m} \cup W \cup \{z_1, z_2, z_3\}$ with $\lvert Y_i \rvert = M^2$, $\lvert W \rvert = M^3$, where $Y'_i = Y_i \cup B_i$ with $B_i = \{b_{i, j}: j \in [m], j \neq i\}$.
    Define the valuations of $G$ by 
    \begin{itemize}
        \item $v_e(e')= \alpha$ for all $e, e' \in E$,
        \item $v_{e}(F) = v_F(e) = \alpha$ for all $F \in \mathcal{F}$ with $e \in E \setminus F$,
        \item $v_F(F') = \alpha$ for all $F, F' \in \mathcal{F}$,
        \item $v_e(w) = v_w(e) = \alpha$ for all $e \in E$ and $w \in W$,
        \item $v_{e_i}(y_i) = \alpha$ for all $y_i \in Y_i$ and $i \in [m]$
        \item $v_{e_i}(y_j) = v_{y_j}(e_i)$ for all $y_j \in Y_j$ and $i, j \in [m]$ with $i \neq j$,
        \item $v_{e_{\ell}}(b_{i, j}) = v_{b_{i, j}}(e_{\ell}) = \alpha$ for all $i, j, \ell \in [m]$ with $\ell \neq i, j$,
        \item $v_{z_1}(z_2) = v_{z_2}(z_3) = v_{z_3}(z_1) = \alpha$,
        \item $-\beta$ for all other valuations,
    \end{itemize}
    where we set $\alpha := \alpha (\lvert N \rvert)$ and $\beta := \beta(\lvert N \rvert)$. Since $\alpha$ and $\beta$ are polynomial-time computable functions, $G$ is polynomial in the encoding size of $(E, \mathcal{F})$.
    
    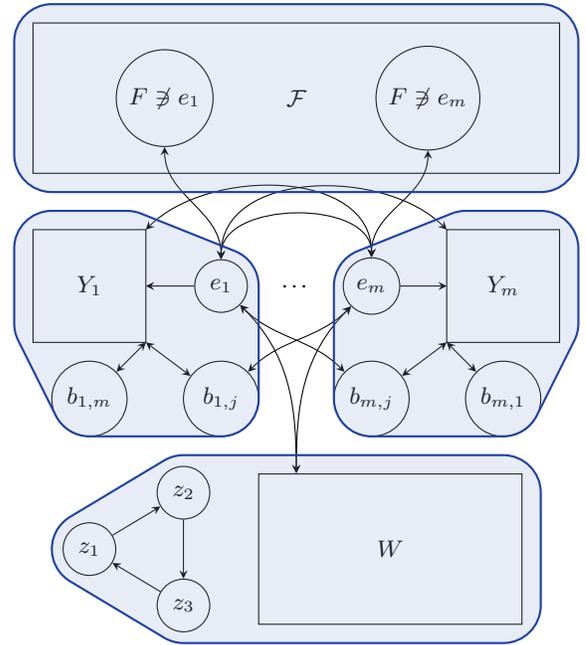
\begin{figure}[ht]
    \centering
\begin{tikzpicture}[>=stealth,
        every node/.style={font=\small},
        main/.style={circle,draw,minimum size=7mm,font=\small},
        rect/.style={draw, align=center}]
        
    \node[rect, minimum width = 7cm, minimum height = 2cm] (S*)  at (0,0) {$\mathcal{F}$};

    \node[main] (e1) at (-1, -2.5) {$e_1$};
    \node[main] (em) at (1, -2.5) {$e_m$};

    \node[main] (Snm) at (1.75, 0) {$F \not\ni e_m$};
    \node[main] (Sn1) at (-1.75, 0) {$F \not \ni e_1$};

    \node[rect, minimum width = 1.5cm, minimum height = 1.5cm] (Y1) at (-2.75, -2.5) {$Y_1$};
    \node[rect, minimum width = 1.5cm, minimum height = 1.5cm] (Ym) at (2.75, -2.5) {$Y_m$};
    \node at ($(e1)!.5!(em)$) {\ldots};

    \node[main, minimum size = 10mm, font = \footnotesize] (b1j) at (-1, -4) {$b_{1, j}$};

    \node[main, minimum size = 10mm, font = \footnotesize] (bmj) at (1, -4) {$b_{m, j}$};

    \node[main, minimum size = 10mm, font = \footnotesize] (b1m) at (-2.75, -4) {$b_{1,m}$};

    \node[main, minimum size = 10mm, font = \footnotesize] (bm1) at (2.75, -4) {$b_{m, 1}$};

    \node[main] (z2) at (-1.5,-5.25) {$z_2$};
    \node[main] (z1) at (-2.75, -6) {$z_1$};
    \node[main] (z3) at (-1.5, -6.75) {$z_3$};

    \node[rect, minimum width = 3.5cm, minimum height = 2cm] (X) at (1.25, -6) {$W$};

    \node (xm) at (0, -5.115) {};

    \node (e1Ym) at (2.125, -1.875) {};
    \node(emY1) at (-2.125, -1.875) {};

    \begin{scope}[every node/.style={font=\small}]
                \path
                    (b1j) edge[<->] node[] {} (-2, -3.25)
                    (b1m) edge[<->] node[] {} (-2, -3.25)
                    (b1j) edge[<->, out = 45, in = 225] node[] {} (em)
                    (bmj) edge[<->] node[] {} (2, -3.25)
                    (bm1) edge[<->] node[] {} (2, -3.25)
                    (bmj) edge[<->, out = 135,in = 315] node[] {} (e1)
                    (z1)  edge[->] node[] {} (z2)
                    (e1) edge[<->, out = 90, in = 90] node[] {} (em)
                    (z2) edge[->] node[] {} (z3)
                    (z3) edge[->] node[] {} (z1)
                    (e1) edge[<->, out = 315, in = 90] node[] {} (xm)
                    (em) edge[<->, out = 225, in = 90] node[] {} (xm)
                    (e1) edge[<->, out = 90, in = 135] node[] {} (e1Ym)
                    (em) edge[<->, out = 90, in = 45] node[] {} (emY1)
                    (e1) edge[<->, out=90, in = 270] node[] {} (Sn1)
                    (em) edge[<->, out=90, in = 270] node[] {} (Snm);                    
            \end{scope}     
            
    \draw[->] (e1) -- node[out = 315] {} (-2, -2.5);
    \draw[->] (em) -- node[] {} (2, -2.5);

    \node (Xtr) at (2.75, -5.25) {};
    \node (Xbr) at (2.75, -6.75) {};
    \draw[thick,CoalitionColor, fill=CoalitionColor!50, fill opacity=0.2] 
    \convexpath{z3, z1,z2, Xtr, Xbr}{.5cm};

    \node (Y1bl) at (-3.25, -3) {};
    \node (Y1br) at (-2.25, -3) {};
    \node (Y1tl) at (-3.25, -2) {};
    \node (Y1tr) at (-2.25, -2) {};
    \draw[thick,CoalitionColor, fill=CoalitionColor!50, fill opacity=0.2] 
    \convexpath{e1, b1j,b1m, Y1bl, Y1tl, Y1tr}{.5cm};

    \node (Ymbr) at (3.25, -3) {};
    \node (Ymbl) at (2.25, -3) {};
    \node (Ymtl) at (3.25, -2) {};
    \node (Ymtr) at (2.25, -2) {};
    \draw[thick,CoalitionColor, fill=CoalitionColor!50, fill opacity=0.2] 
    \convexpath{bm1, bmj, em,Ymtr, Ymtl, Ymbr}{.5cm};

    \node (Stl) at (-3.25, 0.75) {};
    \node (Sbl) at (-3.25, -0.75) {};
    \node (Str) at (3.25, 0.75) {};
    \node (Sbr) at (3.25, -0.75) {};
    \draw[thick,CoalitionColor, fill=CoalitionColor!50, fill opacity=0.2] 
    \convexpath{Sbl, Stl, Str, Sbr}{.5cm};
\end{tikzpicture}
\caption{Illustration of $G$ from the reduction in \Cref{thm:CISandCNS_single_asymmetric_NP_complete} for $X = \text{CIS}$. Rectangles are cliques. There is an arc $(i, j)$ if and only if $v_i(j) = \alpha$. The coalitions in $\pi$ are highlighted in blue.}
\label{Figure:CISReductionXThreeC}
\end{figure}

    We claim that $\pi = \{\{e_1\} \cup Y'_1, \dots, \{e_m\} \cup Y'_m, \mathcal{F}, W \cup \{z_1, z_2, z_3\}\}$ is CIS in $G$. To see this, we perform a case analysis to show that no type of agent can make a CIS deviation.
    \begin{itemize}
        \item An agent $e_j \in E$ can only increase her utility by joining $\pi'(z_1)$ or $\pi'(e_i)$ for some $i \neq j$. However, it holds that $v_{z_1}(e_j), v_{b_{i, j}}(e_j) = - \beta < 0$ and $z_1 \in \pi(z_1), b_{i, j} \in \pi(e_i)$, so agent $e_j$ cannot make a CIS deviation.
        \item An agent $a \in \mathcal{F}\cup Y'_1 \cup \dots\cup Y'_m \cup X \cup \{z_1, z_2, z_3\}$ cannot make a CIS deviation since there is some $a' \in \pi(a)$ with $v_{a'}(a) > 0$ that blocks her from leaving her coalition.
    \end{itemize}
    Alter the valuation of agent $z_2$ for $z_3$ such that $v_{z_2}(z_3)= -\beta$ and denote the altered game by $G'$. We claim that an exact cover $\mathcal{C \subseteq F}$ of $E$ exists if and only if $G'$ has a CIS partition $\pi'$ with $d(\pi, \pi') \leq \frac{m}{3} + 3$.

    To this end, let $\mathcal{C \subseteq F}$ be an exact cover $E$ and let $\pi' = \{\mathcal{F \setminus C}, \{e_1\} \cup Y'_1, \dots, \{e_m\} \cup Y'_m, \mathcal{C} \cup W, \{z_1\}, \{z_2\}, \{z_3\}\}$. Clearly, it holds that $d(\pi, \pi') \leq \frac{m}{3} + 3$. To see that $\pi'$ is CIS in $G'$, we again perform a case analysis on the types of agents in $G'$.
    \begin{itemize}
        \item An agent $e_j \in E$ can only increase her utility by joining $\mathcal{C} \cup W$ or $\pi'(e_i)$ for some $i \neq j$. However, it holds that $v_F(e_j), v_{b_{i, j}}(e_j) = - \beta < 0$ and $F \in \mathcal{C}, b_{i, j} \in \pi'(e_i)$ since $\mathcal{C}$ is a set cover of $E$. Hence, agent $e_j$ cannot make a CIS deviation.
        \item For an agent $F \in \mathcal{\mathcal{C}}$, there is some $F' \in \pi'(F) \cap \mathcal{F}$ with $v_{F'}(F) > 0$ since $\lvert \mathcal{C} \rvert > 1$. Hence, $S$ cannot make a CIS deviation.
        \item Analogously, an agent $F \in \mathcal{F \setminus C}$ cannot make a CIS deviation since there is some $F' \in \pi'(F) = \mathcal{F \setminus C}$ with $v_{F'}(F) > 0$ since $\lvert \mathcal{C} \rvert = \frac{m}{3}$.
        \item An agent $a \in Y'_1 \cup \dots\cup Y'_m \cup W$ cannot CIS deviate to another coalition since there is some $a' \in \pi(a)$ with $v_{a'}(a) > 0$.
        \item Finally, consider an agent $z_i$ for some $i \in [3]$. For all $a \in N \setminus \{z_1, z_2, z_3\}$ it holds that $v_{z_i}(a) = v_{a}(z_i) = - \beta < 0$, so $z_i$ cannot make a CIS deviation to any coalition containing such an agent $a$. Moreover, if $v_{z_i}(z_j) > 0$, then $v_{z_j}(z_i) < 0$ so $z_i$ cannot deviate to any other coalition as well.
        \end{itemize}

    To prove the reverse direction, suppose that $\pi'$ is CIS in $G'$ with $d(\pi, \pi') \leq \frac{m}{3} + 3$.
    Let $\pi'(w^*)$ denote the coalition in $\pi'$ that contains most agents in $W$. This coalition is unique since $\frac{1}{2}\lvert W \rvert > m$. 
    Since no agent has a positive valuation for $z_3$, it holds that $z_3 \not \in \pi'(w^*)$. 
    Thus, it also holds that $z_1, z_2 \not \in \pi'(w^*)$.
    Let $D = \{d_1,\dots d_r\}$ be the agents that change their coalition in a shortest sequence of changes from $\pi$ to $\pi'$. Moreover, let $\pi^{(i)}$ denote the partition after $d_1, \dots, d_i$ have changed their coalition in this sequence and let $\pi^{(i)}(w^*)$ denote the coalition in $\pi^{(i)}$ containing most agents in $W$. Since $\frac{1}{2} \lvert W \rvert > m$, $\pi^{(i)}(w^*)$ is well-defined. Clearly, it hold that $\pi^{(0)} = \pi$ and $\pi^{(r)} = \pi'$. 
    Note that the order of the agents in $D$ does not matter, so without loss of generality, we may assume that $d_1 = z_1, d_2 = z_2, d_3 = z_3$. In $\pi^{(3)}$, any $e_i$ can make a CIS deviation to $\pi^{(3)}(w^*) = W$. 
    Let $R_{\mathcal{F}} : = D\cap \mathcal{F}$ and suppose for contradiction that $\lvert R_{\mathcal{F}} \rvert < \frac{m}{3}$. Then, there are at least $m - 3 \lvert R_{\mathcal{F}} \rvert > 0$ agents in $E$ for which there is no $F \in \pi'(w^*) \cap R_{\mathcal{F}}$ with $v_F(e) < 0$. Denote these agents by $E^*$. Then, for any $i > 3$ with $d^{(i)} \not \in \mathcal{F}$, at most two agents in $E^*$ that could make a CIS deviation in $\pi^{(i - 1)}$, cannot make a CIS deviation in $\pi^{(i)}$. 
    
    Hence, there need to be at least $\frac{1}{2}\lvert E^* \rvert \geq  \frac{m - 3\lvert R_{\mathcal{F}} \rvert}{2}$ many such $i$'s since each $F \in R_{\mathcal{F}}$ blocks at most 3 agents in $E$ from making a CIS deviation.
    Therefore, by assumption that $\lvert R_{\mathcal{F}} \rvert < \frac{m}{3}$, it holds that
    \begin{align*}
    r - 3 \geq \frac{1}{2}\lvert E^* \rvert + \lvert R_{\mathcal{F}} \rvert &\geq \frac{1}{2}(m- \lvert R_{\mathcal{F}}\rvert) \\ &> \frac12(m - \frac{m}{3}) = \frac{m}{3}.
    \end{align*}    
    Since $D$ corresponds to a shortest sequence of coalition changes, this implies that $d(\pi, \pi') > \frac{m}{3} + 3$, which is a contradiction. Therefore, we have that $\lvert R_{\mathcal{F}} \rvert = \frac{m}{3}$ and that $R_{\mathcal{F}}$ is a set cover of $E$ since no $e \in E$ can make a CIS deviation in $\pi'$. Finally, by $ \lvert R_{\mathcal{F}} \rvert \leq \frac{m}{3}$, $R_{\mathcal{F}}$ is an exact cover of $E$.
\end{proof}

\subsection{NS and IS}

\ISsymmetric*

\begin{proof}
    We first prove hardness for \textsc{IS-1-1-Sym-Altered}. At the end of the proof, we explain how to slightly adapt the reduction for the non-symmetric case.
    Let $(E, \mathcal{S}, k)$ be a set cover instance. Without loss of generality, we can assume that $1 < k < \lvert E \rvert-1, \lvert \mathcal{S} \rvert - 1$ and that $\mathcal{S}$ is a set cover of $E$, as otherwise the problem is trivial.
    Let $M = 4(\lvert E \rvert + \lvert \mathcal{S} \rvert)$. Define a symmetric \ac{ASHG} $G$ (see \Cref{fig:is_hardness}) with agent set $N = E \cup \mathcal{S} \cup X \cup Y \cup \{z_1, z_2\}$ with $\lvert X \rvert = M^2$ and $\lvert Y \rvert = M^3$, where $X = X_1 \cup X_2$ with $\lvert X_1 \rvert = 3\lvert X_2 \rvert$. Let its valuations be defined by 
    \begin{itemize}
        \item $v(e, S) = \alpha$ for all $S \in \mathcal{S}$ and $e \in E \setminus S$,
        \item $v(e, e') = v(e, x) = v(x,x')  = \alpha$ for all $e, e' \in E$ and $x, x' \in X$,
        \item $v(S, x_1) = \alpha$ for all $S \in \mathcal{S}$ and $x_1 \in X_1$,
        \item $v(z_1, x) = v(z_1, y) = \alpha$ for all $x \in X$ and $y \in Y$,
        \item $v(z_2, y) = v(y, y')= \alpha$ for all $y, y' \in Y$,
        \item $-\beta$ for all other valuations,
    \end{itemize}
    where we set $\alpha := \alpha(\lvert N \rvert)$ and $\beta := \beta(\lvert N \rvert)$. Since $\alpha$ and $\beta$ are polynomial-time computable functions, $G$ is polynomial in the encoding size of $(E, \mathcal{S})$.
    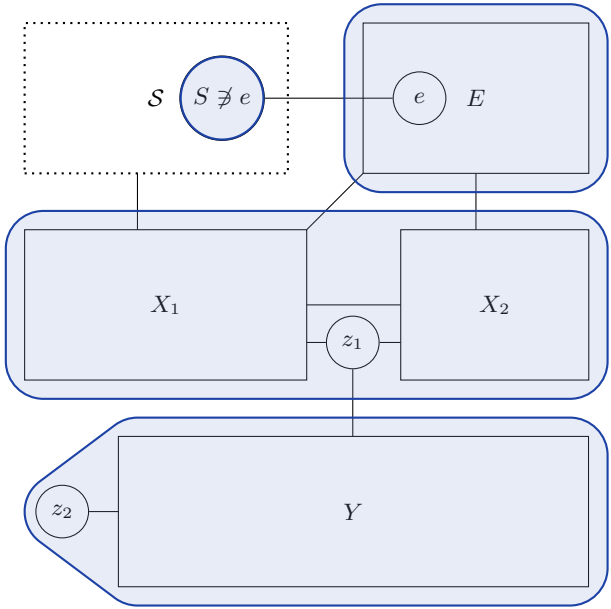
\begin{figure}[ht]
    \centering
\begin{tikzpicture}[>=stealth,
        every node/.style={font=\small},
        main/.style={circle,draw,minimum size=7mm,font=\small},
        rect/.style={draw, align=center}]
        
    \node[rect, minimum width = 3cm, minimum height = 2cm] (E)  at (2.25,0) {$E$};

    \node[rect, dotted, thick, minimum width = 3.5cm, minimum height = 2cm] (S*)  at (-2, 0) {$\mathcal{S}$};

    \node[rect, minimum width = 3.75cm, minimum height = 2cm] (X1)  at (-1.875, -2.75) {$X_1$};

    \node[rect, minimum width = 2.5cm, minimum height = 2cm] (X2)  at (2.5, -2.75) {$X_2$};

    \node[main] (zp) at (-3.25, -5.5) {$z_2$};
    \node[main] (z) at (0.6125, -3.25) {$z_1$};

    \node[rect, minimum width = 6.25cm, minimum height = 2cm] (Y) at (0.6215,-5.5) {$Y$};
    
    \node[main] (e) at (1.5, 0) {$e$};
    \node[main, font = \footnotesize] (Se) at (-1.125, 0) {$S \not\ni e$};

    \begin{scope}[every node/.style={font=\small}]
                \path
                    (e)  edge[-] node[] {} (Se);
                    
            \end{scope}     
            
    \draw[-] (0, -2.75) -- node[] {} (1.25, -2.75);
    \draw[-] (-2.25, -1) -- node[] {} (-2.25, -1.75);
    \draw[-] (2.25, -1) -- node[] {} (2.25, -1.75);
    \draw[-] (0.75, -1) -- node[] {} (0, -1.75);
    \draw[-] (z) -- node[] {} (0,-3.25);
    \draw[-] (z) -- node[] {} (1.25,-3.25);
    \draw[-] (zp) -- node[] {} (-2.5,-5.5);
    \draw[-] (z) -- node[] {} (0.6125, -4.5);

    \node (s) at (-1.125, 0.0001) {};
    \draw[thick,CoalitionColor, fill=CoalitionColor!50, fill opacity=0.2] 
    \convexpath{Se, s}{.55cm};

    \node (Etr) at (3.5, 0.75) {};
    \node (Ebr) at (3.5, -0.75) {};
    \node (Etl) at (1, 0.75) {};
    \node (Ebl) at (1, -0.75) {};
    \draw[thick,CoalitionColor, fill=CoalitionColor!50, fill opacity=0.2] 
    \convexpath{Etr, Ebr, Ebl, Etl}{.5cm};

    \node (Xtr) at (3.5, -2) {};
    \node (Xbr) at (3.5, -3.5) {};
    \node (Xtl) at (-3.5, -2) {};
    \node (Xbl) at (-3.5, -3.5) {};
    \draw[thick,CoalitionColor, fill=CoalitionColor!50, fill opacity=0.2] 
    \convexpath{Xtr, Xbr, Xbl, Xtl}{.5cm};

    \node (Ytr) at (3.5, -4.75) {};
    \node (Ybr) at (3.5, -6.25) {};
    \node (Ytl) at (-2.25, -4.75) {};
    \node (Ybl) at (-2.25, -6.25) {};
    \draw[thick,CoalitionColor, fill=CoalitionColor!50, fill opacity=0.2] 
    \convexpath{Ytr, Ybr, Ybl,zp, Ytl}{.5cm};
\end{tikzpicture}
\caption{\label{fig:is_hardness}Illustration of $G$ from the reduction in \Cref{IS_completeness_symmetric_FEGs}. Undotted Rectangles are cliques, and dotted rectangles are independent sets. There is an edge $\{i, j\}$ if and only if $v(i, j) = \alpha$. Coalitions in $\pi$ are highlighted in blue.}
\end{figure}
    We claim that $\pi = \{E, \{S\}_{S \in \mathcal{S}}, \{z_1\} \cup X, \{z_2\} \cup Y\}$ is IS in $G$. To see this, we perform a case analysis on the types of agents in $G$ to show that no agent can make an IS deviation.
    \begin{itemize}
        \item For an agent $e \in E$, it holds that $u_e(\pi(e)) = \alpha (\lvert E \rvert - 1)$, so $e$ can only increase her utility by deviating to $\{z_1\} \cup X$. However, agent $z_1$ blocks her from making such a deviation since $v(z_1, e) = -\beta < 0$. 
        \item An agent $S \in \mathcal{S}$ cannot make an IS deviation since any other coalition in $\pi$ contains an agent $b$ with $w(S, b) = -\beta <  0$ that blocks $S$ from joining her coalition. 
        \item An agent $a \in X \cup Y \cup \{z_2\}$ cannot deviate since she has maximal utility in $\pi(a)$.
        \item Agent $z_1$ can only increase her utility by deviating to $\{z_2\} \cup Y$, but agent $z_2$ blocks her from making such a deviation since $v(z_1, z_2) = -\beta < 0$.
    \end{itemize}
    Alter the valuation between agents $z_1$ and $z_2$ such that $v'(z_1, z_2) = \alpha$ and denote the altered game by $G'$. We claim that a set cover $\mathcal{C} \subseteq \mathcal{S}$ of $E$ with $\lvert \mathcal{C} \rvert = k$ exists if and only if $G'$ has an IS partition $\pi'$ with $d(\pi, \pi') \leq k + 1$.

    To this end, let $\mathcal{C \subseteq S}$ denote such a set cover of $E$. Let $\pi' = \{ E, \{S\}_{S \in \mathcal{S \setminus C}}, \mathcal{C} \cup X, \{z_1, z_2\} \cup Y \}$. Clearly, it holds that $d(\pi, \pi') \leq k+1$. To see that $\pi'$ is IS, we again perform a case analysis to show that no agent can make an IS deviation.
    \begin{itemize}
        \item An agent $e \in E$ can only increase her utility by deviating to $\mathcal{C} \cup X$. However, since $\mathcal{C}$ is a set cover of $E$, there is some $S \in \mathcal{C}$ with $e \in S$ and $w(e, S) = -\beta < 0$. Hence, $e$ cannot make an IS deviation.
        \item For an agent $S \in \mathcal{C}$, it holds that $u_S(\pi(S)) = \alpha ( \lvert X_1 \rvert + k -1) - \beta (\lvert X_2 \rvert)\geq 0$ since we assumed that $\alpha \geq \beta$. Hence, agent $S$ cannot deviate to the empty coalition. Any other coalition contains an agent $b$ with $v(b, S) < 0$, so $S$ cannot make an IS deviation.
        \item An agent $S \in \mathcal{S\setminus C}$ cannot make an IS deviation since any other nonempty coalition in $\pi'$ contains an agent $b$ with $w(b, S) < 0$ that blocks $S$ from joining this coalition.
        \item An agent $a \in \mathcal{C} \cup X \cup Y \cup \{z_1, z_2\}$ cannot deviate since she has maximal utility in $\pi'(a)$.
    \end{itemize}
    To prove the reverse direction, suppose that $\pi'$ is IS in $G'$ with $d(\pi, \pi') \leq k + 1$. Since any agent in $\{z_2\} \cup Y$ has negative valuations for all other agents except $z_1$, and $z_1$ has maximal utility for the coalition in $\pi'$ that contains the most agents in $Y$ since $d(\pi, \pi') \leq k +1$, it holds that $\{z_1, z_2\} \cup Y \in \pi'$. 

    Let $R_E$ and $R_{\mathcal{S}}$ denote the agents in $E$ and $\mathcal{S}$ that were moved from their coalition in $\pi$ to another. Let $f: E \rightarrow \mathcal{S}$ map each $e \in E$ to some $S \in \mathcal{S}$ with $e \in S$. We claim that $\mathcal{C} := f(R_E) \cup R_{\mathcal{S}}$ is a set cover of $E$ with $\lvert C \rvert \leq k$. By assumption, it holds that $\lvert R_E \rvert + \lvert R_{\mathcal{S}} \rvert \leq k$ since $z_1$ has moved to another coalition as well. To see that $\mathcal{C}$ is a set cover of $E$, suppose that there is some $e \in E$ that is not covered by $\mathcal{C}$. 
    Denote the coalition in $\pi'$ that contains most agents in $X$ by $\pi'(x^*)$. Then, by construction of $\mathcal{C}$, there is no $S \in \pi'(x^*)$ with $w(e, S) < 0$. Since $e \not \in R_E$ and $d(\pi, \pi') \leq k + 1$, it holds that $e \not \in \pi'(x^*)$. Hence, $e$ can make an IS deviation to $\pi'(x^*)$. However, this is a contradiction to $\pi'$ being IS, so $\mathcal{C}$ is in fact a set cover of $E$. If $\lvert \mathcal{C} \rvert < k$, one can obtain a set cover $\mathcal{C'}$ with $\lvert \mathcal{C'} \rvert = k$ by adding arbitrary sets in $\mathcal{S \setminus C}$ to $C$.

    For the non-symmetric case, the proof works analogously. The only valuations we need to change are $v_{z_1}(z_2) = \alpha$ and $v_{z_2}(z_1) = - \beta$. Then, after the valuation of $z_2$ for $z_1$ is altered to $v'_{z_2}(z_1) = \alpha$, the same argument as in the symmetric case applies.
\end{proof}

\AEGsNP*

\begin{proof}
    We first prove hardness for the symmetric case, and explain afterwards how to adapt the reduction for the non-symmetric case.
    
    In symmetric AEGs, Nash stability and individual stability are equivalent due to the following observation. Suppose that $v(a, b) < 0$ and $\pi(a) \neq \pi(b)$, then $u_a(\pi(b)) \leq \lvert N \rvert - 2 + v(a, b) = -2 < 0$. Hence, an agent can only make an NS deviation to another coalition if no agent in this coalition has a negative valuation for it.
    
    Hence, it suffices to prove that \textsc{IS-1-1-Sym-Altered} is NP-complete in AEGs.
    
    Let $(E, \mathcal{F})$ be an X3C instance. Without loss of generality, we can assume that $\mathcal{F}$ is a set cover of $E$ and that $\lvert E \rvert$ is divisible by $3$, as otherwise the problem is trivial. 
    Let $M = \lvert E \rvert + \lvert \mathcal{F} \rvert + 7$ and define $G$ to be a symmetric AEG (see \Cref{fig:aegs_hardness}) with agent set $N = E \cup \mathcal{F} \cup X_1 \cup X_2 \cup Y \cup Y_{\mathcal{F}} \cup Z \cup \{z_1, z_2\} \cup W$, where $\lvert X_1 \rvert = \lvert X_2 \rvert = M^2, \lvert Y \rvert = M^3, \lvert W \rvert = M^4$ and $Y_{\mathcal{F}} = \{y_F : F \in \mathcal{F}\}$. Define its valuations by
    \begin{itemize}
        \item $v(e, F) = 1$ for all $F \in \mathcal{F}$ and $e \in E \setminus F$,
        \item $v(e, x_1) = 1$ for all $e \in E$ and $x_1 \in X_1$,
        \item $v(e, e') = v(e, y) = v(y,y')  = 1$ for all $e, e' \in E$ and $y, y' \in Y \cup Y_{\mathcal{F}}$,
        \item $v(F, x_2) = 1$ for all $F \in \mathcal{F}$ and $x_2 \in X_2$,
        \item $v(F, y) = 1$ for all $F \in \mathcal{F}$ and $y \in Y$,
        \item $v(F, y_{F'}) = 1$ for all $F, F' \in \mathcal{F}$ with $F \neq F'$,
        \item $v(z_1, y) = v(z_1, w) = 1$ for all $y \in Y \cup Y_{\mathcal{F}}$ and $w \in W$,
        \item $v(z_2, w) = v(w, w') = 1$ for all $w, w' \in W$,
        \item $-n$ for all other valuations,
    \end{itemize}
    where $n = \lvert N \rvert$. Clearly, $G$ is polynomial in the encoding size of $(E, \mathcal{F})$.
    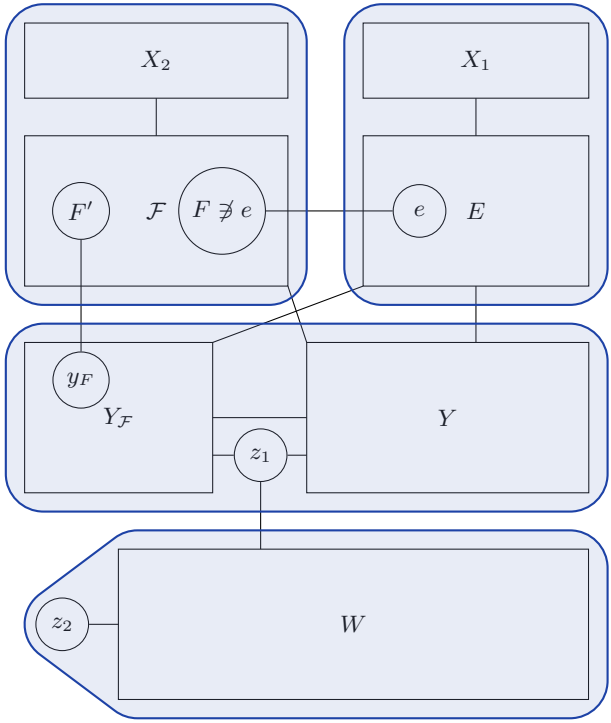
\begin{figure}[ht]
    \centering
\begin{tikzpicture}[>=stealth,
        every node/.style={font=\small},
        main/.style={circle,draw,minimum size=7mm,font=\small},
        rect/.style={draw, align=center}]
        
    \node[rect, minimum width = 3cm, minimum height = 1cm] (X) at (2.25, 2) {$X_1$};
    \node[rect, minimum width = 3cm, minimum height = 2cm] (E)  at (2.25,0) {$E$};

    \node[rect, minimum width = 3.5cm, minimum height = 1cm] (X2) at (-2, 2) {$X_2$};
    \node[rect, minimum width = 3.5cm, minimum height = 2cm] (F*)  at (-2, 0) {$\mathcal{F}$};

    \node[rect, minimum width = 3.75cm, minimum height = 2cm] (X)  at (1.875, -2.75) {$Y$};

    \node[rect, minimum width = 2.5cm, minimum height = 2cm] (XF)  at (-2.5, -2.75) {$Y_{\mathcal{F}}$};

    \node[main] (zp) at (-3.25, -5.5) {$z_2$};
    \node[main] (z) at (-0.6125, -3.25) {$z_1$};

    \node[rect, minimum width = 6.25cm, minimum height = 2cm] (Z) at (0.6215,-5.5) {$W$};

    \node[main] (e) at (1.5, 0) {$e$};
    \node[main] (Se) at (-1.125, 0) {$F \not\ni e$};
    \node[main] (Fp) at (-3, 0) {$F'$};
    \node[main] (xf) at (-3, -2.25) {$y_F$};

    \begin{scope}[every node/.style={font=\small}]
                \path
                    (e)  edge[-] node[] {} (Se)
                    (xf) edge[-] node[] {} (Fp);
                    
            \end{scope}     

    \draw[-] (-0.25, -1) -- node[] {} (0, -1.75);
    \draw[-] (2.25, 1) -- node[] {} (2.25, 1.5);
    \draw[-] (0, -2.75) -- node[] {} (-1.25, -2.75);
    \draw[-] (2.25, -1) -- node[] {} (2.25, -1.75);
    \draw[-] (0.75, -1) -- node[] {} (-1.25, -1.75);
    \draw[-] (z) -- node[] {} (0,-3.25);
    \draw[-] (z) -- node[] {} (-1.25,-3.25);
    \draw[-] (zp) -- node[] {} (-2.5,-5.5);
    \draw[-] (z) -- node[] {} (-0.6125, -4.5);
    \draw[-] (-2, 1) -- node[] {} (-2, 1.5);

    \node (Stl) at (-3.5, 2.25) {};
    \node (Sbl) at (-3.5, -0.75) {};
    \node (Str) at (-0.5, 2.25) {};
    \node (Sbr) at (-0.5, -0.75) {};
    \draw[thick,CoalitionColor, fill=CoalitionColor!50, fill opacity=0.2] 
    \convexpath{Sbl, Stl, Str, Sbr}{.5cm};

    \node (Etr) at (3.5, 2.25) {};
    \node (Ebr) at (3.5, -0.75) {};
    \node (Etl) at (1, 2.25) {};
    \node (Ebl) at (1, -0.75) {};
    \draw[thick,CoalitionColor, fill=CoalitionColor!50, fill opacity=0.2] 
    \convexpath{Etr, Ebr, Ebl, Etl}{.5cm};

    \node (Xtr) at (3.5, -2) {};
    \node (Xbr) at (3.5, -3.5) {};
    \node (Xtl) at (-3.5, -2) {};
    \node (Xbl) at (-3.5, -3.5) {};
    \draw[thick,CoalitionColor, fill=CoalitionColor!50, fill opacity=0.2] 
    \convexpath{Xtr, Xbr, Xbl, Xtl}{.5cm};

    \node (Ytr) at (3.5, -4.75) {};
    \node (Ybr) at (3.5, -6.25) {};
    \node (Ytl) at (-2.25, -4.75) {};
    \node (Ybl) at (-2.25, -6.25) {};
    \draw[thick,CoalitionColor, fill=CoalitionColor!50, fill opacity=0.2] 
    \convexpath{Ytr, Ybr, Ybl,zp, Ytl}{.5cm};
\end{tikzpicture}
\caption{\label{fig:aegs_hardness}Illustration of the symmetric AEG $G$ from the reduction in \Cref{NS_IS_hardness_AEGs}. Rectangles are cliques. There is an edge $\{i, j\}$ if and only if $v(i, j) = 1$. Coalitions in $\pi$ are highlighted in blue.}
\end{figure}
We claim that $\pi = \{E \cup X_1, \mathcal{F} \cup X_2, \{z_1\} \cup Y \cup Y_F, \{z_2\} \cup W\}$ is IS in $G$. To see this, we perform a case analysis to show that no type of agent in $G$ can make an IS deviation. Clearly, every agent $a \in N$ has $u_a(\pi(a)) \geq 0$, so we only need to argue about IS deviation to nonempty coalitions.
\begin{itemize}
    \item For an agent $e \in E$, it holds that $v(z_1, e) = v(z_2, e)  = v(x_2, e) = -n$ for $x_2 \in X_2$. Since any other nonempty coalition contains $z_1, z_2$ or an $x_2$, $e$ cannot make an IS deviation.
    \item Analogously, for $a \in N \setminus E$ any nonempty coalition other than $\pi(a)$ contains an agent with a negative valuation for $a$. Thus, $a$ cannot make an IS deviation.
\end{itemize}
Alter the valuation between agents $z_1$ and $z_2$ such that $v'(z_1, z_2) = 1$ and denote the altered game by $G'$. We claim that an exact cover $\mathcal{C \subseteq F}$ of $E$ exists if and only if $G'$ has an IS partition with $d(\pi, \pi') \leq \frac{2 \lvert E \rvert}{3} + 1 =: k + 1$.
To this end, let $\mathcal{C \subseteq F}$ denote an exact cover of $E$. Let $\pi' = \{E \cup X_1, (\mathcal{F \setminus C}) \cup X_2, \mathcal{C} \cup Y \cup Y_{\mathcal{F \setminus C}}, Y_{\mathcal{C}}, \{z_1, z_2\} \cup W\}$, where $Y_{\mathcal{C}} := \{y_F: F \in \mathcal{C}\}$ and $Y_{\mathcal{F \setminus C}} := Y_{\mathcal{F}} \setminus Y_{\mathcal{C}}$. Clearly, it holds that $d(\pi, \pi') \leq k + 1$ since $\lvert \mathcal{C} \rvert= \frac{\lvert E \rvert}{3} = \frac{k}{2}$. To see that $\pi'$ is IS in $G'$, we perform a case analysis to show that no agent can make an IS deviation. Again, it suffices to consider IS deviations to nonempty coalitions since $u_a(\pi'(a)) \geq 0$ for all $a \in N$.
\begin{itemize}
    \item An agent $e \in E$ cannot deviate to $\mathcal{C} \cup Y \cup Y_{\mathcal{F \setminus C}}$ since $\mathcal{C}$ is a set cover, so there is some $F \in \mathcal{C}$ with $e \in F$ and $v(e, F) = -n$. Because $u_e(Y_{\mathcal{C}}) = \lvert C \rvert \leq k \leq u_e(\pi'(e))$, $e$ cannot deviate to $Y_{\mathcal{C}}$. Any other nonempty coalition either contains an $x_2 \in X_2$ with $v(e, x_2) =-n$ or $z_1$ with $v(e, z_1) = -n$. Hence, $e$ cannot make an IS deviation.
    
    \item An agent $F \in \mathcal{C}$ cannot make an IS deviation since $u_F(\pi'(F)) \geq M^3 \geq u_F((\mathcal{F \setminus C}) \cup X_2), u_F(Y_{\mathcal{C}})$ and any other nonempty coalition contains an agent with a negative valuation for $F$.

    \item An agent $F \in \mathcal{F\setminus C }$ cannot make an IS deviation to $\mathcal{C} \cup Y \cup Y_{\mathcal{F\setminus C}}$ since $y_F \in Y_{\mathcal{F \setminus C}}$ with $v(F, y_F) = -n$. Any other nonempty coalition also contains an agent with a negative valuation for $F$ and $u_F(\pi'(F)) \geq 0$, so she cannot make an IS deviation.

    \item An agent $a \in Y \cup Y_{\mathcal{F}} \cup X_1 \cup X_2$ cannot make an IS deviation since every other nonempty coalition contains an agent with a negative valuation for $a$.

    \item An agent $a \in \{z_1, z_2\} \cup W$ has $u_a(\pi'(a)) \geq M^4 \geq \lvert u_a(\pi'(b)) \rvert$ for any $b \in N \setminus (\{z_1, z_2\} \cup W)$. Hence, her utility is maximal in her coalition, so she cannot make an IS deviation.
\end{itemize}
To prove the reverse direction, suppose that $\pi'$ is IS in $G'$ with $d(\pi, \pi') \leq k + 1$.
Note that any $a, a' \in \{z_2\} \cup W$ have exactly the same valuation for all agents other than themselves, and $v(a, a') = 1$, so they must be in the same coalition. Hence, there is some $S \in \pi'$ with $\{z_2\} \cup W \subseteq S$. Since $\lvert W \rvert > 5 M^3 >  \lvert N \setminus W \rvert$ and any agent in $\{z_2\} \cup W$ has negative valuations for all other agents except $z_1$, it must hold that $\{z_1, z_2\} \cup W = S \in \pi'$.

Let $R_E, R_{\mathcal{F}}$ and $R_{Y_{\mathcal{F}}}$ denote the agents in $E, \mathcal{F}$ and $Y_{\mathcal{F}}$ that were moved from their coalition in $\pi$ to another.
Denote the coalition in $\pi'$ that contains most agents in $Y$ by $\pi'(y^*)$. This coalition is unique since $\frac12\lvert Y \rvert > k + 1$. Since any $y, y'$ have exactly the same valuations for all agents other than $y, y'$, and $v(y, y') = 1$, it holds that $Y \subseteq \pi'(y^*)$.

Note that $\pi'(y^*)$ is the only coalition that an agent $F \in \mathcal{F}$ can join to potentially increase her utility.
However, $F \in \pi'(y^*)$ only holds if $y_F \not \in \pi'(y^*)$, i.e., $y_F \in R_{Y_{\mathcal{F}}}$. Therefore, we have that $\lvert R_{\mathcal{F}} \rvert \leq \lvert R_{Y_{\mathcal{F}}} \rvert$, which implies 
\begin{align}
    \lvert R_E \rvert + 2 \lvert R_{\mathcal{F}} \rvert \leq \lvert R_E \rvert + \lvert R_{\mathcal{F}} \rvert + \lvert R_{Y_{\mathcal{F}}} \rvert \leq \frac{2\lvert E \rvert}{3}
    \label{eq:aeg_ineq}
\end{align}
since $z_1$ has moved to another coalition as well.
Thus, we have that $R_{\mathcal{F}}$ does not cover at least
\begin{align}
    \lvert E \rvert - 3\lvert R_{\mathcal{F}}\rvert \leq  \lvert E \rvert - 3\left(\frac{\lvert E \rvert}{3} - \frac{\lvert R_E \rvert}{2} \right) = \frac{3\lvert R_E \rvert}{2} \label{eq:aeg_covering}
\end{align}
many elements of $E$ since each set in $\mathcal{F}$ has exactly three elements.

By $d(\pi, \pi') \leq k +1 < \lvert E \rvert$, it holds that $E \setminus R_E \neq \emptyset$.
Since $\lvert Y \rvert\geq  \lvert E \rvert + \lvert X \rvert + 2k$, any $e \in E \setminus R_E$ can increase her utility by deviating to $\pi'(y^*)$ unless there is some agent in $\pi'(y^*)$ with a negative valuation for $e$. The only such agent is an $F \in \mathcal{F}$ with $e \in F$.
However, by \Cref{eq:aeg_covering}, at least $\frac{3\lvert R_E \rvert}{2}$ elements of $E$ are not covered by $R_{\mathcal{F}}$, so at least $\frac{\lvert R_E \rvert}{2}$ agents in $E$ can IS deviate to $\pi'(y^*)$. Since $\pi'$ is IS, it follows that $\lvert R_E \rvert = 0$, which implies that $R_{\mathcal{F}}$ is a set cover of $E$.
By \Cref{eq:aeg_ineq}, we conclude that $\lvert R_{\mathcal{F}} \rvert \leq \frac{\lvert E \rvert}{3}$ and that $R_{\mathcal{F}}$ is an exact cover of $E$.

For the non-symmetric case, the proof works analogously. The only valuations we need to change are $v_{z_1}(z_2) = 1$ and $v_{z_2}(z_1) = -n$. After the valuation of $z_2$ for $z_1$ is altered to $v'_{z_2}(z_1) = 1$, the same argument as in the symmetric case applies. Note that the altered game is symmetric, so IS and NS are in fact equivalent. In the original game, it is easy to check that $\pi$ is Nash stable.

\end{proof}

\NSsymFEGHard*
We give two separate reductions for FEGs and AFGs that follow a similar idea.
\begin{proof}[Proof of \Cref{thm:fegs_hardness} for FEGs]
    We first prove hardness for \textsc{NS-1-1-Sym-Altered}, and explain afterwards how to slightly adapt the reduction for the non-symmetric case.
    
   The reduction is from Restricted Exact Cover by 3-Sets (RX3C), which is the same as X3C with the addition that each element of $E$ is in exactly three sets contained in the set $\mathcal{F}$.
   Note that $\lvert E \rvert = \lvert \mathcal{F} \rvert =: 3k$ holds in this case. RX3C is known to be NP-complete \citep{gonzalez1985clustering}.
   
   Let $(E, \mathcal{F})$ be an RX3C instance such that $k \geq 6$. Without loss of generality, we may assume that $k$ is even. If $k$ is odd, replace the instance with the disjoint union of two identical copies of the original instance. Define a symmetric FEG $G$ (see \Cref{fig:fegs_hardness}) with agent set $N = E \cup \mathcal{F} \cup T \cup W \cup X \cup Y \cup Z$ with $X = X_1 \cup X_2$ and $Y = Y_1 \cup Y_2$, where $\lvert T \rvert = 2k-1,\lvert W \rvert = k, \lvert X_1 \rvert = \lvert X_2 \rvert = 3k^2, \lvert Y_1 \rvert = \lvert Y_2 \rvert = \frac{k}{2}+1,$ and $\lvert Z \rvert = 3k - 3$.
   The valuations of $G$ are given by 
   \begin{itemize}
       \item $v(F, F') = 1$ if and only if $F \cap F' = \emptyset$ for all $F, F' \in \mathcal{F}$,
       \item $v(t, F) = v(t, t') = 1$ for all $F \in \mathcal{F}$ and $t, t' \in T$,
       \item $v(e, F) = 1$ for all $F \in \mathcal{F}$ and $e \in E \setminus F$,
       \item $v(e, e') =1$ for all $e, e' \in E$,
       \item $v(e, w) = v(w, w') = 1$ for all $e \in E$ and $w, w' \in W$,
       \item $v(e, x_2) =v(e, y_1) = v(e,z) = 1$ for all $e \in E, x_2 \in X_2, y_1\in Y_1, z \in Z$,
       \item $v(w, x_1) = v(w, y) = v(w, z) = 1$ for all $w \in W, x_1 \in X_1, y \in Y, z \in Z$,
       \item $v(a, b) = 1$ for all $a, b \in X \cup Y \cup Z$,
       \item $-1$ for all other valuations.
   \end{itemize}

\begin{figure}[ht]
    \centering
\begin{tikzpicture}[>=stealth,
        every node/.style={font=\small},
        main/.style={circle,draw,minimum size=7mm,font=\small},
        rect/.style={draw, align=center}]

    \node[rect, minimum width = 2.5cm, minimum height = 1.5cm] (Z)  at (2.5, 0) {$Z$};

    \node[rect, minimum width = 1.5cm, minimum height = 1.5cm] (Y1) at (0, 0) {$Y_1$};

    \node[rect, minimum width = 1.5cm, minimum height = 1.5cm] (Y2) at (-2, 0) {$Y_2$};

    \node[rect, minimum width = 3cm, minimum height = 1.5cm] (X1) at (-1.25, -2.5) {$X_1$};

    \node[rect, minimum width = 3cm, minimum height = 1.5cm] (X2) at (2.25, -2.5) {$X_2$};

    \node[rect, minimum width = 3cm, minimum height = 1.5cm] (E) at (2.25, 2.5) {$E$};

    \node[rect, minimum width = 2.5cm, minimum height = 1.5cm] (W) at (-1.5, 2.5) {$W$};

    \node[rect, minimum width = 1.5cm, minimum height = 1.5cm] (T) at (-1, 5) {$T$};

    \node[rect, minimum width = 3cm, minimum height = 1.5cm, dotted, thick] (F) at (2.25, 5) {$\mathcal{F}$};

    \node[main, font = \tiny] (Fe) at (3, 5) {$F \not \ni e$};
    \node[main, font = \footnotesize] (e) at (3, 2.5) {$e$};
    
    \node (y2) at (-1.36, -0.11) {};
    \node (z) at (1.36, -0.11) {};
    \node (Wbl2) at (-2.72, 1.87) {};
    \node (Ebr2) at (3.72, 1.87) {};
    \node (X1tr) at (3.72, -1.87) {};
    \node (X2tl) at (-2.72, -1.87) {};
    
    \begin{scope}[every node/.style={font=\small}]
                \path
                    (Fe) edge[-] node[] {} (e)
                    (Wbl2) edge[-, bend right = 15] node[] {} (X2tl)
                    (Ebr2) edge[-, bend left = 15] node [] {} (X1tr)
                    (y2)  edge[-, bend left = 45] node[] {} (z);
            \end{scope}     

    \draw[-] (-0.25, 5) -- node[] {} (0.75, 5);
    \draw[-] (-0.25, 2.5) -- node[] {} (0.75, 2.5);
    \draw[-] (-1.25, 0) -- node[] {} (-0.75, 0);
    \draw[-] (0.75, 0) -- node[] {} (1.25, 0);
    \draw[-] (0.25, -2.5) -- node[] {} (0.75, -2.5);
    \draw[-] (-2, -0.75) -- node[] {} (-1.25, -1.75);
    \draw[-] (0, -0.75) -- node[] {} (-1.25, -1.75);
    \draw[-] (2.5, -0.75) -- node[] {} (-1.25, -1.75);
    \draw[-] (-2, -0.75) -- node[] {} (2.25, -1.75);
    \draw[-] (0, -0.75) -- node[] {} (2.25, -1.75);
    \draw[-] (2.5, -0.75) -- node[] {} (2.25, -1.75);
    \draw[-] (-2, 0.75) -- node[] {} (-1.5, 1.75);
    \draw[-] (0, 0.75) -- node[] {} (-1.5, 1.75);
    \draw[-] (2.5, 0.75) -- node[] {} (-1.5, 1.75);
    \draw[-] (0, 0.75) -- node[] {} (2.25, 1.75);
    \draw[-] (2.5, 0.75) -- node[] {} (2.25, 1.75);

    \node (Ttl) at (-1.5, 5.5) {};
    \node (Tbl) at (-1.5, 4.5) {};
    \node (Ftr) at (3.5, 5.5) {};
    \node (Fbr) at (3.5, 4.5) {};
    \draw[thick,CoalitionColor, fill=CoalitionColor!50, fill opacity=0.2] 
    \convexpath{Ttl, Ftr, Fbr, Tbl}{.5cm};

    \node (Wtl) at (-2.5, 3) {};
    \node (Etr) at (3.5, 3) {};
    \node (Ebr) at (3.5, 2) {};
    \node (Wbl) at (-2.5, 2) {};
    \draw[thick,CoalitionColor, fill=CoalitionColor!50, fill opacity=0.2] 
    \convexpath{Wtl, Etr, Ebr, Wbl}{.5cm};

    \node (Y2tl) at (-2.5, 0.5) {};
    \node (Ztr) at (3.5, 0.5) {};
    \node (X2br) at (3.5, -3) {};
    \node (X1bl) at (-2.5, -3) {};
    \draw[thick,CoalitionColor, fill=CoalitionColor!50, fill opacity=0.2] 
    \convexpath{Y2tl, Ztr, X2br, X1bl}{.5cm};
\end{tikzpicture}
\caption{\label{fig:fegs_hardness}Illustration of the symmetric FEG $G$ from the reduction in \Cref{thm:fegs_hardness}. Undotted rectangles are cliques, and in the dotted rectangle it holds that $v(F, F') = 1$ if and only if the corresponding sets are disjoint. There is an edge $\{i, j\}$ if and only if $v(i, j) = 1$. Coalitions in $\pi$ are highlighted in blue.}
\end{figure}
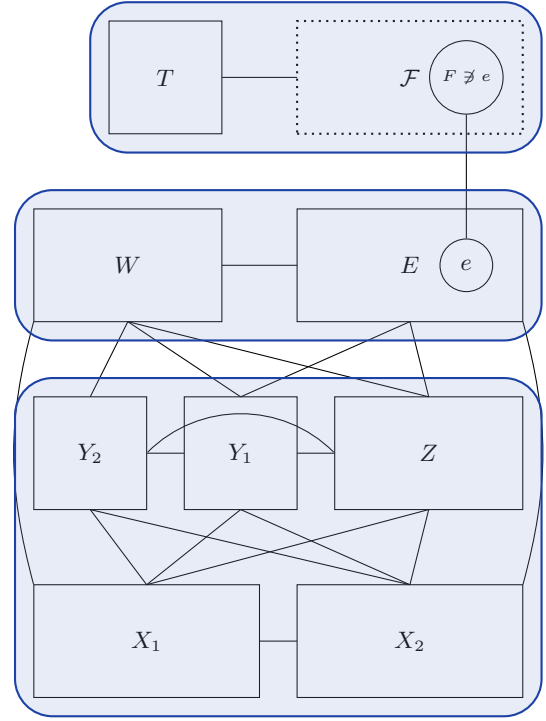
   
Clearly, $G$ is polynomial in the encoding size of $(E, \mathcal{F})$. We claim that $\pi = \{E \cup W, \mathcal{F} \cup T, X \cup Y \cup Z\}$ is NS in $G$. To see this, we perform a case analysis.
\begin{itemize}
    \item For an agent $e \in E$ it holds that $u_e(X \cup Y \cup Z)= \lvert Z \rvert = 3k-3$ and $u_e(\mathcal{F}\cup T) = (3k-6) - (2k- 1) = k - 5$ since exactly $3$ agents in $\mathcal{F}$ have a negative valuation for her. By $u_e(\pi(e)) = 4k - 1$, she cannot deviate.
    \item For an agent $w \in W$ it holds that
    $u_w(\pi(w)) = \lvert W \rvert - 1 + \lvert E \rvert = 4k-1 = \lvert X_1 \rvert - \lvert X_2 \rvert +\lvert Y \rvert + \lvert Z \rvert = u_w(X \cup Y \cup Z) \geq 0$. Thus, this agent does not want to deviate since she has a nonpositive utility for any other coalition in $\pi$.

    \item For an agent $F \in \mathcal{F}$, it holds that $u_F(W \cup E) = -\lvert W \rvert + \lvert E \rvert - 2 \cdot 3 = 2k - 6$.
    Moreover, we have that $u_F(\pi(F)) \geq \lvert T \rvert + (\lvert \mathcal{F} \rvert - 1) - 2 \cdot 6 = 5k - 14 \geq 2k-6 \geq 0$ since $k \geq 6$ and the set $F$ intersects at most $6$ other sets in $\mathcal{F} \setminus \{F\}$ because each set in $\mathcal{F}$ is of size three and each $e \in E$ is contained in exactly three sets in $\mathcal{F}$. Finally, agent $F$ has a negative utility for every other nonempty coalition in $\pi$, so she cannot deviate.  
    
    \item An agent $a \in X \cup Y \cup Z$ cannot deviate to another coalition since $u_a(\pi(a)) \geq 6k^2  \geq N \setminus (\pi'(a))$.

    \item  An agent $t \in T$ cannot deviate to another coalition since she has a negative valuation for any agent $a \in N \setminus \pi(t)$ and $u_t(\pi(t)) =  5k - 2 \geq 0$.
\end{itemize}

Let $\tilde{w}$ and $\tilde{x}_2$ denote agents in $W$ and $X_2$, respectively. Alter the valuation between these agents such that $v'(\tilde{w}, \tilde{x}_2) = 1$ and denote the altered game by $G'$. We claim that an exact cover $\mathcal{C \subseteq F}$ of $E$ exists if and only if $G'$ has an NS partition $\pi'$ with $d(\pi, \pi') \leq 2k$. 

To this end, let $\mathcal{C}$ denote such an exact cover and let $\pi' = \{E \cup \mathcal{C}, (\mathcal{F \setminus C}) \cup T, W \cup X \cup Y \cup Z\}$. Clearly, it holds that $d(\pi, \pi') = 2k$ since $\lvert \mathcal{C}\rvert = \lvert W \rvert = k$. To see that $\pi'$ is NS, we perform a case analysis.
\begin{itemize}
    \item For an agent $e \in E$ it holds that $u_e(\pi'(e)) = \lvert \mathcal{C} \rvert - 1 + \lvert E \rvert - 1 = 4k - 2 > 4k - 3 = \lvert W \rvert + \lvert Z \rvert = u_e(W \cup X \cup Y \cup Z) \geq 0$ since $\mathcal{C}$ is an exact cover of $E$. Furthermore, it holds that $u_e((\mathcal{F \setminus C}) \cup T) = \lvert T \rvert + 2k - 2\cdot2 = 4k - 5$ since $e$ is contained in exactly two sets in $\mathcal{F \setminus C}$. Hence, $e$ cannot deviate since she has a nonpositive utility for any other coalition.
    \item For an agent $C \in \mathcal{C}$ it holds that $u_C(\pi'(C)) = \lvert \mathcal{C} \rvert - 1 + \lvert E \rvert - 2\cdot 3 = 4k-7 =
    \lvert \mathcal{F \setminus C} \rvert - 2 \cdot 3 + \lvert T \rvert \geq u_C((\mathcal{F \setminus C}) \cup T) \geq 0$ since $\mathcal{C}$ is an exact cover, the set $C$ intersects at least 3 sets in $\mathcal{F \setminus C}$, and $\lvert C \cap E \rvert = 3$. Thus, agent $C$ cannot deviate since she has a nonpositive utility for any other coalition in $\pi'$.
    
    \item For an agent $F \in \mathcal{F \setminus C}$ it holds that $u_F(\pi'(F)) \geq \lvert \mathcal{F \setminus C} \rvert - 8 + \lvert T \rvert = 4k - 9 \geq \lvert \mathcal{C} \rvert - 4 + \lvert E \rvert - 6 \geq u_F(\mathcal{C} \cup E) \geq 0$ since any set in $\mathcal{F \setminus C}$ intersects at most $3$ other sets in $(\mathcal{F \setminus C}) \setminus\{F\}$ and at least $2$ sets in $\mathcal{C}$.
    To see this, note that every element of $F$ is contained in some set in $\mathcal{C}$ exactly once, so each element of $F$ is contained in exactly one set in $\mathcal{F \setminus C} \setminus \{F\}$, which are at most 3 sets since $\lvert F \rvert = 3$. Moreover, since $F \not \in \mathcal{C}$, we have that $F$ either intersects two or three sets in $\mathcal{C}$.
    
    Finally, agent $F$ has a negative valuation for every other coalition in $\pi'$, so she cannot deviate.
    
    \item For an agent $w \in W$ it holds that $u_w(\pi'(w)) \geq \lvert W \rvert - 1 + \lvert Y \rvert + \lvert Z \rvert = 5k - 2 \geq 4k = u_w(\mathcal{C} \cup E) \geq 0$. Since this agent has a nonpositive utility for any other coalition in $\pi'$, she cannot deviate.
    \item An agent $a \in X \cup Y \cup Z$ cannot deviate to another coalition since $u_a(\pi'(a)) \geq 6k^2 \geq N \setminus(\pi'(a))$.
    \item An agent $t \in T$ has a nonpositive utility for any other coalition, and it holds that $u_t(\pi(t)) = 4k-2\geq 0$. Hence, she cannot deviate.

\end{itemize}

To prove the reverse direction, suppose that $\pi'$ is NS in $G'$ with $d(\pi, \pi') \leq 2k$. First, note that $\pi$ is not NS in $G$ since agent $\tilde{w}$ can make an NS deviation. 
By $d(\pi, \pi') \leq 2k, X\cup Y \cup Z \in \pi$, and $k \geq 6$, there is some coalition $C \in \pi'$ with $\lvert C \cap (X \cup Y \cup Z) \rvert \geq \lvert X \cup Y \cup Z \rvert - 2k = 6k^2 + 2k - 1$. Since any $a \in X  \cup Y \cup Z$ has a utility of at least $6k^2$ for $C$ and $\lvert N \rvert - 6k^2 < 0$, we have that $a \in C$. Hence, there is some coalition in $\pi'$, denoted by $\pi'(X) \in \pi'$, with $X \cup Y \cup Z \subseteq \pi'(X)$.

For all $a \in \mathcal{F} \cup T$ and $x \in X$, it holds that $v(a, x) = -1$, which implies that \[u_a(\pi'(X)) \leq \lvert N\setminus X\rvert -\lvert X \rvert\leq (13k-2) - 6k^2 < 0,\] so we have that $\pi'(X) \cap (E \cup \mathcal{F} \cup T) = \emptyset$. 
Hence, it must hold that $\tilde{w} \in \pi'(X)$ since \[u_{\tilde{w}}(\pi'(X)) \geq \lvert Y\rvert + \lvert Z \rvert + \lvert X_1 \rvert - (\lvert X_2\rvert -1) + 1 = 4k + 1\]
and $\tilde{w}$ has a positive valuation for at most $\lvert W \rvert - 1 + \lvert E \rvert = 4k-1$ agents outside of $\pi'(X)$.
By an analogous argument, it follows that $W \subseteq \pi'(X)$.
Hence, we have that $u_e(\pi'(X)) \geq \lvert W \rvert + \lvert Z \rvert = 4k - 3$ for all $e \in E$. 

Let $\ell := \lvert \pi'(X) \cap E\rvert$. For $\ell > 0$, we have that for any agent $e \in \pi'(X)$, $u_e(\pi'(e)) \leq (3k - 2) + (k - \ell) = 4k- 2 - \ell$ and $u_e(\pi'(X)) \geq \lvert W \rvert + \lvert Z \rvert + \ell = 4k - 3 + \ell$, so any such agent can deviate. Since $\pi'$ is NS, we have that $\ell = 0$, i.e., $\pi'(X) \cap E = \emptyset$.
Moreover, there is some coalition $\pi'$, denoted by $\pi'(E)$, with $E \subseteq \pi'(E)$, as otherwise some agent in $E$ could deviate to $\pi'(X)$.
Then, since $\pi'$ is Nash stable, $\lvert W \rvert = k$, and $d(\pi, \pi') \leq 2k$, we have that $\lvert \pi'(E)\rvert \leq 4k$ and $u_e(\pi'(E)) \geq 4k-3$ for all $e \in E$. Hence, every $e \in E$ must have a positive valuation for at least $k-1$ many agents in $\pi'(E) \cap \mathcal{F}$.
Hence, every $e \in E$ is contained in exactly one set in $\pi'(E) \cap \mathcal{F}$, meaning that this is an exact cover of $E$.

For the non-symmetric case, the proof works analogously. The only valuation we need to change is $v_{\tilde{x_2}}(\tilde{w}) = -1$. After the valuation of $\tilde{w}$ is altered for $\tilde{x_2}$, the same argument as in the symmetric case applies.
\end{proof}

\begin{proof}[Proof of \Cref{thm:fegs_hardness} for AFGs]
    We first prove hardness for \textsc{NS-1-1-Sym-Altered}, and explain afterwards how to slightly adapt the reduction for the non-symmetric case.
    
    Let $(E, \mathcal{F})$ be an RX3C instance such that $k \geq 8$. Define a symmetric AFG $G$ with agent set $N = E \cup \mathcal{F} \cup T \cup W \cup X \cup Y \cup Z$, where $\lvert T \rvert = 2k - 1, \lvert W \rvert = k - 1, \lvert X \rvert = 8k, \lvert Y \rvert = k - 1$, and $\lvert Z \rvert = 3k - 1$. The valuations of $G$ are given by
    \begin{itemize}
        \item $v(F, F') = n$ if and only if $F \cap F' = \emptyset$ for all $F, F' \in \mathcal{F}$,
        \item $v(t, F) = v(t, t') = n$ for all $F \in \mathcal{F}$ and $t, t' \in T$,
        \item $v(e, F) = n$ for all $F \in \mathcal{F}$ and $e \in E \setminus F$,
        \item $v(e, e') = n$ for all $e, e' \in E$,
        \item $v(e, w) = v(w, w') = n$ for all $e \in E$ and $w, w' \in W$,
        \item $v(e, z) = n$ for all $e \in E$ and $z \in Z$,
        \item $v(w, y) = v(w, z) = n$ for all $w \in W, y \in Y$, and $z \in Z$,
        \item $v(a, b) = n$ for all $a, b \in X \cup Y \cup Z$,
        \item $-1$ for all other valuations.
    \end{itemize}

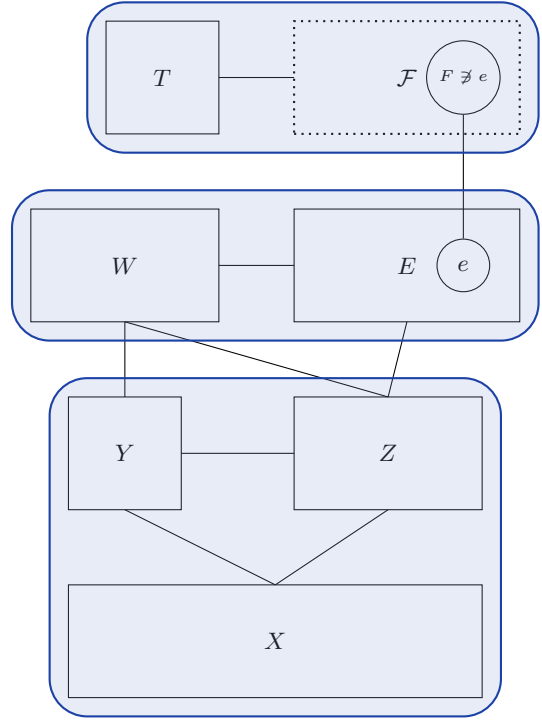
\begin{figure}[h]
    \centering
\begin{tikzpicture}[>=stealth,
        every node/.style={font=\small},
        main/.style={circle,draw,minimum size=7mm,font=\small},
        rect/.style={draw, align=center}]

    \node[rect, minimum width = 2.5cm, minimum height = 1.5cm] (Z)  at (2, 0) {$Z$};

    \node[rect, minimum width = 1.5cm, minimum height = 1.5cm] (Y1) at (-1.5, 0) {$Y$};

    \node[rect, minimum width = 5.5cm, minimum height = 1.5cm] (X1) at (0.5, -2.5) {$X$};

    \node[rect, minimum width = 3cm, minimum height = 1.5cm] (E) at (2.25, 2.5) {$E$};

    \node[rect, minimum width = 2.5cm, minimum height = 1.5cm] (W) at (-1.5, 2.5) {$W$};

    \node[rect, minimum width = 1.5cm, minimum height = 1.5cm] (T) at (-1, 5) {$T$};

    \node[rect, minimum width = 3cm, minimum height = 1.5cm, dotted, thick] (F) at (2.25, 5) {$\mathcal{F}$};

    \node[main, font = \tiny] (Fe) at (3, 5) {$F \not \ni e$};
    \node[main, font = \footnotesize] (e) at (3, 2.5) {$e$};
    
    \node (y2) at (-1.36, -0.11) {};
    \node (z) at (1.36, -0.11) {};
    \node (Wbl2) at (-2.72, 1.87) {};
    \node (Ebr2) at (3.72, 1.87) {};
    \node (X1tr) at (3.22, -1.87) {};
    \node (X2tl) at (-2.22, -1.87) {};
    
    \begin{scope}[every node/.style={font=\small}]
                \path
                    (Fe) edge[-] node[] {} (e);
            \end{scope}     

    \draw[-] (-0.25, 5) -- node[] {} (0.75, 5);
    \draw[-] (-0.25, 2.5) -- node[] {} (0.75, 2.5);
    \draw[-] (-0.75, 0) -- node[] {} (0.75, 0);
    \draw[-] (-1.5, 0.75) -- node[] {} (-1.5, 1.75);
    \draw[-] (2, 0.75) -- node[] {} (-1.5, 1.75);
    \draw[-] (2, 0.75) -- node[] {} (2.25, 1.75);
    \draw[-] (-1.5, -0.75) -- node[] {} (0.5, -1.75);
    \draw[-] (2, -0.75) -- node[] {} (0.5, -1.75);

    \node (Ttl) at (-1.5, 5.5) {};
    \node (Tbl) at (-1.5, 4.5) {};
    \node (Ftr) at (3.5, 5.5) {};
    \node (Fbr) at (3.5, 4.5) {};
    \draw[thick,CoalitionColor, fill=CoalitionColor!50, fill opacity=0.2] 
    \convexpath{Ttl, Ftr, Fbr, Tbl}{.5cm};

    \node (Wtl) at (-2.5, 3) {};
    \node (Etr) at (3.5, 3) {};
    \node (Ebr) at (3.5, 2) {};
    \node (Wbl) at (-2.5, 2) {};
    \draw[thick,CoalitionColor, fill=CoalitionColor!50, fill opacity=0.2] 
    \convexpath{Wtl, Etr, Ebr, Wbl}{.5cm};

    \node (Y2tl) at (-2, 0.5) {};
    \node (Ztr) at (3, 0.5) {};
    \node (X2br) at (3, -3) {};
    \node (X1bl) at (-2, -3) {};
    \draw[thick,CoalitionColor, fill=CoalitionColor!50, fill opacity=0.2] 
    \convexpath{Y2tl, Ztr, X2br, X1bl}{.5cm};
\end{tikzpicture}
\caption{\label{fig:afgs_hardness}Illustration of the symmetric AEG $G$ from the reduction in \Cref{thm:fegs_hardness}. Undotted rectangles are cliques, and in the dotted rectangle it holds that $v(F, F') = n$ if and only if the corresponding sets are disjoint. There is an edge $\{i, j\}$ if and only if $v(i, j) = n$, otherwise the valuation is $-1$. Coalitions in $\pi$ are highlighted in blue.}
\end{figure}

Clearly, $G$ is polynomial in the encoding size of $(E, \mathcal{F})$. We claim that $\pi = \{E \cup W, \mathcal{F} \cup T, X \cup Y \cup Z\}$ is NS in $G$. To see this, we perform a case analysis. Note that in AFGs whenever an agent $a$ has more friends in a coalition $C_1$ than in $C_2$, we have that $u_a(C_1) > u_a(C_2)$.

\begin{itemize}
    \item For an agent $e \in E$, we have that $u_e(\pi(e)) = n( \lvert W \rvert + \lvert E \rvert - 1) = n(4 k - 2) > 3nk > u_e(X \cup Y \cup Z), u_e(T \cup \mathcal{F}) > 0$ since $k \geq 8$. Hence, $e$ cannot deviate.
    \item For an agent $w \in W$, we have that $u_w(\pi(w)) = n(4k - 2) > n(4k - 2) - \lvert X \rvert = u_w(X \cup Y \cup Z)$. Since $w$ has no friends in $T \cup \mathcal{F}$ and $\emptyset$, she cannot deviate. 
    \item For an agent $F \in \mathcal{F}$, we have that $u_F(\pi(F)) \geq n(\lvert T \rvert + \lvert \mathcal{F} \rvert - 7) - 6 = n(5k - 8) - 6 \geq 3nk \geq u_F(W \cup E)$ since $k \geq 8$ and the set $F$ intersects at most 6 other sets in $\mathcal{F} \setminus \{F\}$. Since agent $F$ has no friends in $X \cup Y \cup Z$ and $\emptyset$, she cannot deviate.
    \item An agent $t \in T$ does not have any friends outside her coalition, so she cannot deviate.
    \item An agent $a \in X \cup Y \cup Z$ has $12k - 3$ friends in $\pi(a)$ and at most $\lvert N \setminus \pi(a) \rvert = \lvert E\cup \mathcal{F} \cup W \cup T \rvert = 9k - 2$ friends in any other coalition, so she cannot deviate.
\end{itemize}

Let $\tilde{w}$ be an agent in $W$ and let $\tilde{x}$ be an agent in $X$. Alter the valuation between these two agents, such that $v'(\tilde{w}, \tilde{x}) = n$ and denote the altered game by $G'$. We claim that an exact cover $\mathcal{C \subseteq F}$ of $E$ exists if and only if $G'$ has an NS partition $\pi'$ with $d(\pi, \pi') \leq 2k - 1$. 

To this end, let $\mathcal{C}$ denote such an exact cover and let $\pi' = \{E \cup \mathcal{C}, (\mathcal{F \setminus C}) \cup T, W \cup X \cup Y \cup Z\}$. Clearly, it holds that $d(\pi, \pi') = 2k - 1$ since $\lvert \mathcal{C} \rvert = k$ and $\lvert W \rvert = k - 1$. To see that $\pi'$ is NS in $G'$, we perform a case analysis. 

\begin{itemize}
    \item For an agent $e \in E$, we have that $u_e(\pi'(e)) = n(\lvert E \rvert - 1 + \lvert \mathcal{C} \rvert - 1) - 1= n(4k - 2) - 1 > n(\lvert W \rvert + \lvert Z \rvert) - \lvert Y \rvert - \lvert X \rvert = u_e(W \cup X \cup Y \cup Z) > 0$ since $e$ is contained in exactly one set in the exact cover $\mathcal{C}$. Moreover, we have that $u_e((\mathcal{F \setminus C}) \cup T) \leq n(\lvert \mathcal{F \setminus C} \rvert)  = 2nk$, so $e$ cannot deviate.
    \item For an agent $C \in \mathcal{C}$, we have that $u_C(\pi'(C)) = n(\lvert E \rvert - 3 + \lvert \mathcal{C} \rvert - 1) - 3 = n(4k - 4) - 3$ since $C$ is disjoint to any set in $\mathcal{C} \setminus \{C\}$ and $\lvert E \cap C\rvert = 3$. Moreover, we have that $u_C((\mathcal{F \setminus C}) \cup T) \geq n(\lvert T \rvert + \lvert \mathcal{F \setminus C}\rvert - 3) - 3 = n(4k - 4) - 3$ since $C$ intersects at least $3$ sets in $\mathcal{F \setminus C}$.
    Hence, it holds that $u_C(\pi'(C)) = u_C((\mathcal{F \setminus C}) \cup T)$. $C$ does not deviate to another coalition since she does not have a friend in $W \cup X \cup Y \cup Z$ and $\emptyset$.
    \item For an agent $F \in \mathcal{F \setminus C}$, we have that  $u_F(\pi'(F)) \geq n(\lvert T \rvert - \lvert \mathcal{F} \rvert - 4) - 3 = n(4k - 5) - 3 > n(4k - 5)- 5 = n(\lvert E \rvert - 3 + \lvert \mathcal{C} \rvert - 2) - 5 \geq u_F(\mathcal{C} \cup E)$ since $F$ intersects at most $3$ sets in $(\mathcal{F \setminus C}) \setminus \{F\}$, at least 2 sets in $\mathcal{C}$, and we have that $\lvert F \cap E \rvert = 3$. $F$ does not have a friend in $W \cup X \cup Y \cup Z$ and $\emptyset$, so she cannot deviate.
    \item For an agent $t \in T$, we have that $u_t(\pi'(t)) = n(\lvert T \rvert - 1 + \lvert \mathcal{ F \setminus C}\rvert) = n(4k - 2) > nk > u_t(E \cup \mathcal{C})$. Since $t$ does not have friends in $W \cup X \cup Y \cup Z$ and $\emptyset$, she cannot deviate.
    \item An agent $w \in W$ has $\lvert W \rvert - 1 + \lvert Y \rvert + \lvert Z \rvert = 5k - 2$ many friends in $\pi'(w)$ and at most $\lvert E \rvert = 3k$ many friends in any other coalition, so she cannot deviate.
    \item An agent $a \in X \cup Y \cup Z$ has $13k - 4$ friends in $\pi(a)$ and at most $\lvert N \setminus \pi(a) \rvert = \lvert E\cup \mathcal{F} \cup T \rvert = 8k - 1$ friends in any other coalition, so she cannot deviate.
\end{itemize}

To prove the reverse direction, suppose that $\pi'$ is NS in $G'$ with $d(\pi, \pi') \leq 2k - 1$. First, note that $\pi$ is not NS in $G$ since $\tilde{w}$ can make an NS deviation.

By $d(\pi, \pi') \leq 2k - 1$ and $X \cup Y \cup Z \in \pi$, there is some coalition $C \in \pi'$ with $\lvert C \cap (X \cup Y \cup Z)\rvert \geq \lvert X \cup Y \cup Z \rvert - (2k - 1) = 10k -1$. Since any $a \in X \cup Y \cup Z$ has at least $10k - 1$ many friends in $C$ and at most $\lvert N \rvert - (10k - 1) < \lfloor \frac{N}{2} \rfloor$ many friends in any other coalition, we have that $a \in C$. Hence, there is some coalition in $\pi'$, denoted by $\pi'(X) \in \pi'$, with $X \cup Y \cup Z \subseteq \pi'(X)$.

Therefore, we have that $\tilde{w}$ has at least $\lvert \{\tilde{x}\} \cup Y \cup Z \rvert = 4k - 1$ friends in $\pi'(X)$ and at most $\lvert W \setminus \{\tilde{w}\} \rvert + \lvert E \rvert = 4k - 2$ many friends in $N \setminus \pi'(X)$, so $\tilde{w} \in \pi'(X)$. An analogous argument implies that $W \subseteq \pi'(X)$.

Let $\ell := \lvert \pi'(X) \cap E\rvert$. For any agent $e \in E \setminus \pi'(X)$, we have that $e$ has at most $(3k - 1) + (k - \ell) = 4k - 1 - \ell$ many friends in $\pi'(e)$ and at least $4k - 2 + \ell$ many friends in $\pi'(X)$, so any such agent can make a NS deviation unless $\ell = 0$. Thus, we have that $\pi'(X) \cap E = \emptyset$.

Now, suppose that there are two coalitions $C_1, C_2 \in \pi'$ with $C_1 \neq C_2$ and $\emptyset \neq C_i \cap E \neq E$ for $i \in \{1, 2\}$. By $d(\pi, \pi') \leq 2k - 1, E \cup W \in \pi$, and $\lvert W \rvert = k - 1$, we have that $\lvert C_i \cap E \rvert \leq k$ for some $i \in \{1, 2\}$. Without loss of generality, we may assume that $i = 1$.
Since $C_1 \subseteq E \cup \mathcal{F} \cup T$, any $e \in C_1 \cap E$ has at most $(k - 1) + (3k - 3) = 4k- 3$ friends in $C_1$ because $e$ is contained in exactly $3$ sets in $\mathcal{F}$. However, this is a contradiction to $C_1$ being NS because any $e \in E$ has at least $\lvert W \cup Y \rvert = 4k - 2$ friends in $\pi'(X)$. Hence, we have that there is some coalition in $\pi'$, denoted by $\pi'(E) \in \pi'$, with $E \subseteq \pi'(E)$.

Then, every $e \in E$ must have at least $4k - 2$ friends in $\pi'(E)$ and $\lvert \pi'(E) \rvert \leq 4k$, so every $e \in E$ must have at least $k - 1$ many friends in $\pi'(E) \cap \mathcal{F}$. Hence, every $e \in E$ is contained in exactly one set in $\pi'(E) \cap \mathcal{F}$, meaning that this is an exact cover of $E$.

For the non-symmetric case, the proof works analogously. The only valuation we need to change is $v_{\tilde{x_2}}(\tilde{w}) = -1$. After the valuation of $\tilde{w}$ is altered for $\tilde{x_2}$, the same argument as in the symmetric case applies.
\end{proof} 

\section{Missing Proofs in Section \ref{sec:AverageAnalysis}}

In this appendix, we present the proofs of \Cref{sec:AverageAnalysis} missing from the main body of our paper. 

\existencelimes*
\begin{proof}
    Denote by  \[w_{X,n}^{\mathcal{G}}(m) = \max_{\substack{G^m \in \mathcal{G}^{m+1}(n) \\ \pi_0}}\min_{\pi^m}\frac{\sum_{i=1}^m d(\pi_{i-1}, \pi_{i})}{m}\]
    the maximum average number of changes of coalitions for any sequence of length $m$ for a set of $n$ agents. Since any coalition can be reached from any other coalition with at most $n-1$ changes we get that  $0 \leq w_{X,n}^{\mathcal{G}}(m) \leq n-1$ and thus by Bolzano-Weierstrass there exist a subsequence $(a_{m_t})_{t \in \mathbb{N}}$ that converges to some $a \in \mathbb{R}$.
    
    It holds that every $m = q \cdot y + r$,  $w_{X,n}^{\mathcal{G}}(m) \leq
    \frac{q \cdot w_{X,n}^{\mathcal{G}}(y) + r\cdot n}{m} \leq w_{X,n}^{\mathcal{G}}(y) + \frac{r \cdot n}{m}$ since we can divide the sequence in parts of length $y$ and upper bound the changes for the remaining part by $n$. 
    So let $\epsilon > 0$. Choose $s \in \mathbb{N}$ such that $\frac{1}{s} \leq \epsilon$. Then, by the convergence of the subsequence, there exists an $m_1$ such that $\lvert w_{X,n}^{\mathcal{G}}(m_1) - a \rvert  \leq \frac{1}{2s}$. Now choose $m_2 = 2s \cdot  m_1 \cdot n$ and let $m \geq m_2$.

    Then we get that
    \begin{align*}
    w_{X,n}^{\mathcal{G}}(m) & \leq \frac{(m-m_1) \cdot w_{X,n}^{\mathcal{G}}(m_1) + m_1 \cdot n}{m}\\
    &\leq w_{X,n}^{\mathcal{G}}(m_1) + \frac{m_1 \cdot n}{m}\\    
    & \leq a + \frac{1}{2s} + \frac{m_1 \cdot n}{2s \cdot m_1 \cdot n} \\
    &= a + \frac{1}{s} \leq a + \epsilon.
    \end{align*}
    The first inequality follows from the fact that we can write $m = q \cdot m_1 +r$ for some $r \in \{0\} \cup [m_1 -1]$.
    Therefore, we can divide any sequence of length $m$ into $q$ sequences of length $m_1$ and one sequence of length $r$. Next we upper bound the number of necessary changes for each of the $q$ sequences of size $m_1$ by $m_1 \cdot w_{X,n}^{\mathcal{G}}(m_1)$ and the number of necessary changes for the sequence of length $r$ by $r \cdot n$. Finally, we use the fact that $w_{X,n}^{\mathcal{G}}(m_1) \leq n$ to get the desired inequality. 

    For the other direction choose $m_3 \in \mathbb{N}$ such that 
    $m_3 \geq m \cdot n \cdot 2s$ and $\lvert w_{X,n}^{\mathcal{G}}(m_3) - a \rvert  \leq \frac{1}{2s}$. Such an $m_3$ exists due to $(a_{m_t})_{t \in \mathbb{N}}$.
    Then, we get that $w_{X,n}^{\mathcal{G}}(m_3) \geq a - \frac{1}{2s}$ and that
\begin{align*}
    w_{X,n}^{\mathcal{G}}(m_3) &\leq \frac{(m_3-m) \cdot w_{X,n}^{\mathcal{G}}(m) + m \cdot n}{m_3}\\
    &\leq w_{X,n}^{\mathcal{G}}(m) + \frac{m \cdot n}{m_3} \\
    &\leq w_{X,n}^{\mathcal{G}}(m) + \frac{m \cdot n}{m \cdot n \cdot 2s} \\
    &\leq w_{X,n}^{\mathcal{G}}(m) + \frac{1}{ 2s}. 
\end{align*}
Thus, we can conclude that
\begin{align*}
    a - \epsilon \leq a- \frac{1}{s} \leq  w_{X,n}^{\mathcal{G}}(m_3) - \frac{1}{2s}  &\leq w_{X,n}^{\mathcal{G}}(m).  \\
\end{align*}
This proves that $\lvert a - w_{X,n}^{\mathcal{G}}(m)\rvert \leq \epsilon$.   
\end{proof}

\CISasymmetricPotential*
\begin{proof}
    First, observe that for CIS partitions in FENGs, every CIS deviation of an agent $a$ increases the social welfare by at least $1$. This follows since no agent in the coalition of the deviating agent can have a positive valuation towards her. To be a valid CIS deviation, it also has to hold that every agent in the coalition $a$ deviates to has a non-negative valuation towards her. Finally, agent $a$ increases her utility by at least $1$.

    Furthermore, in FENGs, the update of a valuation between two agents can decrease the social welfare of the current partition by at most $2$.
    Finally, observe that the social welfare of every coalition is upper bounded by $n(n-1)$ and lower bounded by $-n(n-1)$.

Therefore, $m$ deviations can have decreased the social welfare by at most a total of $2m$. Furthermore, in the beginning, the social welfare was at least $-n(n-1)$ and at the end of the sequence it can be at most $n(n-1)$. Thus, it holds that the number of deviations is at most $2n(n-1) + 2m$ and therefore we get that

\begin{align*}
d_{\text{CIS}}^{\text{FENG}}(n) &= \lim_{m\to \infty} \max_{\substack{G^m \in \mathcal{G}^{m+1}(n) \\ \pi_0}}\min_{\pi^m}\frac{\sum_{i=1}^m d(\pi_{i-1}, \pi_{i})}{m} \\&\leq \lim_{m\to \infty} \max_{G^m \in \mathcal{G}^{m+1}(n)}\frac{2n(n-1)+ 2m}{m}
\\&= \lim_{m\to \infty} \frac{2n(n-1)}{m} + \frac{2m}{m}
\\&= 2.
\end{align*}
    
We can further improve this bound to $1$ as follows. Consider the valuation for agent $y$ for another agent $z$, where $y$ and $z$ are fixed.
Each time we decrease $v_y(z)$ to $-1$, at least one increasing update is required to be able to decrease this same valuation again. Furthermore, for every sequence containing an update from $1$ to $0$ to $-1$, we can find a shorter sequence that directly updates $1$ to $-1$. Thus, without loss of generality, we can assume that a sequence maximizing the total distance is alternating between $1$ and $-1$.
Thus, if we alter the valuation of an edge $x$ many times, we can decrease the valuation in total by at most $\lceil \frac{x}{2} \rceil \cdot 2$. Therefore, in the limit, only every second deviation can decrease the social welfare, which improves the bound to $1$.
\end{proof}

\thmsymfegs*
\begin{proof}
Every deviation in a symmetric game increases the social welfare by at least $2$ since the deviating agent $a$ increases her utility by at least one (and due to symmetry, the same holds for at least one other agent in $N \setminus \{a\}$).  
Furthermore, in FENGs, the update of a valuation between two agents can decrease the social welfare of the current partition by at most $4$ (changing a valuation from $1$ to $-1$ decreases the utility at most by two for each of the two corresponding agents). 
Finally, observe that the social welfare of every coalition is upper bounded by $n(n-1)$ and lower bounded by $-n(n-1)$.

Therefore, $m$ deviations can decrease the social welfare by at most a total of $4m$. Furthermore, in the beginning, the social welfare was at least $-n(n-1)$ and at the end of the sequence, it can be at most $n(n-1)$. Thus, it holds that the number of deviations is at most $\frac{2n(n-1) + 4m}{2}$
and therefore we get that

\begin{align*}
d_{\text{X}}^{\text{symFENG}}(n) &= \lim_{m\to \infty} \max_{\substack{G^m \in \mathcal{G}^{m+1}(n) \\ \pi_0}}\min_{\pi^m}\frac{\sum_{i=1}^m d(\pi_{i-1}, \pi_{i})}{m} \\&\leq \lim_{m\to \infty} \max_{G^m \in \mathcal{G}^{m+1}(n)}\frac{2n(n-1)+ 4m}{2m}
\\&= \lim_{m\to \infty} \frac{2n(n-1)}{2m} + \frac{4m}{2m}
\\&= 2.
\end{align*}

Similar to the proof of \Cref{theorem:CIS_FEG_amortized},
we can improve this bound to $1$. Consider a specific valuation between two agents $y$ and $z$.
Each time we decrease $v(y, z)$ to $-1$, at least one increasing update is required to be able to decrease the valuation again. Furthermore, for every sequence containing an update from $1$ to $0$ to $-1$, we can find a shorter sequence that directly updates $1$ to $-1$. Thus, without loss of generality, we can assume that a sequence maximizing the total distance is alternating between $1$ and $-1$.
Thus, if we alter the valuation of an edge $x$ many times, we can decrease the valuation in total by at most $\lceil \frac{x}{2} \rceil \cdot 2$. Therefore, in the limit, only every second deviation can decrease the social welfare, which improves the above bound to $1$.

\end{proof}

\end{document}